\documentclass[sigconf, nonacm]{acmart}

\AtBeginDocument{%
	}

\usepackage{amsmath,amsfonts,amsthm,bm}
\usepackage{pgf}
\usepackage{pgfplots}
\usepgfplotslibrary{fillbetween}

\usepackage{algorithm}
\usepackage{algpseudocode}
\usetikzlibrary{patterns}
\usetikzlibrary{shapes,fit}
\usepackage{subcaption}

\usepackage{wrapfig}

\DeclareMathOperator{\U}{\mathbf{U}}
\DeclareMathOperator{\F}{\mathbf{F}}
\DeclareMathOperator{\G}{\mathbf{G}}

\newcommand{\dist}{\mathfrak{D}}
\newcommand{\logic}{FPL}
\newcommand{\reals}{\mathbb R}
\newcommand{\preals}{\mathbb R_+}

\algnewcommand\algorithmicswitch{\textbf{switch}}
\algnewcommand\algorithmiccase{\textbf{case}}
\algnewcommand\algorithmicassert{\texttt{assert}}
\algnewcommand\Assert[1]{\State \algorithmicassert(#1)}%
\algdef{SE}[SWITCH]{Switch}{EndSwitch}[1]{\algorithmicswitch\ #1\ \algorithmicdo}{\algorithmicend\ \algorithmicswitch}%
\algdef{SE}[CASE]{Case}{EndCase}[1]{\algorithmiccase\ #1}{\algorithmicend\ \algorithmiccase}%
\algtext*{EndSwitch}%
\algtext*{EndCase}%

\newcommand{\para}[1]{\smallskip \textbf{#1}\ }

\newtheorem{remark}{Remark}
\newcommand{\qee}{\hfill$\triangle$}

\definecolor{color1}{HTML}{D55E00}
\definecolor{color2}{HTML}{0072B2}
\definecolor{color3}{HTML}{56B4E9}
\definecolor{color4}{HTML}{E69F00}
\definecolor{color5}{HTML}{E1300A}
\definecolor{color5l}{HTML}{FE947D}

\definecolor{blizzardblue}{rgb}{0.67, 0.9, 0.93}
\definecolor{lightcyan}{rgb}{0.9, 1.0, 1.0}

\providecommand{\todo}[1]{}
\newcommand{\acrshort}[1]{\logic}

\newcommand{\Real}{\mathbb{R}}

\newcommand{\atom}[2]{\bm{\pi_{#1}}^{[0,#2]}}
\newcommand{\fuzzypathset}{\mathcal{I}}

\newcommand{\myspace}{\vspace*{-0.5em}}

\begin{document}
	
\title{Logic of Fuzzy Paths}

\author{Kush Grover}
\email{kgrover@fbk.eu}
\affiliation{%
	\institution{Fondazione Bruno Kessler}
	\streetaddress{}
	\city{Trento}
	\country{Italy}
}

\author{Pratham Gupta}
\email{prathamgupta@iisc.ac.in}
\affiliation{%
	\institution{Indian Institute of Science}
	\city{Bengaluru}
	\country{India}
	\postcode{560012}
}

\author{Jan K\v{r}et\'insk\'y}
\email{jan.kretinsky@tum.de}
\affiliation{%
	\institution{Masaryk University}
	\streetaddress{}
	\city{Brno}
	\country{Czechia}
}

\begin{abstract}
	We introduce a new family of temporal logics intended for specifications in motion planning (MP). 
	It builds upon the signal temporal logic (STL), which is a linear-time logic over real-valued signals that possess quantitative semantics and thus became popular in the areas of cyber-physical systems, robotics, and specifically robot MP. 
	However, in contrast to STL, the proposed logic works with paths as first-class citizens, separating the concerns of geometry and of logic.
	This in turn leads to simpler and more understandable formulae, and a more refined notion of satisfaction being able to reflect also preferences over behaviours.
	Technically, the logic is built on fuzzy, time-varying signal constraints. 
	As a consequence of this expressivity, it is (i) more usable for human-given specifications in MP and (ii) more amenable to learning specifications from demonstrations than other logics. 
	The former is important for the traditional style of verification in robot MP; the latter is becoming recognized as crucial for mining data-given tasks and controller synthesis in human-aware MP.
	We expose the advantages of our proposed logic on examples and show the versatility and flexibility of the framework on a number of scenarios. 
	Finally, we give a learning algorithm with a prototype implementation and discuss the possibilities of model checking and monitoring.	
\end{abstract}

\ccsdesc[300]{Computing methodologies~Robotic planning}
\ccsdesc[500]{Computing methodologies~Learning from demonstrations}
\ccsdesc[300]{Theory of computation~Logic and verification}

\keywords{Temporal Logic, Verification,  Motion Planning, Learning and Control}

\maketitle

\section{Introduction}

In this paper, we propose a new logical framework to cater for some practical needs elegantly and we demonstrate its feasibility by several naive algorithms.
We hope to stimulate interest in exploring this new perspective %
and more efficient algorithms, to achieve real practical applicability.

\emph{Motion planning (MP)} \cite{LaValle2006PlanningA,Latombe1991RobotMP} is the problem of finding a path that satisfies a given specification in a given environment with given dynamics.
While the traditional MP typically considers the specification of the form ``go from A to B without a collision'',
MP with goals specified in the richer language of temporal logics aims at ensuring several objectives be achieved and arranged properly in time.
For instance, linear temporal logic (LTL) \cite{DBLP:conf/focs/Pnueli77} is popular for high-level planning \cite{KressGazit2009TemporalLogicBasedRM,Kloetzer2008AFA,Fainekos2009TemporalLM,Smith2010OptimalPP,DBLP:conf/rss/GroverBTK21}, where exact distances and paths are not as crucial as the qualitative arrangement of tasks.
Typical examples include periodic surveillance $\G\F \mathit{goal}_A \wedge\G\F \mathit{goal}_B$, request-response $\G(\mathit{req}\to\F\mathit{resp})$, or sequencing $\F (\mathit{goal}_1\wedge\F(\mathit{goal}_2\wedge\F\mathit{goal}_3))$, where $\F$ and $\G$ are the temporal operators future (eventually) and globally (always), often also denoted by $\Diamond$ and $\Box$, respectively.

In contrast, specifications such as ``keep a safety distance from an obstacle at least 0.1 meters'' $\G (\mathit{o}\geq0.1)$ require the specification language to work with \emph{numeric} variables, e.g.\ capturing the current \emph{value} of a distance sensor, rather than \emph{Boolean} atomic propositions, e.g.\ capturing whether the robot is currently in the goal area or not.
Such numeric variables and comparison predicates are used in \emph{signal temporal logic (STL)} \cite{DBLP:conf/formats/MalerN04}, which is popular for robot MP that intends to reflect also lower-level requirements.
Moreover, since robots are moving in real-time and changes in the variables' values can occur anytime, the variables in STL take the form of \emph{signals} $\mu:\preals\to\reals$, assigning a real value to each time moment.
Further, the temporal operators can be parametrized by time intervals, limiting their scope.

In this paper, we focus on extending or modifying STL so that it provides \emph{native support} for (1) \textbf{paths} as the basic element in MP and fundamental building blocks of the logic, and for (2) \textbf{preferences} on behaviors rather than pure safety or other functional specifications.

While STL works reasonably well for specifying \emph{hard safety} constraints, we show it is convenient neither for specifying softer constraints on the \emph{desired} behavior, nor for \emph{learning}, which is naturally producing soft generalizations.
To cater for these deficiencies, we introduce paths and \emph{fuzzy} paths as richer analogs of signal constraints and robustness.
As a result, while not aiming at greater expressivity, we obtain a higher-level and more natural language for this purpose, which is (i) easier to use for the purposes above, (ii) more amenable to learning, and (iii) facilitating explainability.

\para{Summary of the technical contribution}
After the intuitive motivation and the justification, we define the logic, exemplify it, and prove that its semantics can be approximated, by providing a naive algorithm for a concrete setting.
The framework leaves lots of degrees of freedom, which (i) can be used to make the logic fit various settings, as we discuss, and (ii) whose fortunate choices may lead to practically efficient algorithms.
Further, we provide a learning algorithm with a prototype implementation. 
We conclude with some insights on comparison to STL and some extensions, and outlining some promising further directions to look into.

\subsection*{Related Work}

\textbf{STL} was first proposed by Maler and Nickovic \cite{DBLP:conf/formats/MalerN04} as an extension to metric interval temporal logic (MITL)~\cite{DBLP:journals/jacm/AlurFH96} with \emph{continuous} signals. 
Over the past decade, much of the attention in motion planning with temporal specifications moved from non-continuity of LTL~\cite{KressGazit2009TemporalLogicBasedRM,Kloetzer2008AFA,Fainekos2009TemporalLM,Smith2010OptimalPP} to the \emph{appealing continuity in STL}, e.g.~\cite{Raman2014ModelPC}, recently focusing a lot on robustness and its maximization~\cite{DBLP:conf/iros/VasileRK17,Lindemann2019ControlBF,DBLP:conf/iros/ScherSK22}.
The various concepts of robustness semantics are introduced already in~\cite{DBLP:journals/tcs/FainekosP09,DBLP:conf/formats/DonzeM10}. 
Besides the mentioned maximization of robustness, many works focus on the case where satisfaction cannot be achieved and aim at a more realistic \emph{minimum-violation planning}~\cite{DBLP:conf/amcc/TumovaCKFR13,DBLP:conf/icra/VasileTKBR17,DBLP:conf/cdc/SchluterSB18,DBLP:conf/amcc/WongpiromsarnSF21,DBLP:conf/itsc/HalderA22}. 
While these works mostly aim at algorithmic solutions within the specification framework of STL, our work, in contrast, treats the issue of finding a better specification framework.

A class of works that deal with preferences are weighted temporal logics \cite{weightedSTL,DBLP:conf/eucc/CardonaV23,DBLP:conf/hybrid/CardonaKV23} but they focus on preferences between different subformulas rather than spatial preferences.
\textbf{Spatial preferences} can be expressed by STL to some extent~\cite{KARLSSON202015537}, but lack the convenient continuity properties discussed in the next section.
There are also dedicated logics for expressing preferences on trajectories, such as spatio-temporal logic (SpaTeL)~\cite{DBLP:conf/hybrid/HaghighiJKBGB15}, which however are suitable for describing complex spatial patterns rather than spatial relations between a trajectory and the environment.
An idea to use a fuzzy set to express spatial preferences appears in \cite{DBLP:conf/birthday/Guesgen14} in the context of spatial logics, with no relation to temporal aspects, temporal logics, or motion planning.
Fuzzy temporal logics have also appeared but with fuzziness applied to the time, not the space, e.g. \cite{5151,10.1145/2629606,9076249}.

\para{Probabilistic predicates}
Gaussian Distribution Temporal Logic (GDTL) and probabilistic STL (PrSTL) incorporate time-dependent predicates and explicitly account for estimation uncertainty via probability thresholds \cite{DBLP:journals/ijrr/LeahyCVJMSB19,sadigh2016safe}.
While sharing the high-level motivation of moving beyond purely Boolean constraints, these formalisms are not designed to express graded spatial preferences over entire trajectories; instead, they reason about probabilistic satisfaction of Boolean predicates.
Moreover, their additional probabilistic layers increase complexity, which can impede explainability and learning.

\section{Motivation}
\label{sec:motivation}

In this section, through a series of examples, we motivate the structure of our logic for preferences over paths.
With \emph{explainability} and \emph{learnability} as primary aspects in mind, the implied desiderata are the following:
(i) while able to capture desired paths, the geometrical and the logical concerns should be separated in order to achieve better understandability, 
(ii) the preferences on behaviours should be naturally expressible, also in the case when they are given by data (i.e., as sample paths), facilitating the learning and the subsequent validation.
While STL can be used in principle to describe preferences over paths, we argue that explicit, path-centric constructs provide a better fit.

\begin{example}
	Consider the following running example. 
	A robot is supposed to move in a wide corridor along the wall at a distance of 2 meters until there is an obstacle in the way.
    Once it spots the obstacle on the horizon of 3 meters, it is supposed to gradually get closer to the side so that it passes the obstacle with the same margin as to the wall.
    The situation is illustrated in Fig.~\ref{fig:avoid_obs}.
	
    Let us assume two signals $o$ and $w$ for the values of the distance to the nearest $o$bstacle and to the $w$all, respectively.
    For instance, the STL formula $$w= 2\ \U^{[0,3]}\ (o=3\ \wedge\ \F^{[3,4]} o=w)$$ constrains the distance to the wall until the obstacle is close, and then soon requires balancing the distances. It can be further conjoined with the previously mentioned $\G o>0.1$ for collision-free safety.
	\label{ex:init}

\begin{figure}[h]
	\centering

	\begin{tikzpicture}[scale=0.9]
		\begin{axis}[xmin=-1,xmax=5,ymin=-1.5,ymax=1.5,samples=50,width=\linewidth, height=0.6\linewidth,
			axis x line=box,
			axis y line =box,
			axis line style={-},
			enlarge x limits=0,
			enlarge y limits=0,
			ticks=none,
			x label style={anchor=west},
			every tick/.style={-},
			]
			\addplot[
			color=red,
			very thick,
			mark=,
			]
			coordinates {
				(-0.2,0)(2,0)(3,1)(4,1)
			};
			\draw[fill=color1] (450, 100) rectangle (550,200) ;
			
		\end{axis}	
	\end{tikzpicture}
	\caption[Running example]{The running example of avoiding an obstacle in a corridor.}
	\label{fig:avoid_obs}
	\myspace\myspace
\end{figure}
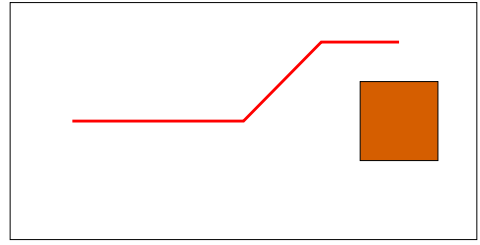

The \emph{gradual} reduction of the distance to the wall from 2 meters to, say for simplicity, 0.5 meters can be written in STL as
\begin{equation}
\begin{split}
\G^{[0,1]}2 \geq w \geq 1.5 \ \wedge\ \G^{[1,2]} 1.5 \geq w \geq 1 \ \wedge\ \\
\G^{[2,3]}1 \geq w \geq 0.5 \ \wedge\ \F^{[3]} w = 0.5 
\end{split} \label{eq:boxes}
\end{equation}
which specifies the respective areas where the robot should be at each time interval, see Fig.~\ref{fig:avoid_obs_bad}.
Clearly, while a finer discretization with a larger number of these ``boxes'' could make the specification more precise, it would be clumsy to write, read, understand, and work with.
Consequently, one can consider an extension of STL containing the logical signal of \emph{time} (as opposed to standard, extra-logical signals, which can get any meaning via an interpretation), allowing for \emph{time-dependent constraints} such as
\begin{equation}
\G^{[0,3]}\ w=2-0.5\cdot t \label{eq:time}
\end{equation} 
directly capturing the intent and avoiding the crookedness, too.
\qee
\end{example}	

\begin{figure}[h]
	\centering
	\begin{tikzpicture}[scale=0.9]
		\begin{axis}[xmin=-1,xmax=5,ymin=-1.5,ymax=1.5,samples=50,width=\linewidth, height=0.6\linewidth,
			axis x line=box,
			axis y line =box,
			axis line style={-},
			enlarge x limits=0,
			enlarge y limits=0,
			ticks=none,
			x label style={anchor=west},
			every tick/.style={-},
			]
		\draw[fill=color1] (450, 100) rectangle (550,200) ;
		\filldraw[fill=blizzardblue] (225,140) rectangle (300,185)	;
		\filldraw[fill=blizzardblue] (300,185) rectangle (355,220)	;
		\filldraw[fill=blizzardblue] (355,220) rectangle (390,250)	;
		\draw[color=red, line width=0.4mm] (75,140) -- (225,140) (300,185) -- (305,185) -- (305,210) -- (355,210) -- (355,225) -- (385,225) -- (385,250);
		\draw[color=red, line width=0.4mm] (225,140) .. controls (245,250) and (260,65) .. (300,185);
	\end{axis}
	\end{tikzpicture}
	\caption{The ``boxing'' approach of specifying continuous changes, available in STL.}
	\label{fig:avoid_obs_bad} \myspace
\end{figure}
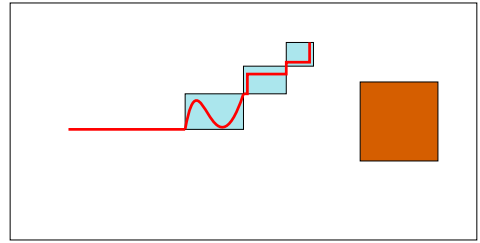

Using predicates that depend on time has been explored before in \cite{Aksaray2021,DBLP:journals/ras/BuyukkocakAY24,sadigh2016safe}.
Yet, whenever the shape is more complex, the formula---even in STL extended with non-linear predicates---inevitably grows and becomes cluttered.
This motivates us to introduce \textbf{paths}, i.e., time-varying constraints, as \emph{atomic building blocks} of our temporal logic.

\begin{example}
	We would like to build formulae such as $(\mathit{middle}\U o=3)\cdot (\mathit{left}\vee\mathit{right})$ describing we move in the middle of the corridor until we come close to an obstacle, and then we go around it from the left or from the right, where $\mathit{middle},\mathit{right},\mathit{left}$ may be very complex paths without obscuring the formula.
    Building on paths as atoms effectively separates the geometric and logical concerns.
\qee
\end{example}

However, just considering the paths is not sufficient for expressing \emph{preferences} over them.
Sometimes it is acceptable to deviate from the desired path by some margin, e.g., when in the corridor, it is acceptable to be a bit farther from the middle line but within safe distance to the wall.
This closeness to the path could work as a measure of preference: the closer to the path, the better.

To express such preferences, we need to be able to express \emph{degrees} of satisfaction with the paths, not just a strict satisfaction/violation as in STL.
This can be achieved by using the robustness semantics of STL \cite{DBLP:journals/tcs/FainekosP09}, which essentially provides a numeric value reflecting the degree of satisfaction/violation.
For an atomic constraint $f(\mu_1,\ldots\mu_n)\geq 0$, evaluating the left-hand side on the signals $\mu_1,\ldots,\mu_n$ yields a real number $r$, taken as the semantics.
In particular, if $r\geq 0$ then the formula is satisfied in the standard Boolean semantics; if $r<0$ it is violated.
The absolute value $|r|$ then gives the ``robustness'': how much the signal (or more precisely, the linear combination) must change to flip satisfaction.

\para{Why not robustness for preferences?} While the degree of satisfaction seems natural for ordering behaviors by preference, it implies several difficulties:

\begin{example}
As one problem in this area, observe that different \emph{scales} of different constraints could imply that the robustness of one completely dominates the others. 
 For instance, if we measure the distance to the wall in millimeters then the robustness semantics of $w-2000\geq0\ \wedge\ \mathit{temp}-18\geq 0$ will completely ignore a huge drop in temperature $\mathit{temp}$ by 10 degrees and it will be dominated by a negligible distance violation as soon as we get too close to the wall by a single centimeter.
 The robustness value of $-20$ thus may mean we are badly freezing or just off by 2 cm.

Hence it may not faithfully report the degree of violation.
While this could be manually addressed by re-scaling the constraints, this requires the domain knowledge and also makes mining preferences from data more difficult \cite{DBLP:conf/tacas/BortolussiGKN22}.
\qee
\label{ex:diff_scales}
\end{example}

\begin{example}
Another problem is posed by the \emph{non-linearity} of the preferences w.r.t. the distance to the threshold.
Furthermore, the measure of preference may even depend on the present context and time; for example, in a narrow passage the margin of deviation should be smaller than in a wide corridor.
\qee
\end{example}

Non-linearities can be captured via ``boxing'' (see Fig.~\ref{fig:avoid_obs_more_boxes}), i.e., geometrically delimiting different areas with the same level of preference.
Such constructs have appeared in probabilistic extensions of STL\footnote{Probabilistic STL (PrSTL) \cite{sadigh2016safe} or Gaussian Distribution Temporal Logic \cite{DBLP:journals/ijrr/LeahyCVJMSB19} can be used to express different margins of acceptance as confidence levels in probabilistic predicates.
	They are used for when there is uncertainty in the state and requires staying in a region with a given probability threshold.
	However, this again runs into the same problems as in Example~\ref{ex:diff_scales} and Example~\ref{ex:human} below, as well as the problem of many boxes.
	Indeed, because of the Boolean semantics of the probabilistic predicates, the issue of many ``boxes'' remains: different confidence parameters are needed for lighter and darker boxes. 
}, capturing smaller probability of more deviating behaviour.
However, to capture the degrees precisely, the complexity of the logical formulae grows, although the requirements are not really more complex.

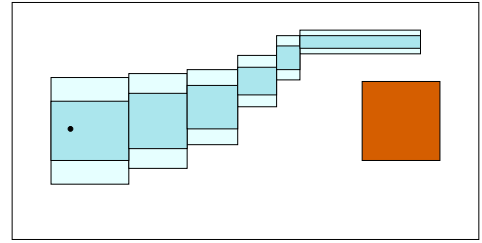
\begin{figure}[h]
	\centering\myspace
	
	\begin{tikzpicture}[scale=0.9]
	\begin{axis}[xmin=-1,xmax=5,ymin=-1.5,ymax=1.5,samples=50,width=\linewidth, height=0.6\linewidth,
	axis x line=box,
	axis y line =box,
	axis line style={-},
	enlarge x limits=0,
	enlarge y limits=0,
	ticks=none,
	x label style={anchor=west},
	every tick/.style={-},
	]
	\draw[fill=color1] (450, 100) rectangle (550,200) ;
	
	\filldraw[fill=lightcyan] (50,70) rectangle (150,205);
	\filldraw[fill=lightcyan] (150,90) rectangle (225,210)	;
	\filldraw[fill=lightcyan] (225,120) rectangle (290,215)	;
	\filldraw[fill=lightcyan] (290,168) rectangle (340,233)	;
	\filldraw[fill=lightcyan] (340,202) rectangle (370,258)	;
	\filldraw[fill=lightcyan] (370,235) rectangle (525,265)	;
	
	\filldraw[fill=blizzardblue] (50,100) rectangle (150,175)	;
	\filldraw[fill=blizzardblue] (150,115) rectangle (225,185)	;
	\filldraw[fill=blizzardblue] (225,140) rectangle (290,195)	;
	\filldraw[fill=blizzardblue] (290,183) rectangle (340,218)	;
	\filldraw[fill=blizzardblue] (340,215) rectangle (370,245)	;
	\filldraw[fill=blizzardblue] (370,242) rectangle (525,258)	;
	\filldraw (75, 140) circle (1pt);
	\end{axis}

	\end{tikzpicture}
	\caption[Preferred areas in STL boxing approach]{More (darker blue) and less (light blue) preferred areas captured by the STL ``boxing'' approach.}
	\label{fig:avoid_obs_more_boxes}\myspace
\end{figure}

{\color{blue}

}

\begin{example}
	Now we revisit the previous example with learning in mind.
	Consider the same avoidance maneuver again and how people behave in this situation, which indicates what the behavior of the robot should ideally be.
	In the narrow part with a 50 cm margin, being off by 10 cm may be perfectly fine, but being off by 40 cm is substantially worse, much less frequent among humans, and bordering on a real danger of collision.

	Further, being off by 40 cm in the wide part may, in contrast, be perfectly fine and considered normal human behavior.
	The robustness semantics, being a single numeric value of the signal combination, is too crude.
	Instead, a more refined notion of ``distance to the desired behavior'' should be considered.\qee
    \label{ex:human}
\end{example}
This motivates us to use \textbf{fuzzy} paths as specifications in our logic, which encapsulate the degrees of our satisfaction with possible paths—our \textbf{preferences}.
In short, instead of many boxes, we can consider two simple parameters---mean and deviation---per time point, yielding the nominal path and the deviation along it\footnote{This notion is akin to Gaussian Processes (GP), where a mean function and a covariance kernel function are used to define distributions over functions.
	However, the GP is primarily a probabilistic model, whereas our fuzzy paths are a way to express preferences over paths.
}.
Because the mean and variance at each time point capture the most relevant part of preferences, we want them to be \emph{native} in the logic.
The advantage is twofold: not only can we express preferences in a convenient way, but we can also learn them very easily:

\para{Are these preferences easy to mine?}
\emph{Learning STL formulas} can be classified into two main categories, template-based \cite{telex} and template-free \cite{grid-tli}. 
Finding the right template, again, requires a domain expert to suggest the templates and template-free approaches often produce formulas with similar ``boxing'' as we have seen in Fig.~\ref{fig:avoid_obs_bad} \cite{DBLP:conf/iros/Linard0LT22,grid-tli,kong2014temporal,bombara2016decision}.
Even more dramatically, in order to capture the experienced frequency (and thus preference) of various areas, for each box in Fig.~\ref{fig:avoid_obs_bad}, there would need to be multiple boxes of different widths, where wider and wider boxes determine the areas with specific lower and lower preference, see Fig.~\ref{fig:avoid_obs_more_boxes}.
This results in even clumsier, larger, and less precise formulas.
Moreover, current tools lack the capability to learn such frequency-based preferences and instead only try to fit the data to a specification. 
For further details and concrete examples of the deficiencies of STL learning, please see Appendix~\ref{app:learn-exper}.

In contrast, \emph{learning formulas of our proposed logic} from observed paths is easier.
Indeed, we can easily compute (i) the mean at each time instant, giving us a nominal, ``mean'' path (see Fig.~\ref{fig:avoid_obs_exper}), and (ii) the deviation at each time instant, giving us a natural, compact way of representing the \emph{frequency}—and thus preference—of farther areas.
Consequently, learning specifications from demonstrations as paths in our proposed logic, e.g., with Gaussian distributions describing their spread, is straightforward.

\begin{figure}[h]
	\centering\myspace
	\begin{tikzpicture}[scale=0.9]
	\begin{axis}[xmin=-1,xmax=5,ymin=-1.5,ymax=1.5,samples=50,width=\linewidth, height=0.6\linewidth,
	axis x line=box,
	axis y line =box,
	axis line style={-},
	enlarge x limits=0,
	enlarge y limits=0,
	ticks=none,
	x label style={anchor=west},
	every tick/.style={-},
	]
	\draw[fill=color1] (450, 100) rectangle (550,200) ;
	
	\draw[color=blue, line width=0.1mm] (75, 175) -- (250, 185) .. controls (300, 185) and (325, 240)  .. (405, 240) -- (500, 245) ;
	\draw[color=blue, line width=0.1mm] (75, 160) -- (250, 155) .. controls (300, 155) and (325, 235)  .. (400, 235) -- (500, 240) ;
	\draw[color=blue, line width=0.1mm] (75, 190) -- (250, 180) .. controls (300, 180) and (325, 220)  .. (400, 220) -- (500, 228) ;
	\draw[color=blue, line width=0.1mm] (75, 120) -- (250, 135) .. controls (300, 135) and (325, 220)  .. (410, 230) -- (500, 235) ;
	\draw[color=blue, line width=0.1mm] (75, 105) -- (260, 115) .. controls (300, 125) and (325, 250)  .. (450, 250) -- (500, 245) ;
	\draw[color=blue, line width=0.1mm] (75, 150) -- (200, 170) .. controls (275, 170) and (330, 210)  .. (400, 245) -- (500, 250) ;
	\draw[color=blue, line width=0.1mm] (75, 125) -- (250, 130) .. controls (300, 130) and (325, 210)  .. (425, 225) -- (500, 225) ;  
	\draw[color=blue, line width=0.1mm] (75, 155) -- (250, 150) .. controls (295, 145) and (325, 210)  .. (400, 255) -- (500, 250) ;  
	\draw[color=blue, line width=0.1mm] (75, 135) -- (250, 130) .. controls (290, 135) and (340, 260)  .. (385, 260) -- (500, 255) ;  
	\draw[color=blue, line width=0.1mm] (75, 165) -- (250, 160) .. controls (300, 150) and (335, 235)  .. (400, 230) -- (500, 230) ;  
	\draw[color=blue, line width=0.1mm] (75, 100) -- (250, 110) .. controls (300, 110) and (325, 235)  .. (400, 220) -- (500, 225);
	\draw[color=red, line width=0.5mm] (75, 140) -- (235, 140) .. controls (290, 140) and (325, 215)  .. (400, 240) -- (500, 240);
	
	\end{axis}	
	\end{tikzpicture}
	\caption{Real, observed paths are drawn in blue, the learned average in red.}
	\label{fig:avoid_obs_exper}\myspace
\end{figure}
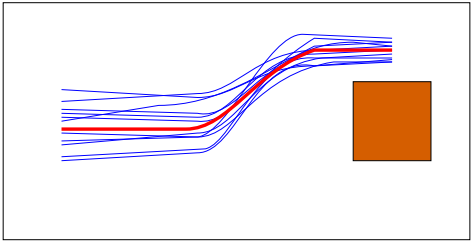

As a result, our logic can also be useful for imitation learning \cite{imitation}, where the task is to learn from demonstration and then operate in similar manner. 
Imitation learning finds applications in autonomous vehicles \cite{imitation-auto-vehicle-1,imitation-auto-vehicle-2,imitation-auto-vehicle-3}, robotics \cite{imitation-robotics-1, imitation-robotics-2}, game-playing \cite{imitation-game-1, imitation-game-2}, healthcare \cite{imitation-healthcare-1}, etc., where behaviors are acquired through observation of human experts. 

And here in the learning setting again, the clear separation of (i) the geometric and continuous-space aspects of the paths from (ii) the logical and temporal relationships of the path segments makes formulas in our formalism easier to understand and thus also \emph{easier to validate} once they are learned.

\section{Preliminaries}

In this paper, $\mathbb{N}$ and $\mathbb{R}$ denote the sets of naturals and reals respectively, $\bar{\mathbb{R}} = \mathbb{R} \cup \{\infty, -\infty\}$, $\mathbb{R}^+$ denotes non-negative reals and $\mathbb{R}^n$ the $n$-dimensional Euclidean space. 
We also use $\|\cdot\|$ to denote the Euclidean norm.
For a matrix $M$, $M^T$ is its transpose.

We assume that the workspace is an $n$-dimensional Euclidean space $E=E_1\times E_2\times\dots\times E_n\subseteq\bar{\reals}^n$.
A \emph{signal} is a continuous function $s_i:[0,\nu] \rightarrow E_i\subseteq \bar{\mathbb{R}}$ where $[0,\nu]$ is its time-horizon. 
We assume that the system has $n$ signals $s_1, s_2, \dots, s_n$ and a \emph{trajectory} with time horizon $[0,\nu]$ is a function $\zeta^{[0,\nu]}:[0,\nu]\rightarrow E$ such that $\zeta^{[0,\nu]}(t)= (s_1(t), s_2(t), \dots, s_n(t))$ for all $t\in[0,\nu]$.
We also use $\zeta$ to denote the trajectory if the time horizon is clear from the context.
A \emph{trajectory} $\zeta^{[0,\nu]}$ shifted by time $\delta$ is defined as $\zeta^{[0,\nu-\delta]}_{\rightarrow \delta}(t) := \zeta^{[0,\nu]}(t+\delta)$ for $t\in [0,\nu-\delta]$.

A function $f:A\rightarrow B$ is called \emph{Lipschitz Continuous} (LC) if there exists a constant $K$ such that $d_B(f(a),f(a'))\leq K\cdot d_A(a,a')$ for all $a,a' \in A$. 
Here, $A$ and $B$ are metric spaces with metrics $d_A$ and $d_B$ respectively, and $K$ is the Lipschitz constant.

An $n$-dimensional continuous probability distribution over a set $E\subseteq\reals^n$ is defined using a \emph{probability density function} $\mathcal{D}: E \rightarrow [0,1]$ such that $\int_{E}\mathcal{D}(\bm{x}) d\bm{x} = 1$.
Its mean vector is denoted by $\bm{\mu}_{\mathcal{D}}$ and the covariance matrix by $\bm{\sigma_{\mathcal{D}}}=(\sigma_{ij})_{n\times n}$.

\section{Syntax and Semantics}
In this section, we will give the formal definition of our fuzzy-path logic (\acrshort{fpl}) and its quantitative semantics, following the intuition and sketches from Introduction.

\subsection{Syntax}
A \emph{point-atom} is described as a probability distribution over the space $\Real^n$, which represents the desired values of a signal at a specific moment in time. 
An \emph{atom} generalizes this notion of a point-atom using the temporal dimension, outlining how signals should evolve over a period of time. 
Intuitively, it is expressed as a distribution that evolves with time. 
For the sake of readability, we restrict our atoms to be distributions characterized by the mean and the covariance matrix that are functions of time. Formally:
\begin{definition}
	An \emph{atom} is defined as a pair $\bm{\bm{\pi}}^{[0,\tau]} = \langle \bm{\mu}, \bm{\sigma}\rangle^{[0,\tau]}$ where $\tau$ is the time horizon, $\bm{\mu}(t)=(\mu_i(t))$ is the mean vector and $\bm{\sigma}(t) = (\sigma_{ij}(t))$ is the covariance matrix at time $t$. 
	Here, $\mu_i$ and $\sigma_{ij}$ are functions over $t\in[0,\tau]$ for each $0\leq i,j\leq n$.
\end{definition}
Using this definition of atoms, it is possible to work with distributions which can be uniquely identified using the mean and the covariance matrix, examples being Gaussian, uniform, or symmetric triangular distributions to name a few. 

We also define the \emph{true} atom $\top^{[0,\tau]}$ as a distribution with mean vector $\bm{0}$, and the covariance matrix being a diagonal matrix with all entries as $\infty$ for all $t\in [0,\tau]$.

Given an atom $\atom{}{\tau}$ and $t\leq\tau$, we use $\atom{}{t}$ to denote its initial part with horizon $[0,t]$.

\begin{figure}[!ht]
	\centering\myspace
	\begin{tikzpicture}[scale=1]
		\begin{axis}[xmin=-1,xmax=5,ymin=-1.5,ymax=1.5,samples=50,width=0.9\linewidth, height=0.52\linewidth,
			axis x line=box,
			axis y line =box,
			axis line style={-},
			enlarge x limits=0,
			enlarge y limits=0,
			xlabel=$x$,ylabel=$y$,
			xtick={0,2,4},
			ytick={-1,0,1},
			x label style={anchor=west},
			every tick/.style={-},
			]
			\addplot[
			color=color2,
			thick,
			mark=*,
			]
			coordinates {
				(0,0)(3,1)
			};
			\addplot[
			color=color3,
			dashed,
			]
			coordinates {
				(0,0.5)(3,1.25)
			};
			\addplot[
			color=color3,
			dashed,
			]
			coordinates {
				(0,-0.5)(3,0.75)
			};

			\draw[color=color3, dashed] (100, 150) ellipse (25 and 50);
			\draw[color=color3, dashed] (400, 250) ellipse (25 and 25);
			\draw[fill=color1] (450, 100) rectangle (550,200) ;
		\end{axis}		
	\end{tikzpicture}
	\myspace\myspace
	\caption[Atom]{The atom from Example~\ref{ex:atom}}
	\label{fig:atom}\myspace \myspace\myspace
\end{figure}
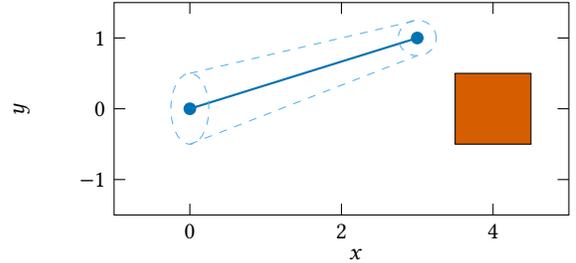

\begin{example}
	Again, for our running example, we use two signals $x$ and $y$, denoting the robot's position.
	Consider the \emph{mean} $\bm{\mu}(t) = 
	\begin{pmatrix}
		t \\ 
		t/3
	\end{pmatrix}$, 
	the \emph{covariance matrix} $\bm{\sigma}(t) = 
	\begin{pmatrix}
		1/16 & 0 \\
		0 & (1/2-t/12)^2
	\end{pmatrix}$,
	and the \emph{atom} $\bm{\bm{\pi_1}}^{[0,3]}(t)=\langle \bm{\mu}, \bm{\sigma}\rangle^{[0,3]}$. 
	This atom is shown in Figure~\ref{fig:atom}. The solid line represents the mean of the atom and the dotted line around it shows the standard deviation.
	\qee
	\label{ex:atom}
\end{example}

\noindent
Now, we have all the elements to define the complete syntax of FPL.

\begin{definition} A formula in \acrshort{fpl} is given by the following syntax:
\begin{equation}
	\notag
	\phi :=  
	\top^{[0,\tau]}\mid
	\bm{\bm{\pi}}^{[0,\tau]}\mid 
	\phi_1\lor\phi_2\mid 
	\phi_1\cdot\phi_2\mid
	\phi_1\U\phi_2
\end{equation}
\end{definition}

\noindent
Where $\bm{\pi}^{[0,\tau]}$ and $\top^{[0,\tau]}$ are atoms as described above, $\lor$ is the Boolean \emph{disjunction}, $\cdot$ is the \emph{concatenation} operator, and $\U$ is the \emph{until} operator.

\begin{example}
	Disjunction can be used, for instance, to capture the possibilities to avoid an obstacle in the corridor either from the left or from the right.
	The until operator intuitively captures switching from one behavior to another at some (a priori not known) time point.
	Besides, one could also introduce a weak until and further operators, see below.
	
	While the concatenation operator is usually not present in temporal logics, it is very natural in the context of motion planning. 
	To illustrate its use, consider two atoms $\atom{1}{3}$ and $\atom{2}{2}$. 
	Their concatenation $\atom{1}{3}\cdot\atom{2}{2}$  specifies we should follow $\atom{1}{3}$ for three time units and then follow $\atom{2}{2}$ for two time-units, thus effectively sequencing fully executed parts of a more complex move.
	
	This also explains why we assume that the time in all atoms starts from 0. 
	It can be achieved by a re-parametrization of time and the use of \emph{true} atom.
	For example, if we want to specify the atom $\langle\bm{\mu}, \bm{\sigma}\rangle^{[2,5]}$, we can convert it to $\bm{\top}^{[0,2]} \cdot \langle\bm{\mu}_{\rightarrow -2}, \bm{\sigma}_{\rightarrow -2}\rangle^{[0,3]}$ where $\bm{\mu}$ and $\bm{\sigma}$ have been shifted by two time units. 
	\qee
\end{example}

We also define $atoms(\phi)$ as the set of all atoms appearing in the syntax of $\phi$ and \emph{minimum atomic horizon} as $h_{min}(\phi) = \min\{\tau\mid \bm{\bm{\pi}}^{[0,\tau]}\in atoms(\phi)\}$.

Note that we do not use the Boolean \emph{negation} or \emph{conjunction} operators here; these are discussed in more detail in section \ref{app:extend_syntax}. 
However, we do allow the use of the \emph{bounded eventually} operator, $\F^{[0,\tau]} \phi := \bm{\top}^{[0,\tau]}\ \U\ \phi$ which can act as a syntactic sugar. 
On the other hand, defining the usual \emph{bounded always} operator is not possible because (i) there is no negation operator, making it syntactically infeasible, and (ii) due to the unusual semantics (defined next) require that an atom be satisfied for the entire given horizon.
We also discuss this more in Section~\ref{app:extend_syntax}.

\subsection{Semantics}

For a \emph{trajectory} $\zeta^{[0,\nu]}$, we monitor the satisfaction of an \logic\ formula using quantitative semantics, i.e. instead of saying that the formula is satisfied or not, we define the extent by which the formula is satisfied.

A \emph{fuzzy path} is any sequence of atoms joined with a concatenation operator, with
concatenation of two fuzzy paths written as $\bm{p}^{[0,v]}=\bm{p_1}^{[0,v_1]} \cdot \bm{p_2}^{[0,v_2]}$ where $v=v_1+v_2$.
For a formula $\phi$, we define the \emph{set of fuzzy paths of $\phi$} as:
\[\mathcal{I}(\phi):=
\begin{cases}
	\{\bm{\pi}^{[0,\tau]}\} & \phi = \bm{\bm{\pi}}^{[0,\tau]} \\
	\mathcal{I}(\phi_1)\cup \mathcal{I}(\phi_2) & \phi_1\lor\phi_2\\[5pt]
	
	\!\begin{aligned}
		\{\bm{p_1}^{[0,v_1]}\cdot \bm{p_2}^{[0,v_2]}\mid\  &\bm{p_1}^{[0,v_1]}\in\mathcal{I}(\phi_1) \\
		\text{ and } &\bm{p_2}^{[0,v_2]} \in \mathcal{I}(\phi_2)\}
	\end{aligned} & \phi = \phi_1\cdot\phi_2 \\[15pt]
	\!\begin{aligned}
		\{\bm{p_1}^{[0,t]} \cdot \bm{p_2}^{[0,v_2]} \mid\  &\bm{p_1}^{[0,v_1]} \in \fuzzypathset(\phi_1),\\
		t\leq v_1 \text{ and } &\bm{p_2}^{[0,v_2]}\in\fuzzypathset(\phi_2)\}
	\end{aligned} & \phi_1 \U \phi_2
\end{cases}
\]
Intuitively, it is the set of all prototypical ways to satisfy $\phi$.
Semantics of the formula $\phi$ evaluated on a trajectory $\zeta$ is defined as the minimum distance to any fuzzy path in this set; intuitively stating how far $\zeta$ is from prototypical satisfaction:
\[\dist(\zeta, \phi) = \min_{\bm{p} \in \mathcal{I}(\phi)}\dist(\zeta, \bm{p})\]

There are several ways to define the distance $\dist(\zeta, \bm{p})$ of a trajectory to a fuzzy path which mainly depend on the application and the particular use case. 
We explain these choices next and also discuss what to use where.

\subsubsection{Distance of a point to a point-atom}
First, we need to define a function $d$ which computes a distance between a point and a \emph{point-atom}.
This distance should have the property that it gradually decreases as we approach (in the Euclidean sense) the center of the point-atom, and vice versa.
Moreover, we want  the distance to reflect the ``frequency'' of the points to, i.e., the more typical the ``closer'' it is. 

As the main example of a metric here, we use the \emph{Mahalanobis distance}, which defines the distance between a point and an $n$-dimensional distribution, intuitively, by how many standard deviations away a point is from the mean.
It also takes into account the correlation between signals.
Formally, for a point $\bm{x}\in \Real^n$ and a distribution $\bm{\mathcal{D}}_{\bm{\mu},\bm{\sigma}}$ with mean vector $\bm{\mu}$ and covariance matrix $\bm{\sigma}$
\[d^2_M(\bm{x}, \bm{\mathcal{D}}_{\bm{\mu},\bm{\sigma}}):=
\sqrt{(\bm{x}-\bm{\mu})^T\bm{\sigma}^{-1}(\bm{x}-\bm{\mu})}\]

We show the evaluation of this metric for our running example on three points.
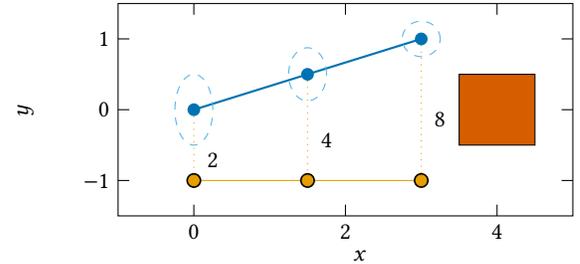
\begin{figure}[!ht]
	\centering\myspace
	\begin{tikzpicture}[scale=1]
		\begin{axis}[xmin=-1,xmax=5,ymin=-1.5,ymax=1.5,samples=50,width=0.9\linewidth, height=0.52\linewidth,
			axis x line=box,
			axis y line =box,
			axis line style={-},
			enlarge x limits=0,
			enlarge y limits=0,
			xlabel=$x$,ylabel=$y$,
			xtick={0,2,4},
			ytick={-1,0,1},
			x label style={anchor=west},
			every tick/.style={-},
			]
			\addplot[
			color=color2,
			thick,
			mark=*,
			]
			coordinates {
				(0,0)(1.5,0.5)(3,1)
			};
			
			\draw[color=color3, dashed] (100, 150) ellipse (25 and 50);
			\draw[color=color3, dashed] (400, 250) ellipse (25 and 25);
			\draw[color=color3, dashed] (250, 200) ellipse (25 and 37.5);
			
			\addplot[
			color=color4,
			mark=*,
			]
			coordinates {
				(0,-1)(1.5,-1)(3,-1)
			};
			\addplot[
			color=color4,
			dotted,
			]
			coordinates {
				(0,0)(0,-1)
			};
			\addplot[
			color=color4,
			dotted,
			]
			coordinates {
				(1.5,0.5)(1.5,-1)
			};
			\addplot[
			color=color4,
			dotted,
			]
			coordinates {
				(3,1)(3,-1)
			};
			\draw[fill=color1] (450, 100) rectangle (550,200) ;
			
			\node[label={[yshift=2ex, xshift=4ex]180:{$2$}},circle,fill,inner sep=2pt] at (axis cs:0,-1) {};
			\node[label={[yshift=4ex, xshift=4ex]180:{$4$}},circle,fill,inner sep=2pt] at (axis cs:1.5,-1) {};
			\node[label={[yshift=6ex, xshift=4ex]180:{$8$}},circle,fill,inner sep=2pt] at (axis cs:3,-1) {};
		\end{axis}
	\end{tikzpicture}
	\myspace\myspace
	\caption[Distance between a distribution and a point]{
		Distance of the points $(0,-1)$, $(1.5,-1)$, and $(3,-1)$ from the point-atoms $\bm{\pi_1}^{[0,3]}(0)$, $\bm{\pi_1}^{[0,3]}(1.5)$, and $\bm{\pi_1}^{[0,3]}(3)$ respectively where $\bm{\pi_1}^{[0,3]}$ is the atom from Example \ref{ex:atom}.
	}
	\label{fig:sem-point-atom}\myspace
\end{figure}

\begin{example}
	We look at how semantics behave for the trajectory $\zeta^{[0,3]}(t) = \begin{pmatrix}
		t \\
		-1
	\end{pmatrix}$ w.r.t. the atom from Example \ref{ex:atom} at $t=0$, $t=1.5$, and $t=3$. Look at Figure \ref{fig:sem-point-atom} for a graphical representation where the distances are indicated next to the dotted lines for all three points.
	\qee
	\label{ex:sem-point-atom}
\end{example}

We can also lift the one-dimensional Mahalanobis to higher dimensions by taking maximum among all dimensions, i.e.:
\[d^{\infty}_M(\bm{x}, \bm{\mathcal{D}}_{\bm{\mu},\bm{\sigma}}) :=\max_i\Big\{\Big|\frac{s_i - \mu_i}{\sigma_{ii}}\Big|\Big\}\]

For other examples of metrics, one can consider measuring in \emph{quantiles}, i.e. ``how many percentiles away from the ideal point are we?''. 
For simplicity of explanation, consider a single signal $x$.
We look at the quantile of its value, i.e.~$F(x)$, where $F$ is the cumulative distribution function of the distribution of the point-atom.
Its distance is then the amount of all closer points to the ideal mean, i.e. twice the difference to the quantile of the mean, which is 0.5.
Hence the distance can be defined as $|2F(x)-1|$.
We can generalize this to higher dimensions as:
\[d_{q}(\bm{x}, \bm{\mathcal{D}}_{\bm{\mu},\bm{\sigma}}):=
Pr(\bm{y}~\big|~ \|\bm{y}-\bm{\mu}\|\leq \|\bm{x}-\bm{\mu}\|; \bm{\mathcal{D}}_{\bm{\mu},\bm{\sigma}})\]
This metric may not be as useful for Gaussian distributions since computing this value might be a bit problematic.
Nevertheless, it would work well for uniform or symmetric triangular distributions. 

The final example is \emph{Euclidean distance}, but not of the resulting combination $f(\mathbf {x})$ as in STL, but of the individual differences of signals.
In other words, we could set the covariance matrix to the identity matrix $I$ and use the $d_1$ metric.
This ignores the distribution aspect, yet captures the difference in the \emph{signal} that is required to satisfy the formula.
In other words, it captures the robustness of signals, not of some linear combination of theirs as in STL robustness.

Note that all our distances describe how much the signals need to change to reach the center in terms of how probable/frequent discrepancies are; in the case of Mahalanobis distance, it is specifically relative to the deviation.
In contrast, the robustness does not speak about the change in the signals, but only about the extent the overall linear function should change; and of course, it does not relativize to the deviation.

This generalization allows us to account for different levels of deviation permitted for different directions/signals.
This might be useful when we cannot observe real paths, but only have a priori information about the separate signals, not including the information about their correlation.
For example, while $w$ and $o$ have some positive deviations in the narrow corridor of width 1 meter, $w+o$ as an observed combination has zero deviation and is identical to 1 here.

\subsubsection{Distance of a trajectory to an atom}
We can extend the definition of distance from point-atoms to {atoms} and there are several ways to do this as well.
We describe the following two:
\begin{align}
	&\dist_{max}(\zeta^{[0,\nu]}, 
	\atom{}{\tau}) := 
	\begin{cases}
		\underset{t\in[0,\tau]}{\max}\ d(\zeta^{[0,\nu]}(t), \atom{}{\tau}(t))& \nu \geq \tau \\
		\infty &o.w.
	\end{cases} \notag\\
	&\dist_{int}(\zeta^{[0,\nu]}, \atom{}{\tau}) := 
	\begin{cases}
		\int_{0}^{\tau}d(\zeta^{[0,\nu]}(t), \atom{}{\tau}(t)) dt & \nu \geq \tau \\
		\infty &o.w.
	\end{cases} \notag
\end{align}

If $\tau>\nu$, $\zeta^{[0,\nu]}(t)$ is not defined for $\nu<t\leq \tau$ and we assign $d(\zeta^{[0,\nu]}(t), \bm{\pi}^{[0,\tau]}(t))=\infty$ there.
The choice of infinity here might look random at first but if we assign it e.g. 0 instead, the trajectory with no time horizon suddenly becomes best since it'll always be at 0 distance. 
To exclude this behavior, we need to define distance as infinity if the atom is defined at some point in time but the trajectory is not.

Intuitively, $\dist_{max}$ captures the ``worst case'' and $\dist_{int}$ is more useful to accumulate bad behavior, or equivalently average it over 
time\footnote{We do not define it to be the average exactly since we are not dividing the integration by the length of the time-interval. 
	However, it is not really a restriction since it is possible to compute the semantics for the whole trajectory and divide by the horizon at the end to get the average value.
	Also, it would become unnecessarily tedious to combine semantics of subformulas if always divided by the length of each interval.}.
Averaging might be less relevant to verification of safety, but more to adhering to nominal (or most desirable) behavior.
Averaging has been advocated, e.g., in  \cite{DBLP:journals/automatica/LindemannD19,mehdipour2019average}.
It is also possible to extend it using some other functions which fits the application better as long as it is computable. 
But for this paper, we will only talk about these two distances and whenever we want to talk about general metric, we will denote it by $\dist$.

\begin{example}
	\label{ex:sem-atom}
	We continue our running example to compute the semantics of the path $\zeta^{[0,3]}(t) = 
	\begin{pmatrix}
		t \\
		-1
	\end{pmatrix}$ w.r.t the atom $\bm{\pi_1}^{[0,3]}$ from Example \ref{ex:atom}. Figure \ref{fig:atom} shows the trajectory in yellow.
	The computed distance is $8$ and $(\approx)4.3$ using maximum and integration semantics respectively (see Appendix~\ref{app:example}).

	\qee
\end{example}

\subsubsection{Distance of a trajectory to a fuzzy path}

Once we have fixed the semantics of atoms, we can extend it to fuzzy paths, i.e., concatenations of finitely many atoms.
The definition here depends on the choice of distance for atoms.
For the choice of $\dist_{max}$ it is better to use the definition:
{\small
\[
\dist_{max} (\zeta^{[0,\nu]}, \atom{}{\tau}\cdot \bm{p}^{[0,\tau']}):=
\begin{cases}
    \max\big\{\dist_{max}(\zeta^{[0,\tau]}, \atom{}{\tau}), \\
    \quad \dist_{max}(\zeta^{[0,\nu-\tau]}_{\rightarrow\tau}, \bm{p}^{[0,\tau']})\big\} & \nu > \tau \\
    \infty & \text{o.w.}
\end{cases}
\]
}

\noindent
since it goes well with the intuitive meaning of the worst case. 
In contrast, for $\dist_{int}$, it is more sensible to take the sum:
\[
	\dist_{int} (\zeta^{[0,\nu]}, \atom{}{\tau}\cdot \bm{p}^{[0,\tau']}):= 
	\begin{cases}
		\dist_{int}(\zeta^{[0,\tau]}, \atom{}{\tau}) + \\
		\dist_{int}(\zeta^{[0,\nu-\tau]}_{\rightarrow\tau},\phi_2) & \nu > \tau \\
		\infty & \text{o.w.}
	\end{cases} \notag
\]

\noindent
For any other metric $\dist$, the semantics of the concatenation operator needs to be defined accordingly.

\begin{remark}
	\label{rem:semantics}
	Observe that we could not simply define semantics in a straightforward inductive way.
	If we define the until operator by  \[\dist_{int}(\zeta^{[0,\nu]}, \phi_1\U \phi_2)= \min_{t} \big\{\dist_{int}(\zeta^{[0,t]}, \phi_1) + \dist_{int}(\zeta^{[0,\nu-t]}_{\rightarrow t}, \phi_2)\big\}\]
	then this expression is not well-defined because the domain of $t$ is not clear. 
	For example if $\phi_1 = \atom{1}{\tau_1}\lor \atom{2}{\tau_2}$ then $t$ can be in $[0,\tau_1]$ or $[0,\tau_2]$ depending on which atom the initial part of $\zeta$ is being matched to. 
	This might not look tedious at the first glance but every disjunction operator introduces another split between the domain of $t$.
\end{remark}

\subsubsection{Left normal form (LNF)}
We say an FPL formula is in left normal form if the left side (i.e.\ the first argument) of all concatenation and until operators are atoms.
We show that any FPL formula can be converted into a normal form using the following lemma.
\begin{lemma} 
	\label{lem:lnf}
	The following equivalences hold for \acrshort{fpl} formulas
	\begin{enumerate}
		\item $(\phi_1\lor\phi_2)\cdot\phi_3 \equiv (\phi_1 \cdot\phi_3) \lor (\phi_2\cdot\phi_3)$
		\item $(\phi_1\cdot\phi_2)\cdot\phi_3 \equiv \phi_1 \cdot (\phi_2\cdot\phi_3)$
		\item $(\phi_1\U\phi_2)\cdot\phi_3 \equiv \phi_1 \U (\phi_2\cdot\phi_3)$
		\item $(\phi_1\lor\phi_2)\U\phi_3 \equiv (\phi_1\U\phi_3) \lor(\phi_2 \U \phi_3)$
		\item $(\phi_1\cdot\phi_2)\U\phi_3 \equiv (\phi_1\U\phi_3) \lor(\phi_1\cdot(\phi_2 \U \phi_3))$
		\item $(\phi_1\U\phi_2)\U\phi_3 \equiv \phi_1\U(\phi_2\U\phi_3)$	
	\end{enumerate}
\end{lemma}

\begin{theorem}
	Any FPL formula can be converted to an equivalent left normal form formula.
\end{theorem}
\begin{proof}
	(Sketch.) We proceed by structural induction and for concatenation or until operators, we use the corresponding identity from Lemma~\ref{lem:lnf} to simplify it.
\end{proof}

\section{Computing semantics}
\label{sec:semantics-assumptions}

We discuss computing exact semantics for our framework in Appendix~\ref{app:exact_semantics}, where we identify that different choices of distance metric, semantics, and classes of functions allowed, would result in solving different optimization problems that require different techniques. 
We also conclude, with the help of an example, why computing the exact value may require solving many optimization problems.   
Therefore, instead of giving a different analysis and algorithm to solve the optimization problems for each class of functions and the chosen distance metric, we suggest to compute an approximation of the semantics with a bounded error.
However, for completely general trajectories, even simpler semantics become too difficult to work with and even undecidable because of the continuous domain \cite{DBLP:conf/stoc/HenzingerKPV95}.
Consequently, to approximate it with reasonable bound on the error, we need some computability assumptions on the atoms and some LC assumptions on trajectories, atoms, and the distance metric $d$ which computes the distance between a point and a point-atom.
In particular, 
\begin{itemize}
	\item each system trajectory $\zeta^{[0,\nu]}(\cdot)$ is LC (with own Lipschitz constant), i.e., $\exists K_{\zeta}: \|\zeta^{[0,\nu]}(t)-\zeta^{[0,\nu]}(t')\|\leq K_{\zeta}|t-t'|$ for all $t,t'\in[0,\nu]$
	\item distance $d(\cdot, \atom{}{\tau})$ to each point-atom is LC in space, i.e., $\exists K_{d,\zeta}:|d(\bm{x},\bm{\pi})-d(\bm{y},\bm{\pi})| \leq K_{d,\zeta} \| \bm{x} - \bm{y} \|$ for all $\bm{x},\bm{y} \in E$ and $\bm{\pi}\in atoms(\phi)$
	\item distance $d(\bm{x}, \atom{}{\tau}(\cdot))$ to each atom $\atom{}{\tau}$ is LC in time, i.e., $\exists K_{d,\bm{\pi}}\forall \bm{\pi}\in atoms(\phi):|d(\bm{x}, \bm{\pi}(t)) - d(\bm{x}, \bm{\pi}(t'))|\leq K_{d,\bm{\pi}}|t-t'|$ for $t,t'\in[0,\tau]$ and all $\bm{x}$
	\item $\dist(\zeta^{[0,\nu]}, \atom{}{\tau})$ is computable for all $\atom{}{\tau}\in atoms(\phi)$.
\end{itemize}
The first assumption here restricts the system but this is a standard assumption when dealing with systems over continuous domains~\cite{lipschitz-assumption-1, lipschitz-assumption-2}. 
This is required to exclude the behaviors where system behaves unpredictably which make things difficult to compute.

The second assumption essentially says that the distance to a point-atom changes reasonably as one goes closer or farther from its mean.
The third assumption is similar to the second one but instead of dealing with space, it deals with time. 
It states that, as the time passes, the distance of a point to the "moving" distribution changes reasonably.
For Mahalanobis distance, these assumptions are satisfied if all atoms are LC, i.e., each entry in the mean vector $\bm{\mu}(t)$ and the covariance matrix $\bm{\sigma}(t)$ is LC.

The last assumption here deals with computability of semantics for atoms.
It is satisfied if, for example, the atoms are piecewise linear and the trajectories are also linear combined with the Mahalanobis distances $d_M^2$ or $d_M^\infty$.
For more complicated trajectories and atoms, this assumption can be relaxed to ``approximable with a known bound on the error''.
This would, in turn, increase the error while approximating the semantics.

\subsection{Approximating Algorithm}
\label{sec:alg}
We give a recursive procedure to approximate the semantics for FPL formulas in left normal form with a bound on the error.
The pseudo-code is given in Algorithm~\ref{alg:main}.
It starts by looking at the outermost operator and computes the semantics for the sub-formulas and returns the value by combining them using the appropriate definition of $\dist$.
We use $\tilde{\dist}$ to denote the computed semantics using this procedure.
If the outermost operator is \emph{disjunction} then it is computed as 
\[\tilde{\dist}(\zeta^{[0,\nu]}, \phi_1 \lor \phi_2) = \min\{\tilde{\dist}(\zeta^{[0,\nu]}, \phi_1), \tilde{\dist}(\zeta^{[0,\nu]}, \phi_2)\}
\]
and for the \emph{concatenation} operator, we use
\[
\tilde{\dist}_{max}(\zeta^{[0,\nu]}, \atom{}{\tau} \cdot \phi_2) = 
\begin{cases}
    \max\big\{\tilde{\dist}_{max}(\zeta^{[0,\tau]}, \atom{}{\tau}), \\
    \quad \tilde{\dist}_{max}(\zeta^{[0,\nu-\tau]}_{\rightarrow \tau}, \phi_2)\big\} 
    & \nu \geq \tau \\
    \infty & \text{o.w.}
\end{cases}
\]
\[
\tilde{\dist}_{int}(\zeta^{[0,\nu]}, \atom{}{\tau} \cdot \phi_2) = 
\begin{cases}
    \tilde{\dist}_{int}(\zeta^{[0,\tau]}, \atom{}{\tau})\ + \\
    \quad \tilde{\dist}_{int}(\zeta^{[0,\nu-\tau]}_{\rightarrow \tau}, \phi_2)
     & \nu \geq \tau \\
    \infty & \text{o.w.}
\end{cases}
\]
Finally, for the until formula, we discretize the time interval $[0,\tau]$ into small steps of size $\delta$ and use these points to compute the semantics as follows
\[\tilde{\dist}(\zeta^{[0,\nu]}, \bm{\pi_1}^{[0,\tau]} \U \phi_2) = \min_{t\in S}\tilde{\dist}(\zeta, \bm{\pi_1}^{[0,t]} \cdot \phi_2)\]
where $S = \{m\delta \mid m\in \mathbb{N},\ m\delta \leq \tau\} \cup \{\tau\}$ 
is the discretization.

\begin{algorithm}[t!]
	\caption{Algorithm to compute semantics}
	\small
	\label{alg:main}
	\begin{algorithmic}[1]
		\Procedure{ComputeSemantics}{$\zeta^{[0,\nu]}, \phi$}
		\Switch{$\phi$}
			\Case{$\atom{}{\tau}$}
				\If{$\tau > \nu$} \Return $\infty$
				\Else 
					\State\Return \Call{SemanticsOfAtoms}{$\zeta^{[0,\nu]}, \atom{}{\tau}$}
				\EndIf
			\EndCase
			\Case{$\phi_1 \lor \phi_2$}
				\State $x_1\gets$ \Call{ComputeSemantics}{$\zeta^{[0,\nu]}, \phi_1$}
				\State $x_2\gets$ \Call{ComputeSemantics}{$\zeta^{[0,\nu]}, \phi_2$}
				\State \Return $\min\{x_1,x_2\}$
			\EndCase
			\Case{$\atom{}{\tau}~\cdot~\phi_2$}
				\If{$\tau > \nu$} \Return $\infty$
				\Else 
					\State $x_1\gets$ \Call{ComputeSemantics}{$\zeta^{[0,\tau]}, \atom{}{\tau}$}
					\State $x_2\gets$ \Call{ComputeSemantics}{$\zeta^{[0,\nu-\tau]}_{\rightarrow \tau}, \phi_2$}
					\State \Return $x_1 + x_2$ or $\max(x_1,x_2)$
				\EndIf
			\EndCase
			\Case{$\atom{}{\tau}~\U~\phi_2$}
				\State $S \gets$ \Call{Discretize}{$[0,\tau]$}
				\For{$\tau'\in S$}
					\If{$\tau' > \nu$} \Return $\infty$
					\Else 
						\State $x_1^{\tau'}\gets$ \Call{ComputeSemantics}{$\zeta^{[0,\tau']}, \atom{}{\tau'}$}
						\State $x_2^{\tau'}\gets$ \Call{ComputeSemantics}{$\zeta^{[0,\nu-\tau']}_{\rightarrow \tau'}, \phi_2$}
					\EndIf
				\EndFor
				\State \Return $\min_{\tau'} (x_1^{\tau'} + x_2^{\tau'})$ or $\min_{\tau'}\max \{x_1^{\tau'}, x_2^{\tau'}\}$
			\EndCase
		\EndSwitch
		\EndProcedure
	\end{algorithmic}
\end{algorithm}
\myspace\myspace

\begin{example}
	\label{ex:until}
	We extend our previous example with an until operator and consider
	$\phi = \bm{\pi_1}^{[0,4]}\ \U\ \bm{\pi_2}^{[0,2]}$
	where 
	$\bm{\pi_1}$ is the same atom from Example \ref{ex:atom} except that the time horizon is increased to 4 units.
	and $\bm{\pi_2}^{[0,2]}= \langle\bm{\mu_2}, \bm{\sigma_2}\rangle^{[0,2]}$ with
	\[\bm{\mu_2}(t) = 
	\begin{pmatrix}
		t+3 \\ 
		1
	\end{pmatrix} 
	\text{, \ \ } \bm{\sigma_2}(t) = 
	\begin{pmatrix}
		1/16 & 0 \\
		0 & 1/16
	\end{pmatrix}.\]
Figure~\ref{fig:until} shows both of the atoms.	
To compute the semantics of the trajectory $\zeta^{[0,5]}(t) = \begin{pmatrix}
	t \\
	-1
\end{pmatrix}$ for $\phi$, the algorithm will approximate the semantics $\dist(\zeta^{[0,5]},\phi)$ by discretizing the interval $[0,4]$ into steps of size $\delta$ (taken to be 0.5 here), as shown in Figure \ref{fig:until} using ticks, and for each discretization point, computing the semantics by assuming that the until switches there. 
Then the algorithm selects the point which gives the minimum value.
\qee

\end{example}

\begin{figure}[!ht]
	\centering\myspace
	\begin{tikzpicture}[scale=1]
		\begin{axis}[xmin=-1,xmax=6,ymin=-1.5,ymax=1.5,samples=50,width=0.9\linewidth, height=0.52\linewidth,
			axis x line=box,
			axis y line =box,
			axis line style={-},
			enlarge x limits=0,
			enlarge y limits=0,
			xlabel=$x$,ylabel=$y$,
			xtick={0,2,4},
			ytick={-1,0,1},
			x label style={anchor=west},
			every tick/.style={-},
			]
			\addplot[
			color=color2,
			thick,
			mark=*,
			]
			coordinates {
				(0,0)(4,1.33)
			};
			
			\addplot[
			color=color3,
			dashed,
			]
			coordinates {
				(0,0.5)(4,1.5)
			};
			\addplot[
			color=color3,
			dashed,
			]
			coordinates {
				(0,-0.5)(4,1.16)
			};
			
			\draw[color=color3, dashed] (100, 150) ellipse (25 and 50);
			\draw[color=color3, dashed] (500, 283) ellipse (25 and 17);

			\addplot[
			color=color4,
			thick,
			mark=*,
			]
			coordinates {
				(3,1)(5,1)
			};
			\addplot[
			color=color4,
			dashed,
			]
			coordinates {
				(3,1.25)(5,1.25)
			};
			\addplot[
			color=color4,
			dashed,
			]
			coordinates {
				(3,0.75)(5,0.75)
			};
			\draw[color=color5l, dashed] (400, 250) ellipse (25 and 25);
			\draw[color=color5l, dashed] (600, 250) ellipse (25 and 25);
			
			\draw[fill=color1] (450, 100) rectangle (550,200) ;
			
			\addplot[
			color=color2,
			thick,
			mark=|,
			]
			coordinates {
				(0,0)(0.5,0.165)(1,0.33)(1.5,0.5)(2,0.67)(2.5,0.835)(3,1)(3.5,1.165)(4,1.33)
			};
			
			\node[label={[yshift=0em, xshift=10em]180:{$\bm{\pi_1}$}},circle,fill,inner sep=2pt] at (axis cs:0,0) {};
			\node[label={[yshift=3em, xshift=24em]180:{$\bm{\pi_2}$}},circle,fill,inner sep=2pt] at (axis cs:0,0) {};
		\end{axis}	
	\end{tikzpicture}\myspace\myspace
	\caption{The atoms $\bm{\pi_1}$ and $\bm{\pi_2}$ from Example \ref{ex:until}}
	\label{fig:until}\myspace\myspace
\end{figure}
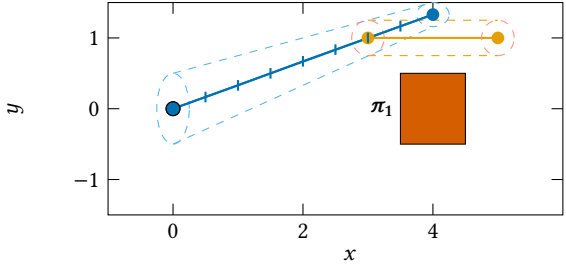

To show the correctness of this algorithm, we need to prove that it can approximate the actual value with a bound on the error.
We have different bounds for the different semantics.

\begin{theorem}
	\label{thm:main-max}
	For an LC trajectory $\zeta^{[0,\nu]}$, an FPL formula $\phi^{[0,v]}$ in LNF, $\exists C,C_2 \in \reals^+$ such that for $0<\delta <h_{min}(\phi)$, 
	\[|\tilde{\dist}_{max}(\zeta^{[0,\nu]}, \phi) - \dist_{max}(\zeta^{[0,\nu]}, \phi)| \leq k(C_2+C)\delta\] 
	where $k$ is the number of until operators in $\phi$.
\end{theorem}

\begin{theorem}
	\label{thm:main-int}
	For an LC trajectory $\zeta^{[0,\nu]}$, an FPL formula $\phi^{[0,v]}$ in LNF, $\exists C_1,C\in\reals^+$ such that for $0<\delta <h_{min}(\phi)$,
	\[|\tilde{\dist}_{int}(\zeta^{[0,\nu]}, \phi) - \dist_{int}(\zeta^{[0,\nu]}, \phi)| \leq k(C_1+kC\nu)\delta\] 
	where $k$ is the number of concatenation or until operators in $\phi$.
\end{theorem}
 	The complete proof is given in Appendix~\ref{app:proof} where we also show that the following values of the constants work
 	\[C_1 = \max_{\bm{\pi}\in atoms(\phi)} \big(d(\zeta(0), \atom{}{\tau}(0)) + (K_{d,\bm{\pi}} + K_\zeta K_{d,\zeta})\cdot\min\{\nu,\tau\}\big)\] \[C_2=\max_{\bm{\pi}\in atoms(\phi)} \big(K_{d,\zeta}K_{\zeta}
 	+ K_{d,\pi}\big) \text{ and } C=K_{d,\zeta} K_{\zeta}\]
	As one would expect, the constants depend on the Lipschitz constants identified in the assumptions and for integration semantics, it also depends on the distance between the initial points of the trajectory and atoms since that could also affect the error accumulated.
	Intuitively, our algorithm and its proof work by "matching" all points of the trajectory to a point close to the optimal matching point and use LC assumptions to give a bound on the error.
	However, one needs to be careful here, since even small changes in the matching could result in a completely different value. 
	For example, consider the formula $\phi=\atom{1}{4}\U\atom{2}{2}$ from Example~\ref{ex:until} and the trajectory (with horizon $[0,5]$) which follows $\bm{\pi_1}$ for 3 time-units and $\bm{\pi_2}$ for the next 2 time-units.
	The distance for this example is $0$ since it follows exactly the mean path by switching the until at time $3$. 
	However, while computing the semantics of $\phi$ if we select any switching time $>3$, the distance becomes $\infty$ because the trajectory is not defined for the last part of $\bm{\pi_2}$.
	Fortunately, this is not a problem for our algorithm since the discretization point just \emph{before} the optimal one is good enough to approximate the value.

\section{Learning}
\label{sec:learning}

In this section, we present a method for inferring FPL formulas from a set of positive trajectories (demonstrations).
The high-level idea is to learn a tree/DAG structure that represents the logical part of the formula, and at each node learn an atom that captures the geometric behavior over a time interval, profiting from the separation of concerns.
The pseudocode is presented in Algorithm~\ref{alg:learning}.
The input is a set of trajectories that describe the intended behavior and some hyperparameters described later.

\begin{algorithm}
	\caption{Fuzzy Path Logic Learning Algorithm}
	\label{alg:learning}
	\begin{algorithmic}[1]
		\State \textbf{Input:} Dataset $D$, Initial trajectory length $k$
		\State \textbf{Output:} Fuzzy Path Logic $FPL$
		\State $T \leftarrow \textsc{LearnAtomDAG}(D, k)$
		\State $M \leftarrow \textsc{MergeNodes}(T)$
		\State $FPL \leftarrow \textsc{SimplifyDAG}(M)$
		\State \textbf{return} $FPL$
		\Procedure{LearnAtomDAG}{$D$, $k$}
		\State $T \gets$ Empty Tree
		\State $C \gets$ \textsc{ClusterTrajectories}($D$, $k$)
		\For{each cluster $c \in C$}
		\State $A \gets \textsc{LearnAtom}(c)$
		\State $R \gets$ \textsc{Regression}(A)
		\State $D' , dTime \gets$ \textsc{RemoveAgreement}($c$, R)
		\State $DT \gets$ \textsc{LearnAtomDAG}($D'$, $k$)
		\State $T \gets$ \textsc{AddSubtree}($T$, DT, A, dTime)
		\EndFor
		\State \Return $T$
		\EndProcedure
	\end{algorithmic}
\end{algorithm}

The algorithm starts with the dataset $D$ and, based on the initial behavior (of length $k$), clusters the trajectories.
For each cluster, a representative atom is learned on the initial segment of the trajectories and extended via regression to cover the maximum length within the cluster.
For each trajectory, we identify the time step at which it starts to deviate from the learned atom.
The last deviation time across all trajectories in this cluster defines the total time horizon of the learned atom ($dTime$), while the range of deviation times determines the time interval of the \textit{Until} operator.

A new dataset is then constructed from the trailing subtrajectories beginning at each deviation point, and the algorithm is invoked recursively to learn the subsequent part of the formula.
Each recursive call returns a DAG, which is connected beneath the previously learned part of the structure using an until operator.
This process results in a DAG of FPL atoms connected by \textit{Until} operators for each cluster.
Next, similar atoms in the DAG are merged based on a similarity score, while ensuring that no merged pair has an ancestor–descendant relationship.

This resulting DAG is equivalent to an FPL formula that we construct by simplifying the DAG through a combination of the following operations:
\begin{itemize}
\item \textbf{Extend:} Combine two nodes connected by an \textit{Until} operator into a single node.
\item \textbf{Join:} Replace two nodes with identical parents and children by their disjunction.
\item \textbf{Split:} Duplicate a node and distribute its children among the two copies.
\end{itemize}

Additional details of this learning algorithm regarding the merging of the atoms and the simplification procedure are provided in Appendix~\ref{app:learning}.
Due to space restrictions, we show the run of the algorithm only on one example here, whose dataset is shown in Figure~\ref{fig:learning_example_dag2d}.
For more examples that also involve learning the until operator, we refer to Appendix~\ref{appendix:learning-examples}.

\begin{example}
    Consider the initial dataset $D$ of trajectories shown in Figure~\ref{fig:learning_example_traj2d}. 
    The algorithm splits the set $D$ into two clusters based on their initial behavior. 
    We obtain a learned DAG as shown in Figure~\ref{fig:learning_example_dag2d} which captures the bifurcation in the trajectories around the obstacle.
    The learned atoms are visualized in Figure~\ref{fig:learning_example_fpl_atoms_2d}.
    For this example, there are no similar atoms to be merged.
    The SimplifyDAG procedure will compute the final FPL formula as the disjunction of the two atoms.
    \qee
\end{example}

\begin{figure}[ht]
    \centering
    \includegraphics[width=\columnwidth]{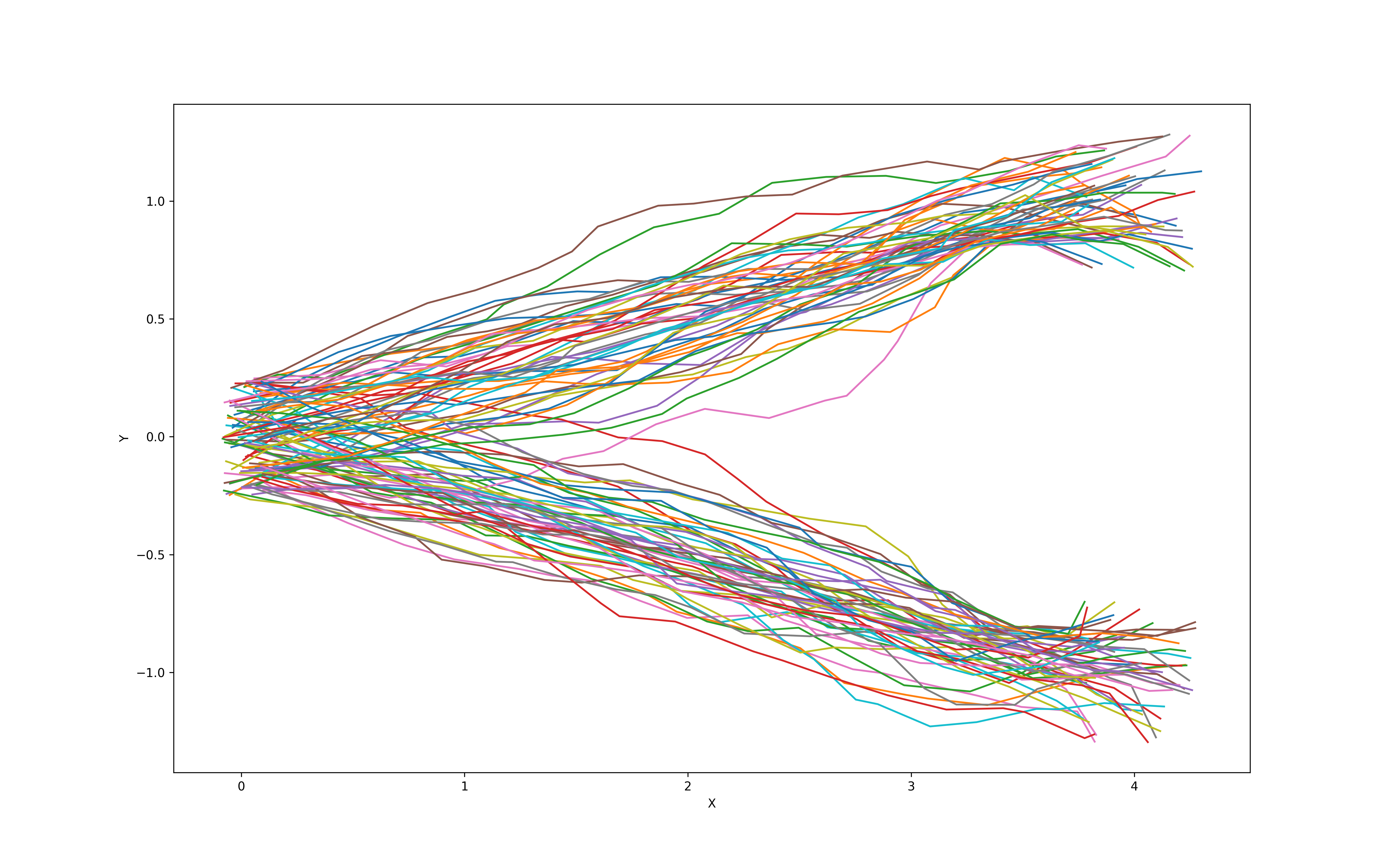}\myspace
    \caption{Initial dataset of trajectories for learning example.}
    \label{fig:learning_example_traj2d} 
\end{figure}

\begin{figure}[ht]
    \centering
    \includegraphics[width=\columnwidth]{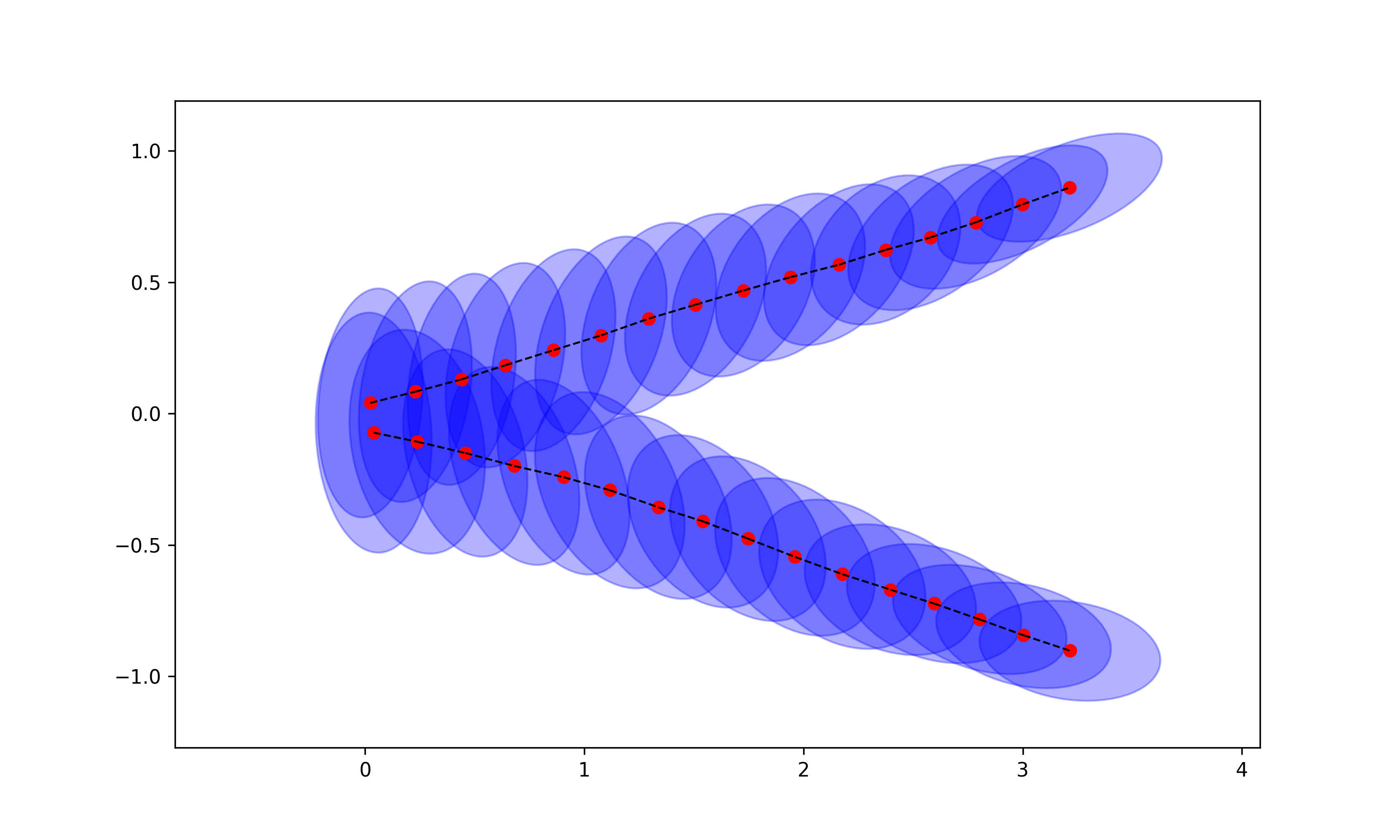}\myspace
    \caption{Learned atoms from the trajectories in Figure~\ref{fig:learning_example_traj2d}.}
    \label{fig:learning_example_fpl_atoms_2d} 
\end{figure}

\begin{figure}[ht]
    \centering
    \includegraphics[width=0.8\columnwidth]{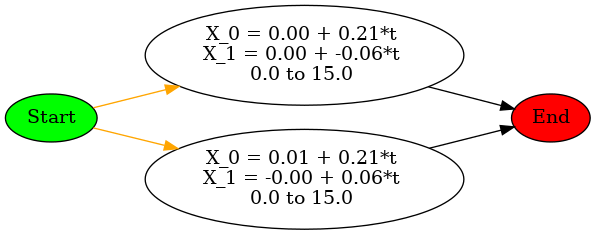}
    \caption{The learned DAG from the trajectories in Figure~\ref{fig:learning_example_traj2d}.}
    \label{fig:learning_example_dag2d}
\end{figure}

\section{Discussion}

\subsection{Advantages over other logics}
\label{app:stl_comparison}
Here we compare FPL with STL and its extensions, discussing why FPL should be a separate logic rather than just an extension of STL. 
First, we argue that replicating the behavior of FPL in STL is possible but comes with significant drawbacks in terms of explainability and learnability.
Then, we discuss these two aspects in more detail.

If we want to express an atom using STL, we need to describe the mean path and the deviation around it. 
To express the mean path with $n$ signals, STL would require $n+2$ predicates describing the half-spaces and hyperplanes whose intersection yields the required object.
For example, the atom $\atom{1}{3}$ from Example~\ref{ex:atom} can be expressed in STL using the predicates $p_1:x-t=0$, $p_2:3y-t=0$, $p_3:t\geq0$, and $p_4:t\leq 3$ to describe the mean path.

Moreover, in FPL, the variance around the mean path is explicitly represented and can be easily visualized.
In order to express the same distance from the mean path in STL, the coefficients of the predicates need to be changed accordingly to reflect the distance we use in FPL.
For our example, the point $(x,y,t)=(0,-1,0)$ has robustness $-3$ with respect to $p_1\land p_2\land p_3 \land p_4$, but if we want to reflect the Mahalanobis distance of $2$ (as computed in Example~\ref{ex:sem-point-atom}) here, we would need to change the coefficients of $p_2$ by multiplying them by $\frac{2}{3}$. 
However, since $\sigma$ changes with time in FPL, the coefficients need to be a function of time, which results in the predicates becoming non-linear and unnecessarily obscured.
In our example, the new predicates are $p_1':4x-4t=0$, $p_2':\frac{12y}{6-t}-\frac{4t}{6-t}=0$, $p_3':p_3$, and $p_4':p_4$.
Moreover, to reflect covariances between the signals (off-diagonal entries in the covariance matrix), more complex transformations are required.

\paragraph{Explainability:} STL formulas are not easily interpretable \cite{DBLP:conf/iros/Siu0M23}. 
One can plot the predicates in STL, but it is not straightforward to visualize the overall formula or understand it directly. Additionally, the robustness semantics can be non-intuitive, as it combines distances from different predicates in a non-trivial way. 
In contrast, atoms in FPL can be easily visualized as paths with variance around them.
Moreover, if we express an FPL formula as an equivalent STL formula, even though the memory requirements of the atom and these predicates are similar, in practice the STL formula becomes large and hard for a human to understand. 
By contrast, FPL handles the geometric and logical aspects separately.
The geometric part (an atom) can easily be visualized, and the logical part is comparatively quite small ($1$ atom compared to $n+2$ predicates), making the formula more interpretable.
Generally, FPL allows for easily incorporating various standard distances and ways of aggregating robustness (e.g., maximum, integrating/averaging).

\paragraph{Learning:} Learning an STL formula directly is a challenging task, as discussed towards the end of Section~\ref{sec:motivation}.
Learning an FPL formula is more straightforward, as can be seen in Section~\ref{sec:learning}.

Altogether, while some extensions of STL could in principle allow the desired functionality, they are comparatively unintuitive, complex to use, and result in large, hard-to-understand formulas that are difficult to learn, see Motivation.
In contrast, we provide one simple framework extending STL more fundamentally and catering for these needs in a systematic, adequate, easy-to-use, and unified way.

\subsection{Extending syntax}
\label{app:extend_syntax}

\subsubsection{Negation}
If we want to define a negation operator, it first needs to be defined for atoms. 
Since our semantics is defined as a distance, it must always be non-negative.
Intuitively, negation of an atom would mean that one wants to stay away from this path as much as possible. 
Therefore, for the negation of an atom, the distance should also be non-negative, increase as one comes closer to it, and decrease as one goes farther.
This would push trajectories to the boundary of the probability distribution, which is typically undesirable.
One example of such semantics could be: \[\dist(\zeta^{[0,\nu]}, \neg\atom{}{\tau}) := \frac{1}{\dist(\zeta^{[0,\nu]}, \atom{}{\tau})}\]

\subsubsection{Conjunction}
Adding conjunction to the syntax of FPL is possible; however, we do not include it here for the following reasons.

First, this operator is often unnecessary or even completely pointless in the preference setting, since it describes two distinct behaviors of the system to be followed at the same time.
The idea behind our logic is to capture good trajectories, which makes requiring two good trajectories to be followed simultaneously less sensible, as following one could lead to violation of the other.
Instead, one could opt for a more nuanced \emph{trade-off analysis}.

Second, defining semantics is not easy, because the basis of our semantics—the set of fuzzy paths—is no longer defined.
If we want to keep %
a definition like $\dist(\zeta,\phi) = \min_{\bm{p}\in\mathcal{I}(\phi)} \dist(\zeta, \bm{p})$, $\mathcal{I}(\phi)$ needs to be replaced by a set of formulas that are until- and disjunction-free. (See also the issues described in Remark~\ref{rem:semantics}.)
Consider the example $\atom{1}{\tau_1}\land\atom{2}{\tau_2}$; then $\mathcal{I}(\phi) = \{\atom{1}{\tau_1}\land\atom{2}{\tau_2}\}$ is not a set of ``paths'' any longer.

Third, converting any formula to a left normal form (LNF) does not work.
Converting the formula to LNF has the advantage that it becomes easier to deal with the semantics of until, because the horizon for until to switch is fixed.
However, this would no longer be the case if we add conjunction to the mix.
Consider the formula $((\phi_1\U\phi_2)\land\phi_3)\U \phi_4$; the outer until cannot be simplified here, and the horizon of the outer until varies depending on the switching time of the inner until.
This renders our recursive procedure inapplicable.

\subsubsection{Safety}
One of the main use cases of formal specification languages such as LTL and STL is to express properties that ensure the system behaves safely in an environment.
However, our logic as defined here essentially deals with ``good'' behaviors. 
Here, avoiding ``bad'' behaviors happens implicitly by preferring the good behaviors.
Therefore, this logic is more suited to settings where the good and bad trajectories have a clear distinction between them.

As discussed in Section~\ref{sec:motivation}, if required, a hard safety specification (that can be described in a more suitable temporal logic like STL) can be combined with our logic.
To adjust the semantics for this hybrid formula, first check whether the safety specification is satisfied.
If it is not satisfied, then set the distance to $\infty$; otherwise, compute the semantics as usual.

\subsection{Model checking and synthesis}%
While this paper focuses on motivating and defining the new logic FPL, we have accumulated further arguments why it is worth this effort, presented for the reader's convenience in the appendix due to space constraints. 
First, we claim that model checking a hybrid automaton with respect to this logic becomes possible in a certain restricted setting. 
We identified the necessary assumptions and their intuitive meanings in Appendix~\ref{app:mc}, where our algorithm for computing semantics can be extended for model checking. 
Since the algorithm itself is not central, it appears only in the appendix; the important message is a proof of concept that reasonable classes of systems can be FPL model checked.
Moreover, in Appendix~\ref{app:sampling_synthesis}, we examine the possibilities of sampling a path from a given formula, which can be useful for path synthesis.

\section{Outlook and future work}
We have introduced the idea of paths being first-class citizens in a temporal logic, coming equipped with a degree of preference. 
This helps us separate the logical and geometric aspects, making them more interpretable than their intertwined counterparts.
Such a specification language can be useful for defining the most desirable behavior.
It is particularly useful when such behaviors are hard to obtain manually in a model-based way and we are limited to observing demonstrations.
In such cases, the specification can be easily learned and thus provides a means to formally work with empirical experience.
This new paradigm opens many questions and suggests several directions for future work.
\emph{Conceptual} questions concern different ways to combine hard constraints with the introduced soft preferences; different approaches to modeling distributions of paths, their variance, and their connection to system dynamics; different concepts of optimality of paths; or criteria for sampling when imitating the variance observed in demonstrations.
\emph{Theoretical} questions on decidability and approximability are open for various classes of systems.
\emph{Practical} algorithms are desirable to apply the framework in real motion planning tasks.
Finally, given the various hyper-parameters of the framework, concrete \emph{applications} should be investigated to identify the most interesting fragments and classes of the logic.
We hope to trigger interest in theoretical aspects of this new perspective addressing a topic of high relevance in today's data-driven design of automated systems.
Logic could thus contribute with rigor to the---by nature vague---learning-based approaches.

\bibliographystyle{ACM-Reference-Format}
\bibliography{refs}

\appendix

\section{Semantics}

\subsection{Computing Exact Semantics}
\label{app:exact_semantics}

We discuss the integration and maximum semantics separately here since they result in solving different types of optimization problems.

For integration semantics, we elucidate our points using an example.
Consider the formula $\atom{1}{1}\U(\atom{2}{1}\lor(\atom{3}{1}\U \atom{4}{1}))$ and a trajectory $\zeta^{[0,5]}$.
To compute its semantics using $\dist_{int}$, one needs to solve the following equation (dropping the time horizons for readability).
\begin{align}
    &\min_{t} \Big\{ \dist_{int}(\zeta,\bm{\pi_1}) + 
    \min \big\{ \dist_{int}(\zeta_{\rightarrow t}, \bm{\pi_2}), \\
    & \quad \min_{t'} \big( \dist_{int}(\zeta_{\rightarrow t}, \bm{\pi_3}) + 
    \dist_{int}(\zeta_{\rightarrow t+t'}, \bm{\pi_4}) \big) \big\} \Big\} \notag
\end{align}

which can be reformulated as:
\begin{align}
    &\min \Big\{ 
    \min_{t} \big\{ \dist_{int}(\zeta, \bm{\pi_1}) + \dist_{int}(\zeta_{\rightarrow t}, \bm{\pi_2}) \big\}, \\
    & \quad \min_{t,t'} \big\{ \dist_{int}(\zeta, \bm{\pi_1}) + \dist_{int}(\zeta_{\rightarrow t}, \bm{\pi_3}) + 
    \dist_{int}(\zeta_{\rightarrow t+t'}, \bm{\pi_4}) \big\} 
    \Big\} \notag
\end{align}

Here, we need to solve 2 separate optimization problems. 
In general, if there are $l$ disjunctions in the formula, it requires solving $2^l$ optimization problems in the worst case.
Each optimization problem can have at most $m$ (= number of until's) time-variables which can be tuned.
The complexity of optimization problems will depend on the classes of functions used and the distance metric between a distribution and a point.
If we consider linear atoms (with linear entries in the mean vector and the covariance matrix), each optimization problem can have summation of terms like $\int_{0}^{t_m}\sqrt\frac{p(t_1,t_2,\dots t_{m-1},t)}{q(t_1,t_2\dots t_{m-1},t)}dt$ where $p$ and $q$ are polynomials, $t_1,t_2,t_m$ are the time variables that need to be optimized (see Example~\ref{ex:sem-atom} to get an idea).
These problem can be difficult to solve due to their generic nature and if we consider other classes of functions, things become even more complicated (and maybe undecidable).

\noindent
For maximum semantics, the optimization problem is given by the equation
\begin{align}
    &\min_{t} \max \bigg\{ \dist_{max}(\zeta, \bm{\pi_1}), 
    \min \Big\{ \dist_{max}(\zeta_{\rightarrow t}, \bm{\pi_2}), \\
    & \quad \min_{t'} \max \big\{ \dist_{max}(\zeta_{\rightarrow t}, \bm{\pi_3}), 
    \dist_{max}(\zeta_{\rightarrow t+t'}, \bm{\pi_4}) \big\} \Big\} \bigg\} \notag
\end{align}

Due to the alternation between min and max here, the problem becomes a lot more complex even for linear atoms, and there is no straightforward way of dealing with it.
Also, due to the difficulty identified in Remark~\ref{rem:semantics}, the value of this function can change suddenly which can be a problem for optimization techniques.

\subsection{Example \ref{ex:sem-atom} with computation}
\label{app:example}
\begin{example}
	We show the computation for Example~\ref{ex:sem-atom} here.
	\begin{itemize}
		\item[-] First, we compute with maximum semantics $\dist_{max}$ which is given by
		\begin{align}
			\hspace{10pt}
			&\dist_{max}(\zeta^{[0,3]}, \bm{\pi}^{[0,3]}) = \max_{t\in[0,3]} d^1_M(\zeta^{[0,3]}(t), \bm{\pi}^{[0,3]}(t))\notag\\
			&= \max_{t\in[0,3]} \sqrt{
				\begin{pmatrix}
					0 \\
					-1-t/3
				\end{pmatrix}^T
				\begin{pmatrix}
					16 & 0 \\
					0 & 144/(6-t)^2
				\end{pmatrix}
				\begin{pmatrix}
					0 \\
					-1-t/3
				\end{pmatrix}	
			}\notag\\
			&= \max_{t\in[0,3]}\ \frac{12+4t}{6-t} = 8\notag
		\end{align}
		It is also visible from the figure that the maximum distance is achieved at $t=3$.
		\item[-] Second, using the integration semantics, we get
		\begin{align} 
			\ \ \ &\dist_{int}(\zeta^{[0,3]}, \bm{\pi}^{[0,3]}) = \int_{0}^{3} d^1_M(\zeta^{[0,3]}(t), \bm{\pi}^{[0,3]}(t))dt \notag\\
			&= \int_{0}^{3} \frac{12+4t}{6-t}dt = \int_{3}^{6} \frac{36-4u}{u}du = 36\ln 2 - 12 \approx 13\notag 
		\end{align}
		This gives us the average distance over the interval $[0,3]$ as $\approx 4.3$.
	\end{itemize}
	\qee
\end{example}

\section{Proofs}
\label{app:proof}
\subsection{Proof of Lemma~\ref{lem:lnf}}

\begin{lemma} 
	The following equivalences hold for \acrshort{fpl} formulas
	\begin{enumerate}
		\item $(\phi_1\lor\phi_2)\cdot\phi_3 \equiv (\phi_1 \cdot\phi_3) \lor (\phi_2\cdot\phi_3)$
		\item $(\phi_1\cdot\phi_2)\cdot\phi_3 \equiv \phi_1 \cdot (\phi_2\cdot\phi_3)$
		\item $(\phi_1\U\phi_2)\cdot\phi_3 \equiv \phi_1 \U (\phi_2\cdot\phi_3)$
		\item $(\phi_1\lor\phi_2)\U\phi_3 \equiv (\phi_1\U\phi_3) \lor(\phi_2 \U \phi_3)$
		\item $(\phi_1\cdot\phi_2)\U\phi_3 \equiv (\phi_1\U\phi_3) \lor(\phi_1\cdot(\phi_2 \U \phi_3))$
		\item $(\phi_1\U\phi_2)\U\phi_3 \equiv \phi_1\U(\phi_2\U\phi_3)$	
	\end{enumerate}
\end{lemma}
\begin{proof}
	To show these equivalences, we just need to prove that their fuzzy path sets $\fuzzypathset(\phi)$ are the same.
	\begin{enumerate}
		\item \underline{$(\phi_1\lor\phi_2)\cdot\phi_3 \equiv (\phi_1 \cdot\phi_3) \lor (\phi_2\cdot\phi_3)$}:
		\[\fuzzypathset((\phi_1\lor\phi_2)\cdot\phi_3)\]
		\begin{align}
			&=\{\bm{p_1}^{[0,v_1]}\cdot\bm{p_2}^{[0,v_2]} \mid \bm{p_1}^{[0,v_1]}\in \fuzzypathset(\phi_1\lor \phi_2), \bm{p_2}^{[0,v_2]}\in \fuzzypathset(\phi_3)\} \notag\\
			&=\{\bm{p_1}^{[0,v_1]}\cdot\bm{p_2}^{[0,v_2]} \mid \bm{p_1}^{[0,v_1]}\in \fuzzypathset(\phi_1)\cup \fuzzypathset(\phi_2), \bm{p_2}^{[0,v_2]}\in \fuzzypathset(\phi_3)\} \notag\\
			&=\{\bm{p_1}^{[0,v_1]}\cdot\bm{p_2}^{[0,v_2]} \mid \bm{p_1}^{[0,v_1]}\in \fuzzypathset(\phi_1), \bm{p_2}^{[0,v_2]}\in\fuzzypathset(\phi_3)\} \notag \\
			&\ \ \ \ \  \cup\ \{\bm{p_1}^{[0,v_1]}\cdot\bm{p_2}^{[0,v_2]} \mid \bm{p_1}^{[0,v_1]}\in\fuzzypathset(\phi_2), \bm{p_2}^{[0,v_2]}\in \fuzzypathset(\phi_3)\} \notag\\
			&=\fuzzypathset(\phi_1\cdot\phi_3)\cup\fuzzypathset(\phi_2\cdot\phi_3) \notag \\
			&=\fuzzypathset((\phi_1\cdot\phi_3)\lor(\phi_2\cdot\phi_3)) \notag
		\end{align}
	
		\item \underline{$(\phi_1\cdot\phi_2)\cdot\phi_3 \equiv \phi_1 \cdot (\phi_2\cdot\phi_3)$}:
		\[\fuzzypathset((\phi_1\cdot\phi_2)\cdot\phi_3)\]
		\begin{align}
			&=\{\bm{p_1}^{[0,v_1]}\cdot\bm{p_2}^{[0,v_2]} \mid \bm{p_1}^{[0,v_1]}\in \fuzzypathset(\phi_1\cdot \phi_2), \bm{p_2}^{[0,v_2]}\in \fuzzypathset(\phi_3)\} \notag\\
			&=\{\bm{q_1}^{[0,w_1]}\cdot\bm{q_2}^{[0,w_2]}\cdot\bm{p_2}^{[0,v_2]} \mid \bm{q_1}^{[0,w_1]}\in \fuzzypathset(\phi_1),\bm{q_2}^{[0,w_2]}\in \fuzzypathset(\phi_2), \notag\\ 
			&\ \ \ \ \ \ \bm{p_2}^{[0,v_2]}\in \fuzzypathset(\phi_3)\} \notag\\
			&=\{\bm{q_1}^{[0,w_1]}\cdot\bm{r_2}^{[0,u_2]} \mid \bm{q_1}^{[0,w_1]}\in \fuzzypathset(\phi_1), \bm{r_2}^{[0,u_2]}\in\fuzzypathset(\phi_1\cdot\phi_3)\} \notag \\
			&=\fuzzypathset((\phi_1)\cdot(\phi_2\cdot\phi_3)) \notag
		\end{align}
		
		\item \underline{$(\phi_1\U\phi_2)\cdot\phi_3 \equiv \phi_1 \U (\phi_2\cdot\phi_3)$}:
		\[\fuzzypathset((\phi_1\U\phi_2)\cdot\phi_3)\]
		\begin{align}
			&=\{\bm{p_1}^{[0,v_1]}\cdot\bm{p_2}^{[0,v_2]} \mid \bm{p_1}^{[0,v_1]}\in \fuzzypathset(\phi_1\U \phi_2), \bm{p_2}^{[0,v_2]}\in \fuzzypathset(\phi_3)\} \notag\\
			&=\{\bm{q_1}^{[0,t]}\cdot\bm{q_2}^{[0,w_2]}\cdot\bm{p_2}^{[0,v_2]} \mid \bm{q_1}^{[0,w_1]}\in \fuzzypathset(\phi_1), t\leq w_1, \notag\\  
			&\ \ \ \ \ \ \bm{q_2}^{[0,w_2]}\in \fuzzypathset(\phi_2),
			\bm{p_2}^{[0,v_2]}\in \fuzzypathset(\phi_3)\} \notag\\
			&=\{\bm{q_1}^{[0,t]}\cdot\bm{r_2}^{[0,u_2]} \mid \bm{q_1}^{[0,w_1]}\in \fuzzypathset(\phi_1),  t\leq w_1, \bm{r_2}^{[0,u_2]}\in\fuzzypathset(\phi_1\cdot\phi_3)\} \notag \\
			&=\fuzzypathset((\phi_1)\U(\phi_2\cdot\phi_3)) \notag
		\end{align}
		
		\item \underline{$(\phi_1\lor\phi_2)\U\phi_3 \equiv (\phi_1 \U\phi_3)\lor (\phi_2\cdot\phi_3)$}:
		\[\fuzzypathset((\phi_1\lor\phi_2)\U\phi_3)\]
		\begin{align}
			&=\{\bm{p_1}^{[0,t]}\cdot\bm{p_2}^{[0,v_2]} \mid \bm{p_1}^{[0,v_1]}\in \fuzzypathset(\phi_1\lor \phi_2), t\leq v_1, \bm{p_2}^{[0,v_2]}\in \fuzzypathset(\phi_3)\} \notag\\
			&=\{\bm{p_1}^{[0,t]}\cdot\bm{p_2}^{[0,v_2]} \mid \bm{p_1}^{[0,v_1]}\in \fuzzypathset(\phi_1)\cup \fuzzypathset(\phi_2), t\leq v_1, \bm{p_2}^{[0,v_2]}\in \fuzzypathset(\phi_3)\} \notag\\
			&=\{\bm{p_1}^{[0,t]}\cdot\bm{p_2}^{[0,v_2]} \mid \bm{p_1}^{[0,v_1]}\in \fuzzypathset(\phi_1), t\leq v_1, \bm{p_2}^{[0,v_2]}\in \fuzzypathset(\phi_3)\} \notag\\
			&\ \ \ \ \ \ \cup\{\bm{p_1}^{[0,t]}\cdot\bm{p_2}^{[0,v_2]} \mid \bm{p_1}^{[0,v_1]}\in \fuzzypathset(\phi_2), t\leq v_1, \bm{p_2}^{[0,v_2]}\in \fuzzypathset(\phi_3)\} \notag\\
			&=\fuzzypathset((\phi_1 \U\phi_3)\lor (\phi_2\cdot\phi_3)) \notag
		\end{align}
	
		\item \underline{$(\phi_1\cdot\phi_2)\U\phi_3 \equiv (\phi_1 \U\phi_3)\lor (\phi_1\cdot(\phi_2\U\phi_3))$}:
		\[\fuzzypathset((\phi_1\cdot\phi_2)\U\phi_3)\]
		\begin{align}
			&=\{\bm{p_1}^{[0,t]}\cdot\bm{p_2}^{[0,v_2]} \mid \bm{p_1}^{[0,v_1]}\in \fuzzypathset(\phi_1\cdot \phi_2), t\leq v_1, \bm{p_2}^{[0,v_2]}\in \fuzzypathset(\phi_3)\} \notag\\
			&=\{\bm{q_1}^{[0,t]}\cdot\bm{p_2}^{[0,v_2]} \mid \bm{q_1}^{[0,w_1]}\in \fuzzypathset(\phi_1), t\leq w_1, \bm{p_2}^{[0,v_2]}\in \fuzzypathset(\phi_3)\}\ \cup\notag\\
			&\ \ \ \ \ \ \{\bm{q_1}^{[0,w_1]}\cdot\bm{q_2}^{[0,t]}\cdot\bm{p_2}^{[0,v_2]} \mid \bm{q_1}^{[0,w_1]}\in \fuzzypathset(\phi_1), \bm{q_2}^{[0,w_2]}\in \fuzzypathset(\phi_2), \notag\\
			&\ \ \ \ \ \ \ \ \ \ \ t\leq w_2, \bm{p_2}^{[0,v_2]}\in \fuzzypathset(\phi_3)\} \notag\\
			&=\fuzzypathset((\phi_1 \U\phi_3)\lor \phi_1\cdot(\phi_2\U\phi_3)) \notag
		\end{align}
		
		\item \underline{$(\phi_1\U\phi_2)\U\phi_3 \equiv \phi_1 \U (\phi_2\U\phi_3)$}:
		\[\fuzzypathset((\phi_1\U\phi_2)\U\phi_3)\]
		\begin{align}
			&=\{\bm{p_1}^{[0,t]}\cdot\bm{p_2}^{[0,v_2]} \mid \bm{p_1}^{[0,v_1]}\in \fuzzypathset(\phi_1\U \phi_2), t\leq v_1, \bm{p_2}^{[0,v_2]}\in \fuzzypathset(\phi_3)\} \notag\\
			&=\{\bm{q_1}^{[0,t']}\cdot\bm{q_2}^{[0,t'']}\cdot\bm{p_2}^{[0,v_2]} \mid \bm{q_1}^{[0,w_1]}\in \fuzzypathset(\phi_1), t'\leq w_1, \notag\\  
			&\ \ \ \ \ \ \bm{q_2}^{[0,w_2]}\in \fuzzypathset(\phi_2), t''\leq w_2, 
			\bm{p_2}^{[0,v_2]}\in \fuzzypathset(\phi_3)\} \notag\\
			&=\{\bm{q_1}^{[0,t']}\cdot\bm{r_2}^{[0,u_2]} \mid \bm{q_1}^{[0,w_1]}\in \fuzzypathset(\phi_1),  t\leq w_1, \bm{r_2}^{[0,u_2]}\in\fuzzypathset(\phi_1\U\phi_3)\} \notag \\
			&=\fuzzypathset((\phi_1)\U(\phi_2\U\phi_3)) \notag
		\end{align}
	\end{enumerate}
\end{proof}

\subsection{Proof of Theorem~\ref{thm:main-max}}

\begin{lemma}
	\label{lem:abcd}
	For $a,b,c,d\in \reals$, we have the following two statements \[|\max\{a,b\}-\max\{c,d\}|\leq \max\{|a-c|,|b-d|\}\]
	\[|\min\{a,b\}-\min\{c,d\}|\leq \max\{|a-c|,|b-d|\}\]
\end{lemma}
\begin{proof}
	(Sketch.) They can be proved by splitting all the cases i.e. when $a>b$, $c>d$, etc.
\end{proof}

\begin{lemma}
	\label{lem:init}
	For a given LC trajectory $\zeta^{[0,\nu]}$, and an FPL formula $\phi$, $\exists C_1\in \reals^+$ such that $\forall \bm{\pi}^{[0,\tau]}\in atoms(\phi)$ and $0\leq t\leq \min\{\nu,\tau\}$
	\[d(\zeta^{[0,\nu]}(t), \bm{\pi}^{[0,\tau]}(t))\leq C_1\ \]

\end{lemma}
\begin{proof}
	Using LC of $\zeta(t)$, for all $t\in[0,\nu]$, we have
	\begin{equation}
		\label{eq:lem11-1}
		\|\zeta^{[0,\nu]}(t)-\zeta^{[0,\nu]}(0)\| \leq K_\zeta t 
	\end{equation}
	Using the assumption on $d(\cdot,\bm{\pi})$ with constant $K_{d,\zeta}$ and equation \ref{eq:lem11-1}, for all $\bm{\pi}\in atoms(\phi)$ and $t\in[0,\min\{\nu,\tau\}]$, we get
	\begin{equation}
		\label{eq:lem11-2}
		|d(\zeta^{[0,\nu]}(t), \bm{\pi}^{[0,\tau]}(t))- d(\zeta^{[0,\nu]}(0), \bm{\pi}^{[0,\tau]}(t))| \leq  K_{d,\zeta}K_\zeta t
	\end{equation}
	Using the assumption on $d(p, \bm{\pi}(\cdot))$ with Lipschitz constant $K_{d,\bm{\pi}}$, for all $\bm{\pi}\in atoms(\phi)$ and $t\in[0,\min\{\nu,\tau\}]$, we get
	\begin{equation}
		\label{eq:lem11-3}
		|d(\zeta^{[0,\nu]}(0), \bm{\pi}^{[0,\tau]}(t)) - d(\zeta^{[0,\nu]}(0), \bm{\pi}^{[0,\tau]}(0))| \leq  K_{d,\bm{\pi}} t
	\end{equation}
	Now, using equations \ref{eq:lem11-2} and \ref{eq:lem11-3}, for any $\bm{\pi}^{[0,\tau]}\in atoms(\phi)$
	\begin{align}
		d(&\zeta^{[0,\nu]}(t), \bm{\pi}^{[0,\tau]}(t)) \notag \\
		&\leq d(\zeta(0), \bm{\pi}^{[0,\tau]}(0))+ K_{d,\bm{\pi}} t + K_\zeta K_{d,\zeta} t \notag\\	
		&\leq d(\zeta(0), \bm{\pi}^{[0,\tau]}(0)) + (K_{d,\bm{\pi}} + K_\zeta K_{d,\zeta})\cdot\min\{\nu,\tau\} \notag
	\end{align}
	Now, for all atoms $\atom{}{\tau}\in atoms(\phi)$, we have
	\begin{align}
		d(&\zeta^{[0,\nu]}(t), \bm{\pi}^{[0,\tau]}(t)) \notag \\
		&\leq \max_{\bm{\pi}\in atoms(\phi)} \big(d(\zeta(0), \atom{}{\tau}(0)) + (K_{d,\bm{\pi}} + K_\zeta K_{d,\zeta})\cdot\min\{\nu,\tau\}\big) \notag\\
		&
		= C_1 \notag
	\end{align}
\end{proof}

\begin{lemma}
	\label{lem:cdelta-max}
	For an LC trajectory $\zeta^{[0,\tau]}$, an FPL formula $\phi$, $\exists C_2\in\reals^+$ such that $\forall\atom{}{\tau}\in atoms(\phi)$
	\[\big|\dist_{max}(\zeta^{[0,\tau]}, \atom{}{\tau})-\dist_{max}(\zeta^{[0,\tau-\delta]}, \atom{}{\tau-\delta})\big| \leq C_2\delta\]
	where $0<\delta <h_{min}(\phi)$.
\end{lemma}
\begin{proof}
	For an atom $\atom{}{\tau}\in atoms(\phi)$, we have
	\begin{align}
		\big|&\dist_{max}(\zeta^{[0,\tau]}, \atom{}{\tau})-\dist_{max}(\zeta^{[0,\tau-\delta]}, \atom{}{\tau-\delta})\big| \notag\\
		&= \Big|\underset{t\in[0,\tau]}{\max}d(\zeta^{[0,\tau]}(t), \atom{}{\tau}(t))\ 
 		-\underset{t\in[0,\tau-\delta]}{\max}d(\zeta^{[0,\tau]}(t), \atom{}{\tau}(t))
		\Big| \notag
	\end{align}
	Let's assume that the optimal values of the first and second terms are achieved at $t_1\in[0,\tau]$ and $t_2\in[0,\tau-\delta]$.
	\[= \Big|d(\zeta^{[0,\tau]}(t_1), \atom{}{\tau}(t_1))\ 
	-d(\zeta^{[0,\tau-\delta]}(t_2), \atom{}{\tau-\delta}(t_2))
	\Big| \]
	Since $\delta < h_{min}(\phi)$, $\tau-\delta>0$. 
	If $t_1\in [0,\tau-\delta]$, both terms are equal and their difference becomes 0 and the lemma holds. 
	If $t_1\in [\tau-\delta, \tau]$, since the maximum of the second term is reached at $t_2$ replacing it with any other value will only increase this difference
	\begin{align}
		&\leq \Big|d(\zeta^{[0,\tau]}(t_1), \atom{}{\tau}(t_1))\ 
		-d(\zeta^{[0,\tau-\delta]}(\tau-\delta), \atom{}{\tau-\delta}(\tau-\delta))
		\Big| \notag\\
		&= \Big|d(\zeta^{[0,\tau]}(t_1), \atom{}{\tau}(t_1))\ 
		-d(\zeta^{[0,\tau-\delta]}(\tau-\delta), \atom{}{\tau}(t_1)) \notag\\
		&\ + d(\zeta^{[0,\tau-\delta]}(\tau-\delta), \atom{}{\tau}(t_1))
		-d(\zeta^{[0,\tau-\delta]}(\tau-\delta), \atom{}{\tau-\delta}(\tau-\delta))
		\Big| \notag\\
		&\leq \Big|d(\zeta^{[0,\tau]}(t_1), \atom{}{\tau}(t_1))\ 
		-d(\zeta^{[0,\tau-\delta]}(\tau-\delta), \atom{}{\tau}(t_1))\Big| \notag\\
		&\ + \Big|d(\zeta^{[0,\tau-\delta]}(\tau-\delta), \atom{}{\tau}(t_1))
		-d(\zeta^{[0,\tau-\delta]}(\tau-\delta), \atom{}{\tau-\delta}(\tau-\delta))
		\Big| \notag
	\end{align}
	Now, using LC of $\zeta(\cdot)$, $d(\cdot,\mathbf{\pi})$, and $d(p,\pi(\cdot))$: for all $\atom{}{\tau}\in atoms(\phi)$, we have
	\[\leq K_{d,\zeta}K_{\zeta}|t_1-\tau+\delta|
	+ K_{d,\pi} |t_1-\tau+\delta|
	\]
	Since $t_1\in[\tau-\delta, \tau]$,
	\[\leq(K_{d,\zeta}K_{\zeta}
	+ K_{d,\pi})\delta \]
	\[= C_2\delta \notag\]
	where $C_2=K_{d,\zeta}K_{\zeta}
	+ K_{d,\pi}$.
	
\end{proof}

\begin{lemma}
	\label{lem:until-free-max}
	For an LC trajectory $\zeta^{[0,\nu]}$, an FPL formula $\phi$ in LNF, $0<\delta< h_{\min}(\phi)$ such that $\dist_{max}(\zeta^{[0,\nu-\delta]}_{\rightarrow \delta}, \phi)<\infty$ then $\exists C\in\reals^+$ such that
	\[|\dist_{max}(\zeta^{[0,\nu-\delta]}_{\rightarrow \delta}, \phi) - \dist_{max}(\zeta^{[0,\nu]}, \phi)|\leq C\delta\]
\end{lemma}
\begin{proof} 
	We will show it using induction on the structure of the formula.
	The base case is $\phi = \atom{}{\tau}$.
	Since we know $\dist(\zeta^{[0,\nu-\delta]}_{\rightarrow \delta}, \phi)<\infty$, $\nu-\delta \geq \tau$.
	Now,
	\begin{align}
		|\dist_{max}(&\zeta^{[0,\nu-\delta]}_{\rightarrow \delta}, \atom{}{\tau}) - \dist_{max}(\zeta^{[0,\nu]}, \atom{}{\tau})| \notag \\ 
		&=\Big|\underset{t\in[0,\tau]}{\max}\big\{d(\zeta^{[0,\nu-\delta]}_{\rightarrow \delta}(t), \atom{}{\tau}(t))\big\} \notag\\ 
		&\quad -\underset{t\in[0,\tau]}{\max}\big\{d(\zeta^{[0,\nu]}(t), \atom{}{\tau}(t))\big\}\Big| \notag 
	\end{align}
	Assuming the first term and the second term reaches the optimal value on $t_1$ and $t_2$ respectively, then
	\begin{align}
		&=|d(\zeta^{[0,\nu-\delta]}_{\rightarrow \delta}(t_1), \atom{}{\tau}(t_1)) - d(\zeta^{[0,\nu]}(t_2), \atom{}{\tau}(t_2))| \label{eq:until-free-max-1}
	\end{align}
	There are two possibilities here. 
	\begin{enumerate}
		\item first term is bigger in equation~\ref{eq:until-free-max-1}. \\
		Since $t_2$ is optimal for second term, we have
		\[\leq |d(\zeta^{[0,\nu-\delta]}_{\rightarrow \delta}(t_1), \atom{}{\tau}(t_1)) - d(\zeta^{[0,\nu]}(t_1), \atom{}{\tau}(t_1))| \]
		Now, using LC of $\zeta(\cdot)$ and $d(\cdot,\mathbf{\pi})$, we get
		\[\leq  K_{d,\zeta} K_{\zeta}\delta = C\delta\]
		where $C=K_{d,\zeta} K_{\zeta}$.
		\item the second term is bigger in equation~\ref{eq:until-free-max-1}. \\
		Since $t_1$ is optimal for the first term, we get
		\[\leq |d(\zeta^{[0,\nu-\delta]}_{\rightarrow \delta}(t_2), \atom{}{\tau}(t_2)) - d(\zeta^{[0,\nu]}(t_2), \atom{}{\tau}(t_2))| \]
		Again, using LC of $\zeta(\cdot)$ and
		$d(\cdot,\mathbf{\pi})$, we get
		\[\leq  K_{d,\zeta} K_{\zeta}\delta = C\delta\]
	\end{enumerate} 
	This completes the base case.\\
	
	\noindent
	To prove the general case, the inductive hypothesis is that the lemma holds for all sub-formulas of $\phi$ and we shall show that it holds for $\phi$ too. 
	\begin{itemize}
		\item
		\underline{$\phi=\phi_1\lor\phi_2$}:\ \  $|\dist_{max}(\zeta^{[0,\nu-\delta]}_{\rightarrow\delta}, \phi)-\dist_{max}(\zeta^{[0,\nu]}, \phi)|$
		\begin{align} 
			&\leq \big|\min\big\{\dist_{max}(\zeta^{[0,\nu-\delta]}_{\rightarrow\delta}, \phi_1), \dist_{max}(\zeta^{[0,\nu-\delta]}_{\rightarrow\delta}, \phi_2)\big\} \notag\\ 
			&\quad -\min\big\{\dist_{max}(\zeta^{[0,\nu]}, \phi_1), \dist_{max}(\zeta^{[0,\nu]}, \phi_2)\big\}\big| \notag\\
			&\text{Using Lemma~\ref{lem:abcd}}\notag\\
			&\leq \max\big\{\big|\dist_{max}(\zeta^{[0,\nu-\delta]}_{\rightarrow\delta}, \phi_1)-\dist_{max}(\zeta^{[0,\nu]}, \phi_1)\big|, \notag\\
			&\quad \big|\dist_{max}(\zeta^{[0,\nu-\delta]}_{\rightarrow\delta}, \phi_2)-\dist_{max}(\zeta^{[0,\nu]}, \phi_2)\big|\big\} \notag\\
			&\leq \max\{C\delta, C\delta\} \tag{using I.H.} \\
			&= C\delta \notag
		\end{align}
		\item 
		\underline{$\phi=\atom{}{\tau}\cdot\phi_2$}:\ \ 
		$\big|\dist_{max}(\zeta^{[0,\nu-\delta]}_{\rightarrow\delta}, \phi)-\dist_{max}(\zeta^{[0,\nu]}, 
		\phi)\big|$
		\begin{align}
			&= \big|\max\big\{\dist_{max}(\zeta^{[0,\tau]}_{\rightarrow\delta}, \atom{}{\tau}), \dist_{max}(\zeta^{[0,\nu-\delta-\tau]}_{\rightarrow\delta+\tau}, \phi_2)\big\} \notag\\
			&\quad -\max\big\{\dist_{max}(\zeta^{[0,\tau]}, \atom{}{\tau}), \dist_{max}(\zeta^{[0,\nu-\tau]}_{\rightarrow \tau}, \phi_2)\big\}\big| \notag\\
			&\text{Using Lemma~\ref{lem:abcd}}\notag\\
			&\leq \max\Big\{\big|\dist_{max}(\zeta^{[0,\tau]}_{\rightarrow\delta}, \atom{}{\tau})- \dist_{max}(\zeta^{[0,\tau]}, \atom{}{\tau})\big|, \notag\\
			&\quad \big|\dist_{max}(\zeta^{[0,\nu-\delta-\tau]}_{\rightarrow\delta+\tau}, \phi_2), \dist_{max}(\zeta^{[0,\nu-\tau]}_{\rightarrow \tau}- \phi_2)\big|\Big\} \notag\\
			&\leq \max\{C\delta, C\delta\} \tag{using I.H.} \\
			& = C\delta \notag
		\end{align}

		\item \underline{$\phi=\atom{}{\tau}\U\phi_2$}:\ \ 
		Let's assume that the optimal time to switch for $\zeta^{[{0,\nu}]}$ and $\zeta^{[0,\nu-\delta]}_{\rightarrow\delta}$ are $t$ and $t'$ respectively.
		There are two cases here.
		\begin{enumerate}
			\item If $\dist_{max}(\zeta^{[0,\nu-\delta]}_{\rightarrow\delta}, \phi)>\dist_{max}(\zeta^{[0,\nu]}, \phi)$ then since $t'$ is not optimal for the second term 
		\begin{align} 
			&\big|\dist_{max}(\zeta^{[0,\nu-\delta]}_{\rightarrow\delta}, \phi)-\dist_{max}(\zeta^{[0,\nu]}, \phi)\big| \notag \\
			&= \big|\max\big\{\dist_{max}(\zeta^{[0,t']}_{\rightarrow\delta}, \atom{}{t'}), \dist_{max}(\zeta^{[0,\nu-\delta-t']}_{\rightarrow\delta+t'}, \phi_2)\big\} \notag\\
			&\quad -\max\big\{\dist_{max}(\zeta^{[0,t]}, 
			\atom{}{t}), \dist_{max}(\zeta^{[0,\nu-t]}_{\rightarrow t}, \phi_2)\big\}\big| \notag\\
			&\leq \big|\max\big\{\dist_{max}(\zeta^{[0,t']}_{\rightarrow\delta}, \atom{}{t'}), \dist_{max}(\zeta^{[0,\nu-\delta-t']}_{\rightarrow\delta+t'}, \phi_2)\big\} \notag\\
			&\quad -\max\big\{\dist_{max}(\zeta^{[0,t']}, \atom{}{t'}), \dist_{max}(\zeta^{[0,\nu-t']}_{\rightarrow t'}, \phi_2)\big\}\big| \notag\\			
			&\text{Using Lemma~\ref{lem:abcd}}\notag\\
			&\leq \max\Big\{\big|\dist_{max}(\zeta^{[0,t']}_{\rightarrow\delta}, \atom{}{t'})- \dist_{max}(\zeta^{[0,t']}, \atom{}{t'})\big|, \notag\\
			&\quad \big|\dist_{max}(\zeta^{[0,\nu-\delta-t']}_{\rightarrow\delta+t'}, \phi_2)- \dist_{max}(\zeta^{[0,\nu-t']}_{\rightarrow t'}, \phi_2)\big|\Big\} \notag \\
			&\leq \max\{C\delta, C\delta\} \tag{using I.H.} \\
			& = C\delta \notag 
		\end{align}
		\item Otherwise, if $\dist_{max}(\zeta^{[0,\nu-\delta]}_{\rightarrow\delta}, \phi)\leq\dist_{max}(\zeta^{[0,\nu]}, \phi)$ then since $t$ is not optimal for the first term
		\begin{align}
			&\big|\dist_{max}(\zeta^{[0,\nu-\delta]}_{\rightarrow\delta}, \phi)-\dist_{max}(\zeta^{[0,\nu]}, \phi)\big| \notag \\
			&= \big|\max\big\{\dist_{max}(\zeta^{[0,t']}_{\rightarrow\delta}, \atom{}{t'}), \dist_{max}(\zeta^{[0,\nu-\delta-t']}_{\rightarrow\delta+t'}, \phi_2)\big\} \notag\\
			&\quad -\max\big\{\dist_{max}(\zeta^{[0,t]}, 
			\atom{}{t}), \dist_{max}(\zeta^{[0,\nu-t]}_{\rightarrow t}, \phi_2)\big\}\big| \notag\\
			&\leq \big|\max\big\{\dist_{max}(\zeta^{[0,t]}_{\rightarrow\delta}, \atom{}{t}), \dist_{max}(\zeta^{[0,\nu-\delta-t]}_{\rightarrow\delta+t}, \phi_2)\big\} \notag\\
			&\quad -\max\big\{\dist_{max}(\zeta^{[0,t]}, \atom{}{t}), \dist_{max}(\zeta^{[0,\nu-t]}_{\rightarrow t}, \phi_2)\big\}\big| \notag\\			
			&\text{Using Lemma~\ref{lem:abcd}}\notag\\
			&\leq \max\Big\{\big|\dist_{max}(\zeta^{[0,t]}_{\rightarrow\delta}, \atom{}{t})- \dist_{max}(\zeta^{[0,t]}, \atom{}{t})\big|, \notag\\
			&\quad \big|\dist_{max}(\zeta^{[0,\nu-\delta-t]}_{\rightarrow\delta+t}, \phi_2)- \dist_{max}(\zeta^{[0,\nu-t]}_{\rightarrow t}, \phi_2)\big|\Big\} \notag \\
			&\leq \max\{C\delta, C\delta\} \tag{using I.H.} \\
			& = C\delta \notag
		\end{align}
		\end{enumerate}

	\end{itemize}
	This completes the proof. 		
\end{proof}

\begin{theorem}
	For an LC trajectory $\zeta^{[0,\nu]}$, an FPL formula $\phi^{[0,v]}$ in LNF and $0<\delta <h_{min}(\phi)/k$, $\exists C,C_2 \in \reals^+$ such that
	\[|\tilde{\dist}_{max}(\zeta^{[0,\nu]}, \phi) - \dist_{max}(\zeta^{[0,\nu]}, \phi)| \leq k(C_2+C)\delta\] 
	where $k$ is the number of until operators in $\phi$.
\end{theorem}

\begin{proof}
	To prove the correctness, we use induction on the structure of the formula $\phi$.
	The base case is trivial since we assume that we can compute the exact distance for an atom.
	Induction hypothesis is that all the sub-formulas of smaller size can be approximated for any trajectory.
	\begin{itemize}
		\item \underline{$\phi = \phi_1 \lor \phi_2$}: \\
		Using induction hypothesis, we get
		\[|\tilde{\dist}_{max}(\zeta^{[0,\nu]}, \phi_1)-\dist_{max}(\zeta^{[0,\nu]}, \phi_1)| \leq k(C_2+C)\delta\]
		and 
		\[|\tilde{\dist}_{max}(\zeta^{[0,\nu]}, \phi_2)-\dist_{max}(\zeta^{[0,\nu]}, \phi_2)| \leq k(C_2+C)\delta\]
		Now, by definition,
		\[\tilde{\dist}_{max}(\zeta^{[0,\nu]}, \phi) = \min\{\tilde{\dist}_{max}(\zeta^{[0,\nu]}, \phi_1), \tilde{\dist}_{max}(\zeta^{[0,\nu]}, \phi_2)\}\]
		This gives us 
		\[|\tilde{\dist}_{max}(\zeta^{[0,\nu]}, \phi)-\dist_{max}(\zeta^{[0,\nu]}, \phi)| \leq k(C_2+C)\delta\]
		\item \underline{$\phi=\atom{}{\tau} \cdot \phi_2$}: \\
		If $\nu < \tau$ then by definition the semantics is $\infty$, so we assume $\nu > \tau$.
		Now, by induction hypothesis we get 
		\[|\tilde{\dist}_{max}(\zeta^{[0,\nu-\tau]}_{\rightarrow \tau}, \phi_2)-\dist_{max}(\zeta^{[0,\nu-\tau]}_{\rightarrow \tau}, \phi_2)| \leq k(C_2+C)\delta\]

		Since we assumed that we can compute the semantics for atoms exactly, we get
		\begin{align}
			|\tilde{\dist}_{max}&(\zeta^{[0,\nu]}, \phi)-\dist_{max}(\zeta^{[0,\nu]}, \phi)| \notag\\
			&=|\max\big\{\tilde{\dist}_{max}(\zeta^{[0,\tau]}, \atom{}{\tau}), \tilde{\dist}_{max}(\zeta^{[0,\nu-\tau]}_{\rightarrow \tau}, \phi_2)\big\} \notag\\
			&\quad - \max\big\{\dist_{max}(\zeta^{[0,\tau]}, \atom{}{\tau}) - \dist_{max}(\zeta^{[0,\nu-\tau]}_{\rightarrow \tau}, \phi_2)\big\}| \notag\\
			&\leq \max\big\{|\tilde{\dist}_{max}(\zeta^{[0,\tau]}, \atom{}{\tau})- \dist_{max}(\zeta^{[0,\tau]}, \atom{}{\tau})|, \notag\\
			&\quad |\tilde{\dist}_{max}(\zeta^{[0,\nu-\tau]}_{\rightarrow \tau}, \phi_2)- \dist_{max}(\zeta^{[0,\nu-\tau]}_{\rightarrow \tau}, \phi_2)| \big\} \notag \\
			&\leq \max\{0,k(C_2+C)\delta\} \notag\\ 
			&=k(C_2+C)\delta \notag
		\end{align}

		\item \underline{$\phi=\atom{}{\tau} \U \phi_2$}: \\
		If the value of semantics is infinite then the algorithm is trivially correct since all the computations would lead to $\infty$.
		So, let's assume the value is finite and that the optimal time to switch for this until was $t$. Therefore,
		\[\dist_{max}(\zeta^{[0,\nu]}, \phi)\] 
		\[=\max\big\{\dist_{max}(\zeta^{[0,t]}, \atom{}{t}),\dist_{max}(\zeta^{[0,\nu-t]}_{\rightarrow t}, \phi_2)\big\}\]
		Let $t'<t$ be the closest point in discretization to $t$, (which also implies $t-t'\leq \delta$)
		\[\tilde{\dist}_{max}(\zeta^{[0,\nu]}, \phi) \] 
		\[\leq\max\big\{\tilde{\dist}_{max}(\zeta^{[0,t']}, \atom{}{t'}),\tilde{\dist}_{max}(\zeta^{[0,\nu-t']}_{\rightarrow t'}, \phi_2)\big\}\]
		
		Now, since $\dist(\zeta^{[0,\nu-t]}_{\rightarrow t}, \phi_2)<\infty$, we use Lemma~\ref{lem:until-free-max} to get
		\begin{equation}
			|\dist_{max}(\zeta^{[0,\nu-t']}_{\rightarrow t'}, \phi_2) - \dist_{max}(\zeta^{[0,\nu-t]}_{\rightarrow t}, \phi_2)| \leq C\delta \label{eq:thm-main-1}
		\end{equation}
		Also, from Lemma~\ref{lem:cdelta-max} and the assumption that computing semantics for atoms is exact, we get 
		\begin{equation}	
			|\tilde{\dist}_{max}(\zeta^{[0,t']}, \atom{}{t'}) - \dist_{max}(\zeta^{[0,t]}, \atom{}{t})| \leq C_2\delta \label{eq:thm-main-2}
		\end{equation}
		
		Since there can be only $k-1$ until operators in $\phi_2$, induction hypothesis gives us 
		\begin{align}
			|\tilde{\dist}_{max}(\zeta^{[0,\nu-t']}_{\rightarrow t'}, &\phi_2) - \dist_{max}(\zeta^{[0,\nu-t]}_{\rightarrow t}, \phi_2)| \notag\\ 
			&\leq (k-1)(C_2+C)\delta \label{eq:thm-main-3}
		\end{align}
		Now,
		\begin{align}
			|\tilde{\dist}_{max}&(\zeta^{[0,\nu]}, \phi)-\dist_{max}(\zeta^{[0,\nu]}, \phi)| \notag\\
			= |&\tilde{\dist}_{max}(\zeta^{[0,t']}, \atom{}{t'}) + \tilde{\dist}_{max}(\zeta^{[0,\nu-t']}_{\rightarrow t'}, \phi_2) \notag\\
			- &\dist_{max}(\zeta^{[0,t]}, \atom{}{t}) - \dist_{max}(\zeta^{[0,\nu-t]}_{\rightarrow t}, \phi_2)| \notag\\
			\leq |&\tilde{\dist}_{max}(\zeta^{[0,t']}, \atom{}{t'}) - \dist_{max}(\zeta^{[0,t]}, \atom{}{t})| \notag\\
			+ &|\tilde{\dist}_{max}(\zeta^{[0,\nu-t']}_{\rightarrow t'}, \phi_2) - \dist_{max}(\zeta^{[0,\nu-t]}_{\rightarrow t}, \phi_2)| \notag 
		\end{align}
		Using equation \ref{eq:thm-main-1}, \ref{eq:thm-main-2}, and \ref{eq:thm-main-3} we get 
		\begin{align}
			&\leq C_2\delta + (k-1)(C_2+C)\delta + C\delta\notag\\
			&\leq k(C_2+C)\delta \notag
		\end{align}
	\end{itemize}
	
\end{proof}

\subsection{Proof of Theorem~\ref{thm:main-int}}

\begin{lemma}
	\label{lem:cdelta-int}
	For an LC trajectory $\zeta^{[0,\tau]}$, an FPL formula $\phi$, $\exists C_1\in\reals^+$ such that $\forall\atom{}{\tau}\in atoms(\phi)$
	\[\big|\dist_{int}(\zeta^{[0,\tau]}, \atom{}{\tau})-\dist_{int}(\zeta^{[0,\tau-\delta]}, \atom{}{\tau-\delta})\big| \leq C_1\delta\]
	where $0<\delta <h_{min}(\phi)$.
	
\end{lemma}
\begin{proof}
	For all atoms $\atom{}{\tau}\in atoms(\phi)$, $\tau-\delta>0$, therefore
	\begin{align}
		\big|\dist_{int}(\zeta^{[0,\tau]},& \atom{}{\tau})-\dist_{int}(\zeta^{[0,\tau-\delta]}, \atom{}{\tau-\delta})\big| \notag\\
		&= \Big|\int_{\tau-\delta}^{\tau}d(\zeta^{[0,\tau]}(t), \atom{}{\tau}(t))dt\Big| \notag 
	\end{align}
	Using, Lemma~\ref{lem:init}, we get
	\begin{align}
		\leq \Big|\int_{\tau-\delta}^{\tau}C_1dt\Big| \notag = C_1\delta
	\end{align}
\end{proof}

\begin{lemma}
	\label{lem:until-free-int}
	For an LC trajectory $\zeta^{[0,\nu]}$, an FPL formula in LNF $\phi$, $0<\delta< h_{\min}(\phi)$ such that $\dist_{int}(\zeta^{[0,\nu-\delta]}_{\rightarrow \delta}, \phi)<\infty$ then $\exists C\in\reals^+$ such that
	\[|\dist_{int}(\zeta^{[0,\nu-\delta]}_{\rightarrow \delta}, \phi) - \dist_{int}(\zeta^{[0,\nu]}, \phi)|\leq kC\nu\delta\]
	where $k$ is the number of concatenation or until operators in $\phi$.
\end{lemma}
\begin{proof} 
	We will show it using induction on the structure of the formula.
	The base case is $\phi = \atom{}{\tau}$.
	Since we assume \\
	$\dist(\zeta^{[0,\nu-\delta]}_{\rightarrow \delta}, \phi)<\infty$, we know $\nu-\delta \geq \tau$. Now,

	\begin{align}
		&|\dist_{int}(\zeta^{[0,\nu-\delta]}_{\rightarrow \delta}, \atom{}{\tau}) - \dist_{int}(\zeta^{[0,\nu]}, \atom{}{\tau})| \notag \\ 
		&=\Big|\int_{0}^{\tau}d(\zeta^{[0,\nu-\delta]}_{\rightarrow \delta}(t), \atom{}{\tau}(t))dt- \int_{0}^{\tau}d(\zeta^{[0,\nu]}(t), \atom{}{\tau}(t))dt\Big| \notag
	\end{align}
	\[\leq \Big|\int_{0}^{\tau}d(\zeta^{[0,\nu-\delta]}_{\rightarrow \delta}(t), \atom{}{\tau}(t))- d(\zeta^{[0,\nu]}(t), \atom{}{\tau}(t))dt\Big|\]
	By using LC of $\zeta(\cdot)$ and $d(\cdot,\mathbf{\pi})$, we get
	\[\leq \Big|K_{d,\zeta} K_{\zeta}\int_{0}^{\tau}\delta dt\Big| = K_{d,\zeta} K_{\zeta}\tau \delta = C\tau\delta
	\]	
	where $C=K_{d,\zeta} K_{\zeta}$. 
	This completes the base case. \\

	\noindent
	To prove the general case, the inductive hypothesis is that the lemma holds for all sub-formulas of $\phi$ and we shall show that it holds for $\phi$ too. 
	\begin{itemize}
		\item
		\underline{$\phi=\phi_1\lor\phi_2$}:\ \  $|\dist_{int}(\zeta^{[0,\nu-\delta]}_{\rightarrow\delta}, \phi)-\dist_{int}(\zeta^{[0,\nu]}, \phi)|$
		\begin{align} 
			&\leq \big|\min\big\{\dist_{int}(\zeta^{[0,\nu-\delta]}_{\rightarrow\delta}, \phi_1), \dist_{int}(\zeta^{[0,\nu-\delta]}_{\rightarrow\delta}, \phi_2)\big\} \notag\\
			&\quad -\min\big\{\dist_{int}(\zeta^{[0,\nu]}, \phi_1), \dist_{int}(\zeta^{[0,\nu]}, \phi_2)\big\}\big| \notag\\
			&\leq \max\big\{\big|\dist_{int}(\zeta^{[0,\nu-\delta]}_{\rightarrow\delta}, \phi_1)-\dist_{int}(\zeta^{[0,\nu]}, \phi_1)\big|, \notag\\
			&\quad  \big|\dist_{int}(\zeta^{[0,\nu-\delta]}_{\rightarrow\delta}, \phi_2)-\dist_{int}(\zeta^{[0,\nu]}, \phi_2)\big|\big\} \notag\\
			&\leq \max\{kC\nu\delta, kC\nu\delta\} \tag{I.H.} \\
			&\leq kC\nu\delta \notag
		\end{align}
		\item 
		\underline{$\phi=\atom{}{\tau}\cdot\phi_2$}:\ \ 
		$\big|\dist_{int}(\zeta^{[0,\nu-\delta]}_{\rightarrow\delta}, \phi)-\dist_{int}(\zeta^{[0,\nu]}, 
		\phi)\big|$
		\begin{align}
			&= \big|\dist_{int}(\zeta^{[0,\tau]}_{\rightarrow\delta}, \atom{}{\tau}) + \dist_{int}(\zeta^{[0,\nu-\delta-\tau]}_{\rightarrow\delta+\tau}, \phi_2) \notag\\
			&\quad -\dist_{int}(\zeta^{[0,\tau]}, 
			\atom{}{\tau})-\dist_{int}(\zeta^{[0,\nu-\tau]}_{\rightarrow \tau}, \phi_2)\big| \notag\\
			&\leq \big|\dist_{int}(\zeta^{[0,\tau]}_{\rightarrow\delta}, \atom{}{\tau}) -\dist_{int}(\zeta^{[0,\tau]}, \atom{}{\tau})\big| \notag\\
			&\quad +\big|\dist_{int}(\zeta^{[0,\nu-\delta-\tau]}_{\rightarrow\delta+\tau}, \phi_2)- \dist_{int}(\zeta^{[0,\nu-\tau]}_{\rightarrow \tau}, \phi_2)\big| \notag\\
			&\leq C\tau\delta+ (k-1)C\delta(\nu-\tau) \tag{I.H.}\\
			&\leq kC\nu\delta \notag
		\end{align}
		
		\item \underline{$\phi=\atom{}{\tau}\U\phi_2$}:\ \ 
		Let's assume that the optimal time to switch for $\zeta^{[{0,\nu}]}$ and $\zeta^{[0,\nu-\delta]}_{\rightarrow\delta}$ are $t$ and $t'$ respectively.
		There are two cases here.
		\begin{enumerate}
			\item If $\dist_{max}(\zeta^{[0,\nu-\delta]}_{\rightarrow\delta}, \phi)>\dist_{max}(\zeta^{[0,\nu]}, \phi)$ then since $t'$ is not optimal for the second term 
			\begin{align} 
				&\big|\dist_{int}(\zeta^{[0,\nu-\delta]}_{\rightarrow\delta}, \phi)-\dist_{int}(\zeta^{[0,\nu]}, \phi)\big| \notag \\
				&= \big|\dist_{int}(\zeta^{[0,t']}_{\rightarrow\delta}, \atom{}{t'}) + \dist_{int}(\zeta^{[0,\nu-\delta-t']}_{\rightarrow\delta+t'}, \phi_2) \notag\\
				&\quad - \dist_{int}(\zeta^{[0,t]}, 
				\atom{}{t}) - \dist_{int}(\zeta^{[0,\nu-t]}_{\rightarrow t}, \phi_2)\big| \notag\\
				&\leq \big|\dist_{int}(\zeta^{[0,t']}_{\rightarrow\delta}, \atom{}{t'}) + \dist_{int}(\zeta^{[0,\nu-\delta-t']}_{\rightarrow\delta+t'}, \phi_2) \notag\\
				&\quad - \dist_{int}(\zeta^{[0,t']}, 
				\atom{}{t'}) - \dist_{int}(\zeta^{[0,\nu-t']}_{\rightarrow t'}, \phi_2)\big| \notag\\
				&\leq \big|\dist_{int}(\zeta^{[0,t']}_{\rightarrow\delta}, \atom{}{t'}) - \dist_{int}(\zeta^{[0,t']}, 
				\atom{}{t'}) \big| \notag\\
				&\quad + \big|\dist_{int}(\zeta^{[0,\nu-\delta-t']}_{\rightarrow\delta+t'}, \phi_2) - \dist_{int}(\zeta^{[0,\nu-t']}_{\rightarrow t'}, \phi_2)\big| \notag\\			
				&\leq C\nu\delta + (k-1)C\nu\delta \tag{using I.H.} \\
				& = kC\nu\delta \notag 
			\end{align}
			\item Otherwise, if $\dist_{int}(\zeta^{[0,\nu-\delta]}_{\rightarrow\delta}, \phi)\leq\dist_{int}(\zeta^{[0,\nu]}, \phi)$ then since $t$ is not optimal for the first term
			\begin{align}
				&\big|\dist_{int}(\zeta^{[0,\nu-\delta]}_{\rightarrow\delta}, \phi)-\dist_{int}(\zeta^{[0,\nu]}, \phi)\big| \notag \\
				&= \big|\dist_{int}(\zeta^{[0,t']}_{\rightarrow\delta}, \atom{}{t'}) + \dist_{int}(\zeta^{[0,\nu-\delta-t']}_{\rightarrow\delta+t'}, \phi_2) \notag\\
				&\quad - \dist_{int}(\zeta^{[0,t]}, 
				\atom{}{t}) - \dist_{int}(\zeta^{[0,\nu-t]}_{\rightarrow t}, \phi_2)\big| \notag\\
				&\leq \big|\dist_{int}(\zeta^{[0,t]}_{\rightarrow\delta}, \atom{}{t}) + \dist_{int}(\zeta^{[0,\nu-\delta-t]}_{\rightarrow\delta+t}, \phi_2) \notag\\
				&\quad - \dist_{int}(\zeta^{[0,t]}, 
				\atom{}{t}) - \dist_{int}(\zeta^{[0,\nu-t]}_{\rightarrow t}, \phi_2)\big| \notag\\
				&\leq \big|\dist_{int}(\zeta^{[0,t]}_{\rightarrow\delta}, \atom{}{t}) - \dist_{int}(\zeta^{[0,t]}, 
				\atom{}{t}) \big| \notag\\
				&\quad + \big|\dist_{int}(\zeta^{[0,\nu-\delta-t]}_{\rightarrow\delta+t}, \phi_2) - \dist_{int}(\zeta^{[0,\nu-t]}_{\rightarrow t}, \phi_2)\big| \notag\\			
				&\leq C\nu\delta + (k-1)C\nu\delta \tag{using I.H.} \\
				& = kC\nu\delta \notag 
			\end{align}
		\end{enumerate}

	\end{itemize}
	
\end{proof}

\begin{theorem}
	\label{app:thm:main-int}
	For an LC trajectory $\zeta^{[0,\nu]}$, an FPL formula $\phi^{[0,v]}$ in LNF and $0<\delta <h_{min}(\phi)$, then $\exists C_1,C\in\reals^+$ such that
	\[|\tilde{\dist}_{int}(\zeta^{[0,\nu]}, \phi) - \dist_{int}(\zeta^{[0,\nu]}, \phi)| \leq k(C_1+kC\nu)\delta\] 
	where $k$ is the number of concatenation or until operators in $\phi$.
\end{theorem}

\begin{proof}
	(Sketch.) 
	The proof works by applying structural induction on $\phi$, and the main case is the until case, which also requires that a trajectory shifted slightly backwards has similar distance as the original one.
	In the proof, we show the following values of the constants work
	\[C_1 = \max_{\bm{\pi}\in atoms(\phi)} \big(d(\zeta(0), \atom{}{\tau}(0)) + (K_{d,\bm{\pi}} + K_\zeta K_{d,\zeta})\cdot\min\{\nu,\tau\}\big)\] \[C_2=\max_{\bm{\pi}\in atoms(\phi)} \big(K_{d,\zeta}K_{\zeta}
	+ K_{d,\pi}\big) \text{ and } C=K_{d,\zeta} K_{\zeta}\]

	(Complete Proof.)
	To prove the correctness, we use induction on the structure of the formula $\phi$.
	The base case is trivial since we assume that we can compute the exact distance for an atom.
	Induction hypothesis is that all the sub-formulas of smaller size can be approximated for any trajectory.
	\begin{itemize}
		\item \underline{$\phi = \phi_1 \lor \phi_2$}: \\
		Using induction hypothesis, we get
		\[|\tilde{\dist}_{int}(\zeta^{[0,\nu]}, \phi_1)-\dist_{int}(\zeta^{[0,\nu]}, \phi_1)| \leq k(C_1+kC\nu)\delta\]
		and 
		\[|\tilde{\dist}_{int}(\zeta^{[0,\nu]}, \phi_2)-\dist_{int}(\zeta^{[0,\nu]}, \phi_2)| \leq k(C_1+kC\nu)\delta\]
		Now, by definition,
		\[\tilde{\dist}_{int}(\zeta^{[0,\nu]}, \phi) = \min\{\tilde{\dist}_{int}(\zeta^{[0,\nu]}, \phi_1), \tilde{\dist}_{int}(\zeta^{[0,\nu]}, \phi_2)\}\]
		This gives us 
		\[|\tilde{\dist}_{int}(\zeta^{[0,\nu]}, \phi)-\dist_{int}(\zeta^{[0,\nu]}, \phi)| \leq k(C_1+kC\nu)\delta\]
		
		\item \underline{$\phi=\atom{}{\tau} \cdot \phi_2$}: \\
		If $\nu < \tau$ then by definition the semantics is $\infty$, so we assume $\nu > \tau$.
		Now, by induction hypothesis we get 
		\[|\tilde{\dist}(\zeta^{[0,\nu-\tau]}_{\rightarrow \tau}, \phi_2)-\dist(\zeta^{[0,\nu-\tau]}_{\rightarrow \tau}, \phi_2)| \leq k(C_1+kC\nu)\delta\]
		
		Since we assume that we can calculate the semantics for atoms exactly, we get
		\begin{align}
			&|\tilde{\dist}_{int}(\zeta^{[0,\nu]}, \phi)-\dist(\zeta^{[0,\nu]}, \phi)| \notag\\
			=&|\tilde{\dist}_{int}(\zeta^{[0,\tau]}, \atom{}{\tau}) + \tilde{\dist}_{int}(\zeta^{[0,\nu-\tau]}_{\rightarrow \tau}, \phi_2) \notag \\ -&\dist_{int}(\zeta^{[0,\tau]}, \atom{}{\tau}) - \dist_{int}(\zeta^{[0,\nu-\tau]}_{\rightarrow \tau}, \phi_2)| \notag\\
			= |&\tilde{\dist}_{int}(\zeta^{[0,\nu-\tau]}_{\rightarrow \tau}, \phi_2)-\dist_{int}(\zeta^{[0,\nu-\tau]}_{\rightarrow \tau}, \phi_2)| \notag \\
			\leq\ \ &k(C_1+kC\nu)\delta \tag{using I.H.}
		\end{align}

		\item \underline{$\phi=\atom{}{\tau} \U \phi_2$}: \\
		If the value of semantics is infinite then the algorithm is trivially correct since all the computations would lead to $\infty$.
		So, let's assume the value is finite and that the optimal time to switch for this until was $t$. Then,
		\[\dist_{int}(\zeta^{[0,\nu]}, \phi) = \dist_{int}(\zeta^{[0,t]}, \atom{}{t})+\dist_{int}(\zeta^{[0,\nu-t]}_{\rightarrow t}, \phi_2)\]
		Let $t'<t$ be the closest point in discretization to $t$, (which also implies $t-t'\leq \delta$)
		\[\tilde{\dist}_{int}(\zeta^{[0,\nu]}, \phi) \leq \tilde{\dist}_{int}(\zeta^{[0,t']}, \atom{}{t'})+\tilde{\dist}_{int}(\zeta^{[0,\nu-t']}_{\rightarrow t'}, \phi_2)\]
		
		Now, since $\dist(\zeta^{[0,\nu-t]}_{\rightarrow t}, \phi_2)<\infty$, we use Lemma~\ref{lem:until-free-int} to get
		\begin{equation}
			|\dist_{int}(\zeta^{[0,\nu-t']}_{\rightarrow t'}, \phi_2) - \dist_{int}(\zeta^{[0,\nu-t]}_{\rightarrow t}, \phi_2)| \leq (k-1)C\nu\delta \label{eq:int-main-thm-1}
		\end{equation}
		Also, from Lemma~\ref{lem:cdelta-int} and the assumption that computing semantics for atoms is exact, we get 
		\begin{equation}	
			|\tilde{\dist}_{int}(\zeta^{[0,t']}, \atom{}{t'}) - \dist_{int}(\zeta^{[0,t]}, \atom{}{t})| \leq C_1\delta \label{eq:int-main-thm-2}
		\end{equation}
		
		Since there can be only $k-1$ until operators in $\phi_2$, induction hypothesis gives us 
		\begin{align}
			|\tilde{\dist}_{int}(\zeta^{[0,\nu-t']}_{\rightarrow t'}, &\phi_2) - \dist_{int}(\zeta^{[0,\nu-t]}_{\rightarrow t}, \phi_2)| \notag\\ 
			&\leq (k-1)(C_1+(k-1)C\nu)\delta \label{eq:int-main-thm-3}
		\end{align}
		Now,
		\begin{align}
			|\tilde{\dist}_{int}&(\zeta^{[0,\nu]}, \phi)-\dist_{int}(\zeta^{[0,\nu]}, \phi)| \notag\\
			= |&\tilde{\dist}_{int}(\zeta^{[0,t']}, \atom{}{t'}) + \tilde{\dist}_{int}(\zeta^{[0,\nu-t']}_{\rightarrow t'}, \phi_2) \notag\\
			- &\dist_{int}(\zeta^{[0,t]}, \atom{}{t}) - \dist_{int}(\zeta^{[0,\nu-t]}_{\rightarrow t}, \phi_2)| \notag\\
			\leq |&\tilde{\dist}_{int}(\zeta^{[0,t']}, \atom{}{t'}) - \dist_{int}(\zeta^{[0,t]}, \atom{}{t})| \notag\\
			+ &|\tilde{\dist}_{int}(\zeta^{[0,\nu-t']}_{\rightarrow t'}, \phi_2) - \dist_{int}(\zeta^{[0,\nu-t]}_{\rightarrow t}, \phi_2)| \notag 
		\end{align}
		Using equation \ref{eq:int-main-thm-1}, \ref{eq:int-main-thm-2}, and \ref{eq:int-main-thm-3} we get 
		\begin{align}
			&\leq C_1\delta + (k-1)(C_1+(k-1)C\nu)\delta + (k-1)C\nu\delta\notag\\
			&\leq kC_1\delta + k^2C\nu\delta \notag \\
			&=k(C_1+kC\nu)\delta \notag
		\end{align}

	\end{itemize}
	
\end{proof}

\section{Discussion}

\subsection{Model Checking}
\label{app:mc}
We can extend the algorithm to compute semantics to model check a restricted class of hybrid systems. 
If the system has finitely many paths then, it is possible to compute semantics for all different paths and check whether all paths satisfy the specification or not.

\begin{definition}
	A hybrid automaton is defined as an 8-tuple $\mathcal{H} = (L, \Sigma, \Delta, S, init, inv, flow, jump)$, where 
	\begin{itemize}
			\item $L$ is a finite set of locations,
			\item $\Sigma$ is a finite set of actions,
			\item $\Delta \subseteq L \times \Sigma \times L$ is a set of labeled edges that represent discrete changes in the control mode,
			\item $S=\{s_1, s_2, \dots, s_n\}$ is a set of real-valued variables, and $\dot{S}=\{\dot{s_1}, \dot{s_2}, \dots, \dot{s_n}\}$ represents the set of first derivative of $S$ w.r.t time and $S'=(s_1', s_2', \dots, s_n')$ are primed variables that represent updates after discrete jumps,
			\item $init$, $inv$ and $flow$ assign three predicates to each location. 
			The initial condition $init(l)$ is a predicate over $S$ which represents the possible values of $S$ when the system starts in location $l$.
			The invariant $inv(l)$ is also a predicate over $S$ which represents the constraints system should satisfy when in location $l$.
			The flow condition $flow(l)$ is a predicate over $S\cup S'$ describes how the system evolves when in location $l$.
			\item $jump(e)$ assigns a predicate over $S\cup S'$ to edge $e\in E$ and it represents whether taking that edge is feasible or not.
		\end{itemize}
\end{definition}

A run of the Hybrid automaton is a trajectory $\zeta_{l_1}^{[0,t_1]}\cdot \zeta_{l_2}^{[0,t_2]} \dots\zeta_{l_m}^{[0,t_m]}$, where each $\zeta_{l_i}$ is spent in the vertex $l_i$ for time $t_i$ such that 
\begin{itemize}
	\item initial point of $\zeta_{l_i}$ satisfies the predicate $init(l_i)$
	\item $\zeta_{l_i}$ satisfies the predicates $inv(l_i)$ and $flow(l_i)$
	\item $\zeta_{l_i}$ satisfies the predicate $jump(e)$ where $e$ is an edge from $l_i$ to $l_{i+1}$.
\end{itemize}

To extend our algorithm directly to model check hybrid systems, we need some assumptions on the models. 
We describe these assumptions next and also discuss why they are required here.
\begin{enumerate}
	\item the evolution is always LC: this assumption is required since we need the trajectories (or runs) to be LC
	\item there are no updates during the discrete changes i.e. $S' =S$: this assumption is in line with the first assumptions, since we don't want any discontinuities in the trajectories.
	\item the edges can only be taken at discrete times when the $jump$ predicate is true: this assumption is required so that the system does not have infinite behaviors and our algorithm can directly work here. However, it might be possible to lift this assumption (or weaken it) by finding smarter algorithms that take the structure of the system into account. It might also be possible to discretize the time and give a bound on the error, as we did in this paper for until operators. All this can be part of future work.
	\item the evolution of the system inside a location is fixed:
	This assumption is similar to assumption 3 since it also restricts the behavior of the system to finitely many runs, making it feasible for our algorithm to work off the shelf. However, we might be able to lift this assumption in the future too.
	
\end{enumerate}
These restrictions reduce the hybrid system into a finite set of trajectories, which are LC, hence they can be enumerated and model checked separately by computing their semantics.

\subsection{Outlooks to Sampling and Synthesis}
\label{app:sampling_synthesis}
While we learn a representation of the nominal behavior and a function capturing  the distance from it, some applications may require producing behaviors close but not equal to the nominal behavior.
For instance, a bit of randomness may be desirable in the movements, e.g.\ to naturally break ties, or slight deviation may be desirable, e.g.\ in order to simulate human-like imprecision.
Despite sounding naturally, the task is not well defined and a number of formulations can be suggested.
We illustrate some of the open points below:
\begin{itemize}
	\item Consider the most naive view, where a point is sampled from the distribution corresponding to the path at each time step. 
	First, not all the sampled points might make sense.
	For example, consider signals $s$ and $v$ for position and velocity, respectively. 
	Since they are in physical reality bound together by $v$ being derivative of $s$, these two cannot be sampled independently when sampling a continuation of a path.
	Consequently, knowledge of the dynamics needs to enter the picture when we learn the path distribution for the sake of sampling.
	\item Second, even when sampling only the ``independent variables'', the path may turn out very sinuous, not reflecting real trajectories.
	Here, for instance, a Bayesian approach might be more useful where you sample the next point based on the current point. 
	For this, knowledge of such conditional probability distributions is required to be known a priori, which may depend on additional variables such as curvature. 
	These variables of interest may also be learnt, so that the sampling is more useful.
	Again, this is additional knowledge, which depends on the particular setting.
	\item Third, a completely different approach to sampling might be desirable.
	For instance, consider a more high-level approach, where endpoints of an atom are sampled and the rest of the path should be ``sensibly'' completed so that, e.g., the integral of mean squared distance to the points on the nominal path is minimized, or the length of the path is minimized while keeping the distance below that of the sampled endpoints etc.
\end{itemize}

The key observation here is that the next point in the trajectory on current values of the signals which is not present in the FPL formula. 
Also, the system dynamics might not allow for completely random changes thereby restricting the sampling of the next point even more.

To conclude, finding a path which adheres to an FPL formula is not entirely trivial and can depend on the dynamics of the system.

\section{STL Learning With TelEx}\label{app:learn-exper}
In this section, we show the learning of STL formulae for our example of avoiding obstacle in a corridor using TeLEx~\cite{telex}.
There are in total 20 trajectories, 10 going from each side of the obstacle and each trajectory is 24 time steps long. 
Note that TeLEx uses a template based approach and we had to provide the templates. 
Each template has some parameters (followed by a question mark and a range, e.g. $a? -1;10$) and the tool tries to synthesize some good values for these parameters from this given range.
We tried three templates, starting with a most basic box to bound just the value of $x$, then bound the value of both $x, y$ for the whole horizon. Lastly, we tried to learn two boxes by splitting the horizon into two parts.

\begin{itemize}
	\item The fist template (shown below) only bounds the values of $x$ for the whole time horizon of 24 time steps.
		Here. time horizon is fixed and the only parameters are the bounds on $x$.
		\begin{equation}
			\G^{[0,24]}(x < a? -1;10) \land \G^{[0,24]}(x > b? -1;10)
			\label{eq:telex_1}
		\end{equation}
		The generated formula is 
		\begin{equation}
			\G^{[0,24]}(x < 3.16) \land \G^{[0,24]}(x > -0.35) \notag
		\end{equation}
		It says that for the whole horizon, the value of $x$ should stay within these bounds. We show it visually in Figure~\ref{fig:telex_1}, the bounds cover all of the signals and look very tight, as one would expect. All trajectories satisfy the formula here.
		\begin{figure}[ht]
			\centering
			\includegraphics[scale=0.5, trim={0 2cm 0 3cm}, clip]{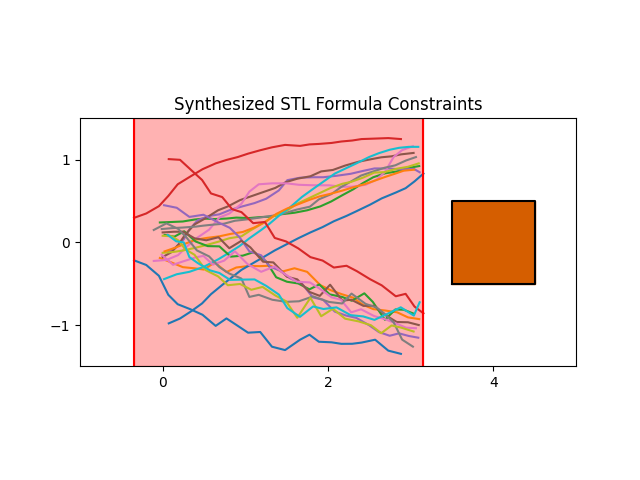}
			\caption{Bounds generated by TeLEx for the template in Equation~\ref{eq:telex_1}}
			\label{fig:telex_1}
		\end{figure}

	\item The second template (shown below) bounds both $x$ and $y$ for the whole horizon of 24 time steps. Again, the time horizon is fixed and the parameters are the bounds on the values of $x$ and $y$. 
		\begin{equation}
			\begin{split}
				\G^{[0,24]}(x < a? -1;10) \land \G^{[0,24]}(x > b? -1;10) \land\\ 
				\G^{[0,24]}(y < c? -10;10) \land \G^{[0,24]}(y > d? -10;10)		
			\end{split}
			\label{eq:telex_2}
		\end{equation}
		The synthesized formula is 
		\begin{equation}
			\begin{split}
				\G^{[0,24]}(x < 3.16) \land \G^{[0,24]}(x > -0.35) \land\\ 
				\G^{[0,24]}(y < 1.27) \land \G^{[0,24]}(y > -1.35)		
			\end{split} \notag
		\end{equation}
		This formula also says that for the whole horizon, the value of $x$ and $y$ should stay within these bounds. We show it visually in Figure~\ref{fig:telex_2}, the box cover all of the signals and again look very tight. Also, all trajectories satisfy the formula here.
		\begin{figure}[ht]
			\centering
			\includegraphics[scale=0.5, trim={0 2cm 0 3cm}, clip]{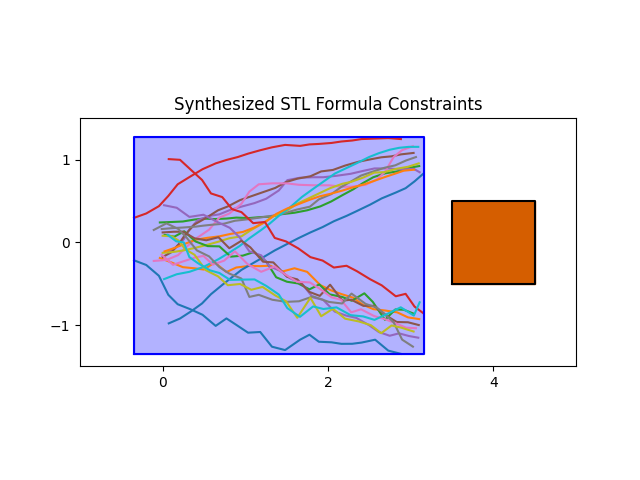}
			\caption{Bounds generated by TeLEx for the template in Equation~\ref{eq:telex_2}}
			\label{fig:telex_2}
		\end{figure}

	\item The third template (shown below) tries to split the boxes into two parts and asks the learning algorithm to generate the right time step to split. 
	Now, the parameters are the bounds for the two boxes as well as the ending time of the first box and starting time of the second box. The end time for the first box ranges from 12 to 18 steps and for the second box, beginning time ranges from 15 to 20 steps.
		\begin{equation}
			\begin{split}
				\G^{[0, a? 12;18]}\big((x < b? -1;10) \land (x > c? -1;10) \land  \\
				(y < d? -10;10) \land (y > e? -10;10)\big) \land  \\
				\G^{[f? 15;20,25]}\big((x < g? -1;10) \land (x > h? -1;10) \land  \\
				(y < i? -10;10) \land (y > j? -10;10)\big)
			\end{split}
			\label{eq:telex_3}
		\end{equation}
		The generated formula is 
		\begin{equation}
			\begin{split}
				\G^{[0, 18]}\big((x < 1.69) \land (x > -0.35) \land \\
				(y < 1.15) \land (y > -1.10)\big) \land \\
				\G^{[15,25]}\big((x < 3.16) \land (x > 2.16) \land \\ 
				(y < 1.27) \land (y > -1.35)\big)
			\end{split} \notag
		\end{equation}
		It is visually shown in Figure~\ref{fig:telex_3}. The formula says that for the first 18 steps stay in the first box (blue) and for steps from 15-25, stay in the second box (red).
		\begin{figure}
			\centering
			\includegraphics[scale=0.5, trim={0 2cm 0 3cm}, clip]{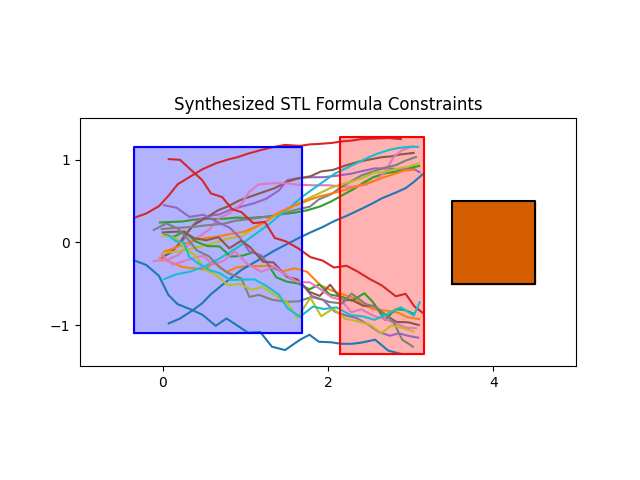}
			\caption[]{Showing the boxes used in STL formula generated by TeLEx for the template in Equation~\ref{eq:telex_3}.}
			\label{fig:telex_3}
		\end{figure}
	
\end{itemize}

\begin{remark}
	To our surprise, the STL formula generated with the third template (Equation~\ref{eq:telex_3}) is not satisfied by any of the trajectories (see Figure~\ref{fig:telex_3}) since the time horizon for the two boxes overlaps but the boxes themselves are disjoint, i.e the first box needs to be satisfied for 0-18 time steps and second box needs to be satisfied for 15-25 time steps.
	Note that, with the template, it was possible to generate a formula that says stay in the first box for first 12 time steps and stay in the second box for 20-25 time steps, which could be easily satisfied by the trajectories with the generated boxes. 
\end{remark}

\section{FPL Learning Algorithm}
\label{app:learning}
In this section, we provide additional details on the learning algorithm for Fuzzy Path Logic (FPL) presented in Section~\ref{sec:learning}.

\subsection{Similar Atom Merging}
\label{appendix:similar-atom-merging}
This process identifies and merges similar atoms within the learned DAG. This is useful for reducing redundancy to achive a more compact FPL. This can also clean noisy atoms that may arise due to data imperfections as can be seen in Figure~\ref{fig:coffee_dag}. 

Examples of this can be seen in section \ref{appendix:learning-examples}.

The algorithm assigns a similarity score to each pair of atoms based on their KL divergence and time mismatch. If the similarity score exceeds a predefined threshold, the atoms are merged by appropriate combined value of their attributes and the time intervals are adjusted accordingly.

\subsection{DAG Simplification Algorithm}
\label{appendix:dag-simplification}
This simplification algorithm operates on a DAG with a single \texttt{Start} node and a single \texttt{End} node. It iteratively applies three operations: \texttt{Extend}, \texttt{Join}, and \texttt{Split}. The process continues until the graph is reduced to a form where it has exactly three nodes: \texttt{Start}, an internal node, and \texttt{End}.

\subsubsection{Operations}
\begin{itemize}
    \item \textbf{Extend}: Merges two nodes \(u\) and \(v\) if \(u\) is the only parent of \(v\) and \(v\) is the only child of \(u\). This operation reduces the number of nodes by one.
    \item \textbf{Join}: Merges two internal nodes \(u\) and \(v\) if they have the same parents and children. This operation also reduces the number of nodes by one.
    \item \textbf{Split}: Splits an internal node \(u\) into two nodes \(u_1\) and \(u_2\), redistributing its children. This operation increases the number of nodes by one.
\end{itemize}

\begin{algorithm}[ht]
\caption{\textsc{SimplifyDAG}$(G)$}
\begin{algorithmic}[1]
\State \textbf{Input:} Directed Acyclic Graph $G$ with \texttt{Start}, \texttt{End}
\While{$\text{SIZE}(G) > 3$}
    \If{\Call{TryExtend}{$G$}} \State \textbf{continue} \EndIf
    \If{\Call{TryJoin}{$G$}} \State \textbf{continue} \EndIf
    \If{\Call{TrySplit}{$G$}} \State \textbf{continue} \EndIf
    \State \textbf{break}
\EndWhile
\end{algorithmic}
\end{algorithm}

\begin{algorithm}[ht]
\caption{\textsc{TryExtend}$(G)$}
\begin{algorithmic}[1]
\ForAll{edges $(u, v)$ in $G$}
    \If{$u \ne v$ and $u \ne$ \texttt{Start} and $v \ne$ \texttt{End}}
        \If{$\text{Children}(u) = \{v\}$ and $\text{Parents}(v) = \{u\}$}
            \State \Call{MergeNodes}{$u, v, \mathbf{U}$}
            \State \Return true
        \EndIf
    \EndIf
\EndFor
\State \Return false
\end{algorithmic}
\end{algorithm}

\begin{algorithm}[ht]
\caption{\textsc{TryJoin}$(G)$}
\begin{algorithmic}[1]
\ForAll{distinct internal nodes $u, v$ in $G$}
    \If{$\text{Parents}(u) = \text{Parents}(v)$ and $\text{Children}(u) = \text{Children}(v)$}
        \State \Call{MergeNodes}{$u, v, \mathbf{\vee} $}
        \State \Return true
    \EndIf
\EndFor
\State \Return false
\end{algorithmic}
\end{algorithm}

\begin{algorithm}[ht]
\caption{\textsc{TrySplit}$(G)$}
\begin{algorithmic}[1]
\ForAll{internal nodes $u$ in $G$}
    \If{$|\text{Children}(u)| > 1$}
        \State \Call{SplitNode}{$u$}
        \State \Return true
    \EndIf
\EndFor
\State \Return false
\end{algorithmic}
\end{algorithm}

\begin{algorithm}[ht]
\caption{\textsc{MergeNodes}$(u, v, Op)$}
\begin{algorithmic}[1]
\State Create node $w$ with merged content of $u$, $v$ according to $Op$
\ForAll{$p$ in $\text{Parents}(u)$} \State \Call{AddEdge}{$p, w$} \EndFor
\ForAll{$c$ in $\text{Children}(v)$} \State \Call{AddEdge}{$w, c$} \EndFor
\State Remove $u$ and $v$ from $G$
\end{algorithmic}
\end{algorithm}

\begin{algorithm}[ht]
\caption{\textsc{SplitNode}$(u)$}
\begin{algorithmic}[1]
\State Create nodes $u_1$, $u_2$ with same content as $u$
\ForAll{$p$ in $\text{Parents}(u)$}
    \State \Call{AddEdge}{$p, u_1$}, \Call{AddEdge}{$p, u_2$}
\EndFor
\State Partition $\text{Children}(u)$ into $C_1, C_2$ such that $|C_1| = 1$ and $|C_2| = |\text{Children}(u)| - 1$ 
\ForAll{$c \in C_1$} \State \Call{AddEdge}{$u_1, c$} \EndFor
\ForAll{$c \in C_2$} \State \Call{AddEdge}{$u_2, c$} \EndFor
\State Remove $u$ from $G$
\end{algorithmic}
\end{algorithm}

\subsubsection{Proof of Correctness of DAG Simplification Algorithm}

To prove the DAG simplification algorithm terminates. We define a potential function that strictly decreases with each operation, ensuring termination.

We define a potential function \(\Phi(G)\) as follows:
\[
\Phi(G) = v(G) + 2 \cdot \psi(G)
\]
where:
\begin{itemize}
    \item \(v(G)\) is the number of internal vertices in \(G\).
    \item \(\psi(G)\) is the \textit{split potential} of \(G\).
\end{itemize}
\(\Phi(G)\) is a non-negative function that decreases with each transformation.

We define the split potential \(\psi(G)\) recursively:
\begin{itemize}
    \item \(\psi(\texttt{start}) = \psi(\texttt{end}) = 0\)
    \item For an internal node \(u\),
    \[
    \psi(u) = \left( \sum_{c \in \text{Children}(u)} \psi(c) + 1 \right) - 1
    \]
    \item The total split potential is:
    \[
    \psi(G) = \sum_{u \in G} \psi(u)
    \]
\end{itemize}

The idea behind these definition is the split potential \(\psi(G)\) captures the maximum number of splits that can be performed on the graph \(G\). \\

To see check this, let \(u\) be an internal node with children \(c_1, c_2, \dots\). Assume each child \(c_i\) has split potential \(\psi(c_i)\). After all splits, \(u\) will have \(\sum_{i} (\psi(c_i) + 1)\) children. Since we can only keep splitting a node till all its children are partitioned, therefore,
\[
\psi(u) = \left( \sum_i (\psi(c_i) + 1) \right) - 1 = \sum_i \psi(c_i) + (k - 1)
\]
where \(k\) is the number of children. This matches the defined formula.

Now we need to check the effects of Operations on Potential We show that \(\Phi(G)\) is strictly decreasing under each of the three operations: \texttt{Extend}, \texttt{Join}, and \texttt{Split}.

\paragraph{General Note:}
We only consider changes in split potential of parents of a modified node. Since \(\psi(\cdot)\) is an increasing function of its children's potential, ancestor nodes cannot have increased potential if their children's potential does not increase.\\

\textbf{\texttt{Extend} Operation}

Let \(u\) and \(v\) be two nodes such that:
\[
\text{Children}(u) = \{v\}, \quad \text{Parents}(v) = \{u\}
\]
Assume \(\psi(v) = x\). Then \(\psi(u) = x\). \\

After merging \(u\) and \(v\) into node \(w\):

\[\psi(w) = \psi(u) = x\]
All parents of \(u\) now point to \(w\), and their \(\psi\) values are unchanged. \\
So new split potential of graph \(G\) is:
\[
\psi(G)_{w} = \psi(G)_{u,v} - x \leq \psi(G)_{u,v}
\]
The number of internal nodes decreases by 1.

Thus,
\[
\Phi(G)_w = v(G)_w + 2 \cdot \psi(G)_w < v(G)_{u,v} + 2 \cdot \psi(G)_{u,v} = \Phi(G)_{u,v}
\]
The potential function \(\Phi(G)\) strictly decreases under the \texttt{Extend} operation.\\

\textbf{\texttt{Join} Operation}

Let \(u, v\) be nodes with the same children and parents \(p_1, p_2, \dots\) , and \(\psi(u) = \psi(v) = x\). After merging into \(w\):

\(\psi(w) = x\) \\
For each parent \(p_i\), before merging:
    \[
    \psi(p_i) = 2x + 1 + \sum_j (\psi(c_{i,j}) + 1)
    \]
where \(c_{i,j}\) are the children of \(p_i\) other than \(u\) and \(v\). \\
After merging:
    \[
    \psi(p_i) = x + \sum_j \psi(c_{i,j})
    \]
So, the decrease in each parent’s potential is \(x + 1\). Number of internal nodes decreases by 1 

Therefore, \(\Phi(G)\) strictly decreases under the \texttt{Join} operation \\

\textbf{\texttt{Split} Operation} \\
Let a node \(u\) with children \(\{c_1, c_2, \dots, c_n\}\) where \(n > 1\) be split into \(u_1\) and \(u_2\). Let \(\psi(c_i) = x_i\).

Before splitting:
\[
\psi(u) = \sum_{i=1}^n (x_i + 1) - 1
\]

After splitting:
\[
\psi(u_1) = x_1,\quad \psi(u_2) = \sum_{i=2}^{n} (x_i + 1) - 1
\]
\[
\Rightarrow \psi(u_1) + \psi(u_2) = \psi(u) - 1
\]

Parents of \(u\) now point to both \(u_1\) and \(u_2\).
Their total potential remains unchanged as: 
\[
\begin{aligned}
    \psi(p_i)_{u} &= (\psi(u)+1) + \sum_{j=1}^{n} (\psi(c_j) + 1) -1 \\
    &= (\psi(u_1) + \psi(u_2)+2) + \sum_{j=1}^{n} (\psi(c_j) + 1) -1 \\
    &= \psi(p_i)_{u_1,u_2}
\end{aligned}
\]

Hence, the split operation does not change the potential of parents. Since the children's potentials remain unchanged only change is split potential is due to the split node itself. Hence 
\[
\begin{aligned}
    \psi(G)_{\text{after}} &= \psi(u_1) + \psi(u_2) + \sum_{n \in G - \{u_1,u_2\}} \psi(n) \\
    &=  \psi(u) - 1 + \sum_{n \in G - \{u_1,u_2\}} \psi(n)  \\
    &= \psi(G)_{\text{before}} - 1 \\
\end{aligned}
\]

Since the number of internal nodes only increases by 1
\[
\begin{aligned}
\psi(G)_{\text{after}} &= v(G)_{\text{before}} + 1 + 2\cdot(\psi(G)_{\text{before}} - 1) \\
&= v(G)_{\text{before}} + 2\cdot\psi(G)_{\text{before}} - 1 \\
&= \Phi(G)_{\text{before}} - 1
\end{aligned}
\]

Hence, \(\Phi(G)\) strictly decreases under \texttt{Split} operation.

The least potential function \(\Phi(G)\) is 1, which occurs when the graph is reduced to the form \texttt{Start} $\rightarrow$ \texttt{Node} $\rightarrow$ \texttt{End}. Since \(\Phi(G)\) is strictly decreasing with each operation, the algorithm is guaranteed to terminate.

Example of DAG Simplification is shown in Figure \ref{fig:dag-simplification-example}.

\begin{figure}[H]
    \centering
    \includegraphics[width=0.5\textwidth]{"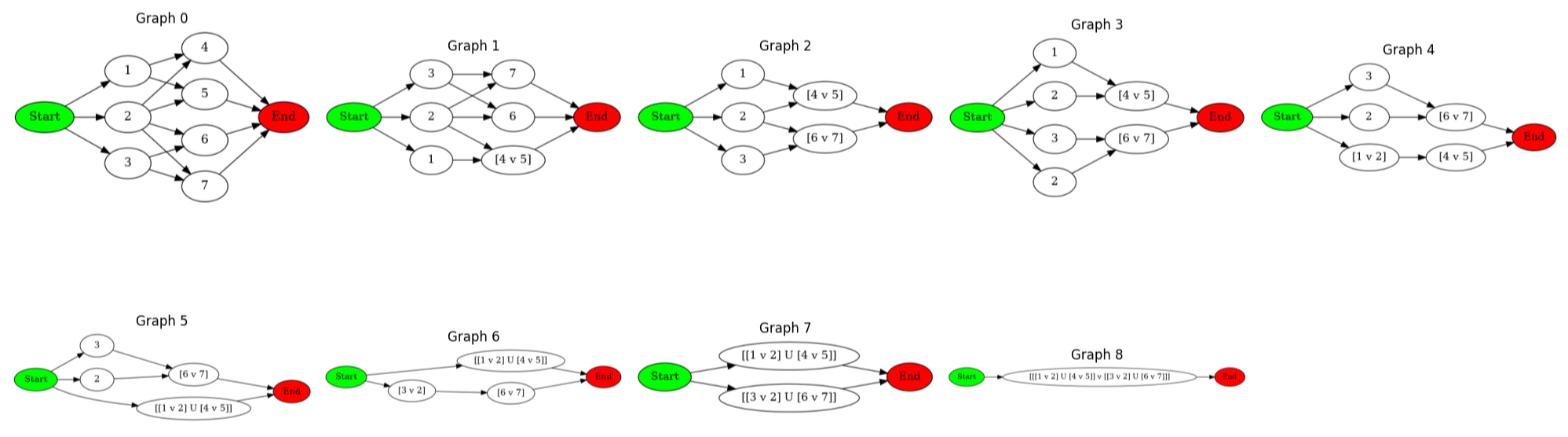"}
    \caption{An example of the DAG simplification algorithm.}
    \label{fig:dag-simplification-example}
\end{figure}

\subsection{Experiments with the Learning Algorithm}
We present 4 additional examples of learning FPL formulae from trajectory data using our learning algorithm.

\paragraph{Example 1}
    Consider the 1-dimensional trajectories shown in Figure~\ref{fig:learning_example_traj1}. Initially, the paths are around 0; around time $t=50$ they move at approximately unit speed for about 50 seconds, and then remain near the final position.
    We obtain a learned DAG as shown in Figure~\ref{fig:learning_example_dag1}.
    This DAG reveals three phases: an initial stationary phase, a moving phase, and a final stationary phase, represented by atom nodes connected by \textit{Until} operators. We can also visualize the learned atoms in Figure~\ref{fig:learning_example_fpl_atoms1}.
    After simplification, we obtain
    \([X_0 = -0.01 + 0.0\,t] \; \U_{(53, 204)} \; \big[[ X_0 = 4.82 + 1.0\,t] \; \U_{(50, 56)} \; [X_0 = 51.07 + 0.0\,t]\big]\).

\paragraph{Example 2}
    Consider the 1-dimensional trajectories shown in Figure~\ref{fig:learning_example_traj2}. The paths are initially around 0; at time $t=50$, some trajectories move with speed $1$ for about 150 seconds, while others move with speed $-1$ for about 150 seconds.
    We obtain a learned DAG as shown in Figure~\ref{fig:learning_example_dag2}.
    This DAG shows a common resting phase followed by a branching into two distinct movement patterns. We can also visualize the learned atoms in Figure~\ref{fig:learning_example_fpl_atoms2}.
    After simplification, we obtain 
    \([X_0 = 0.01 + 0.0\,t]_{(0.0, 57.0)} \; \U_{(50, 57)} \; \big( [X_0 = 3.02 + 1.01\,t]_{(0.0, 150.0)} \; \lor \; [X_0 = -3.16 - 0.99\,t]_{(0.0, 150.0)} \big)\).

\label{appendix:learning-examples}
\begin{figure}[ht]
    \centering
    \includegraphics[width=0.8\columnwidth]{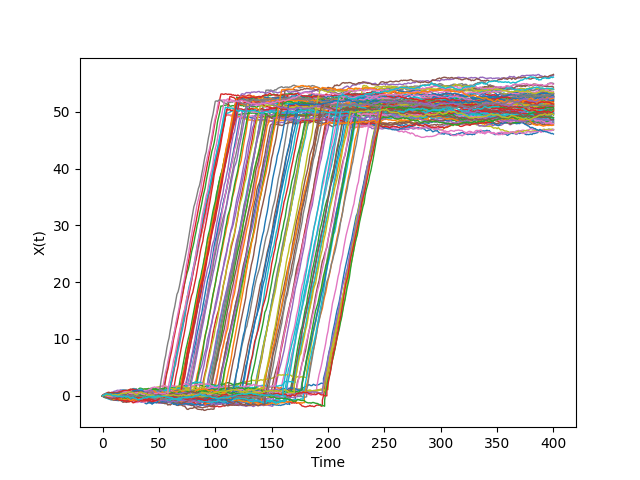}
    \caption{Example 1 trajectories. The $x$-axis here represents the time and the $y$-axis represents the position in 1D.}
    \label{fig:learning_example_traj1} 
\end{figure}

\begin{figure}[ht]
    \centering
    \includegraphics[width=1\columnwidth]{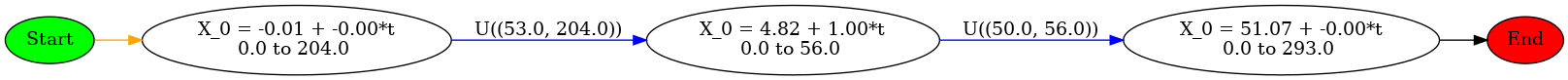}
    \caption{Example 1 learned DAG.}
    \label{fig:learning_example_dag1}
\end{figure} 

\begin{figure}[ht]
    \centering
    \includegraphics[width=0.8\columnwidth]{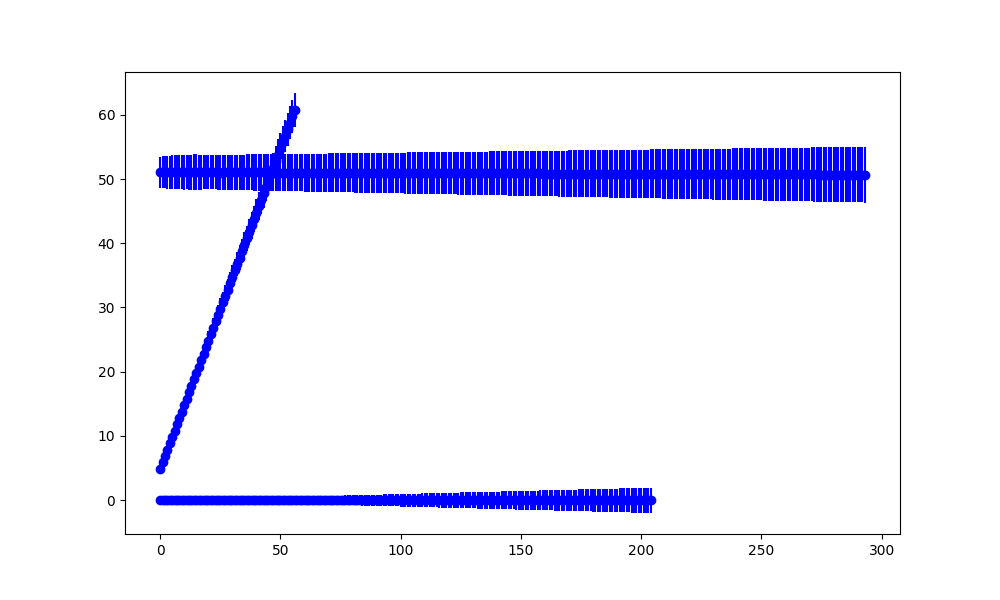}
    \caption{Example 1 FPL atoms. Since it's a 1D example, the deviation is only in y axis.}
    \label{fig:learning_example_fpl_atoms1} 
\end{figure}

\begin{figure}[ht]
    \centering
    \includegraphics[width=0.8\columnwidth]{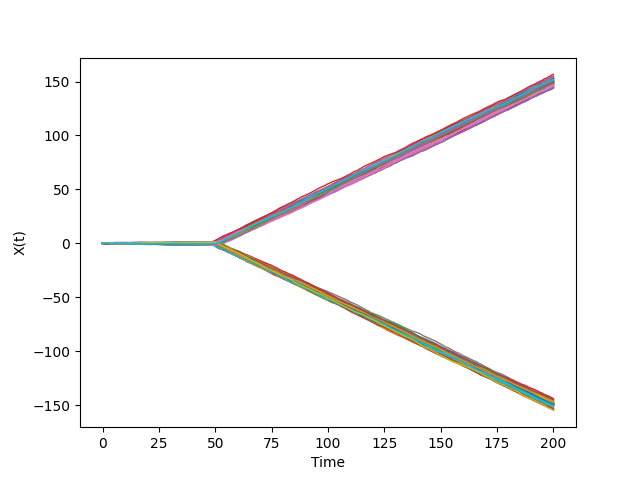}
    \caption{Example 2 trajectories.}
    \label{fig:learning_example_traj2} 
\end{figure}

\begin{figure}[ht]
    \centering
    \includegraphics[width=0.8\columnwidth]{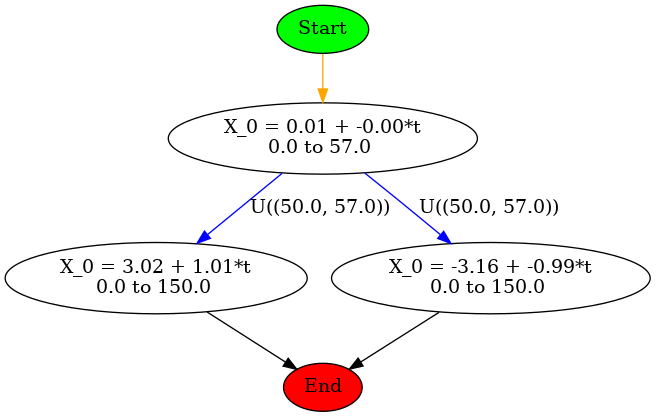}
    \caption{Example 2 learned DAG.}
    \label{fig:learning_example_dag2}
\end{figure}

\begin{figure}[ht]
    \centering
    \includegraphics[width=0.8\columnwidth]{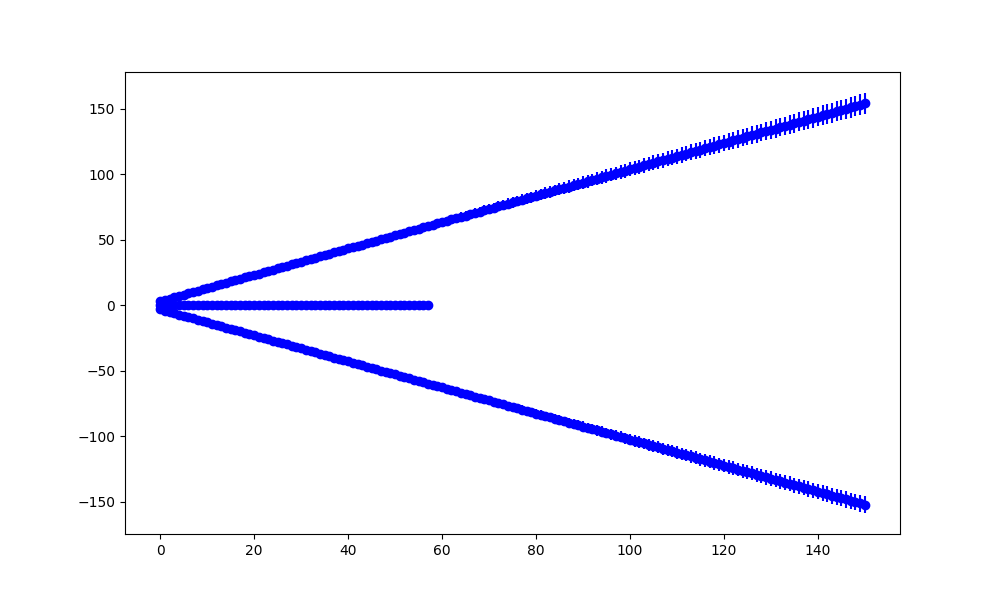}
    \caption{Example 2 FPL atoms.}
    \label{fig:learning_example_fpl_atoms2} 
\end{figure}

\paragraph{Example 3} 
Consider the set of trajectories shown in Figure~\ref{fig:learning_example_traj1_appendix}. In these trajectory, the agent starts at X=0 and remains there for 100 seconds. Then the agent either moves to around X = 50 or X = -50 and remain there for 50 seconds and then returns to X=0. And staying there for another 100 seconds.

The learned DAG before and after merging of similar atoms is shown in Figures~\ref{fig:learning_example_dag1_appendix} and \ref{fig:learning_example_dag_merged1_appendix} respectively.

As expected we can see the learned Atoms captures the diffrent aspects of the trajectory. We also observe that before merging we have multiple atoms corresponding to the final portion of the trajectory where the agent returns to X=0 and stays there for 100 seconds. These atom are merged into a single atom after \texttt{MergeNodes} operation as shown in Figure~\ref{fig:learning_example_dag_merged1_appendix}. We can visualize the the learnt atoms after merging in Figure~\ref{fig:learning_example_atoms1_appendix}

Finally we obtain the following fuzzy path logic formula from the learned DAG:

\begin{figure}[ht]
    \centering
    \includegraphics[width=0.8\columnwidth]{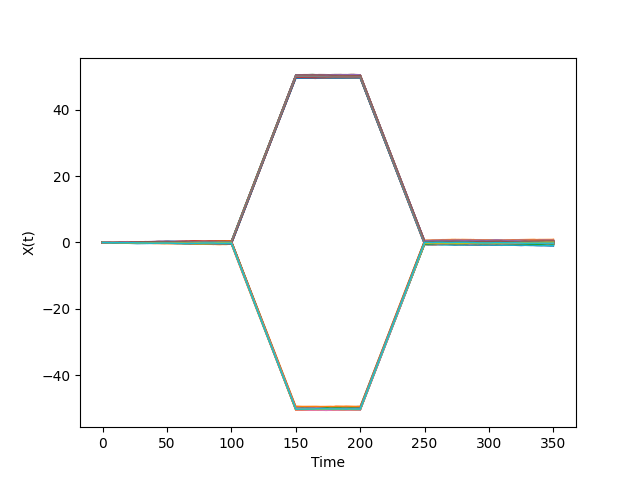}
    \caption{Example 3 trajectories.}
    \label{fig:learning_example_traj1_appendix}
\end{figure}

\begin{figure}[ht]
    \centering
    \begin{subfigure}[b]{0.45\columnwidth}
        \centering
        \includegraphics[width=\columnwidth]{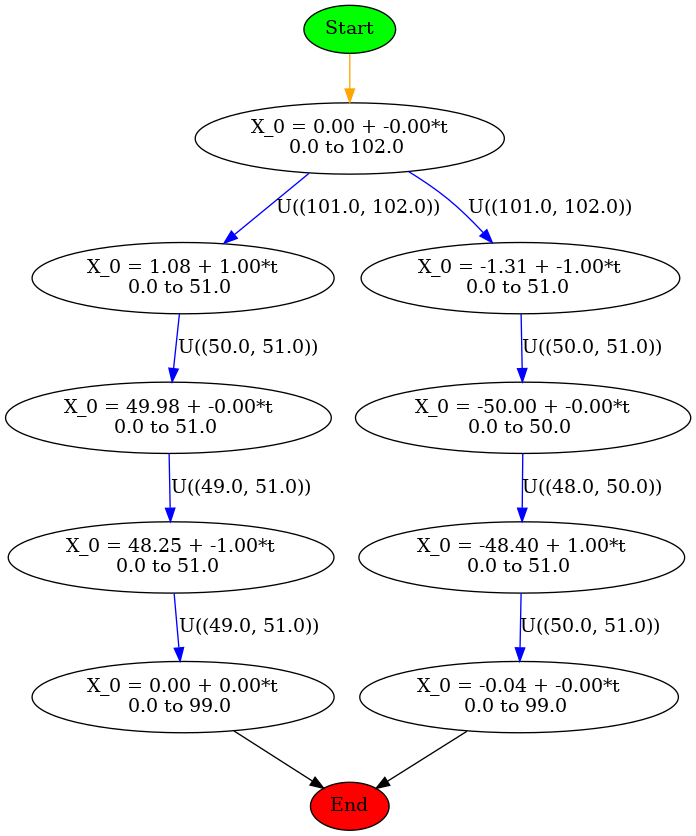}
        \caption{The learned DAG for Example 3 from the trajectories in Figure~\ref{fig:learning_example_traj1_appendix} before atom merging.}
        \label{fig:learning_example_dag1_appendix}
    \end{subfigure}
    \hfill
    \begin{subfigure}[b]{0.45\columnwidth}
        \centering
        \includegraphics[width=\columnwidth]{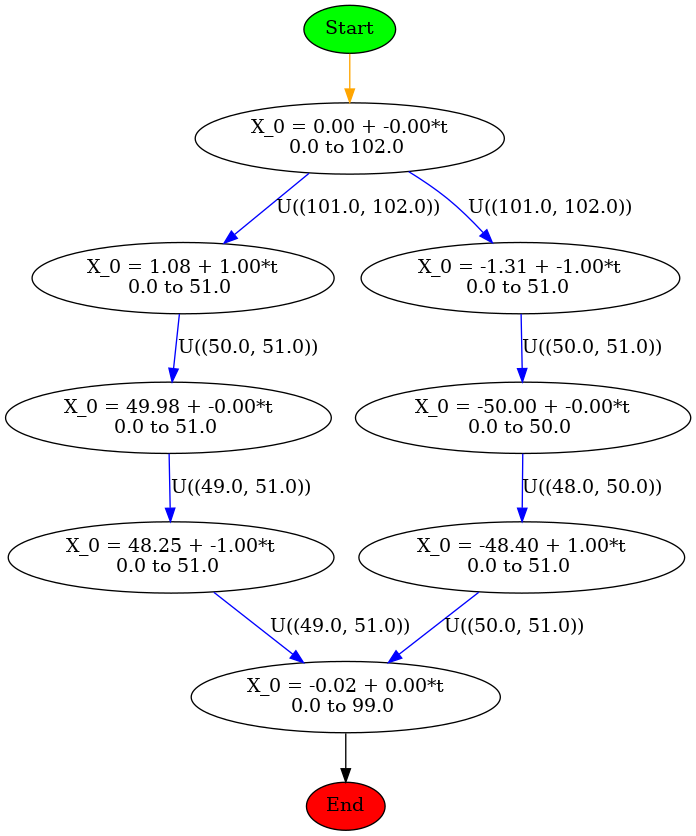}
        \caption{DAG after merging of similar atoms from the DAG in Figure~\ref{fig:learning_example_dag1_appendix}.}
        \label{fig:learning_example_dag_merged1_appendix}
    \end{subfigure}
    \caption{Learned DAG for Example 3.}
\end{figure}

\begin{figure}[ht]
    \centering
    \includegraphics[width=0.8\columnwidth]{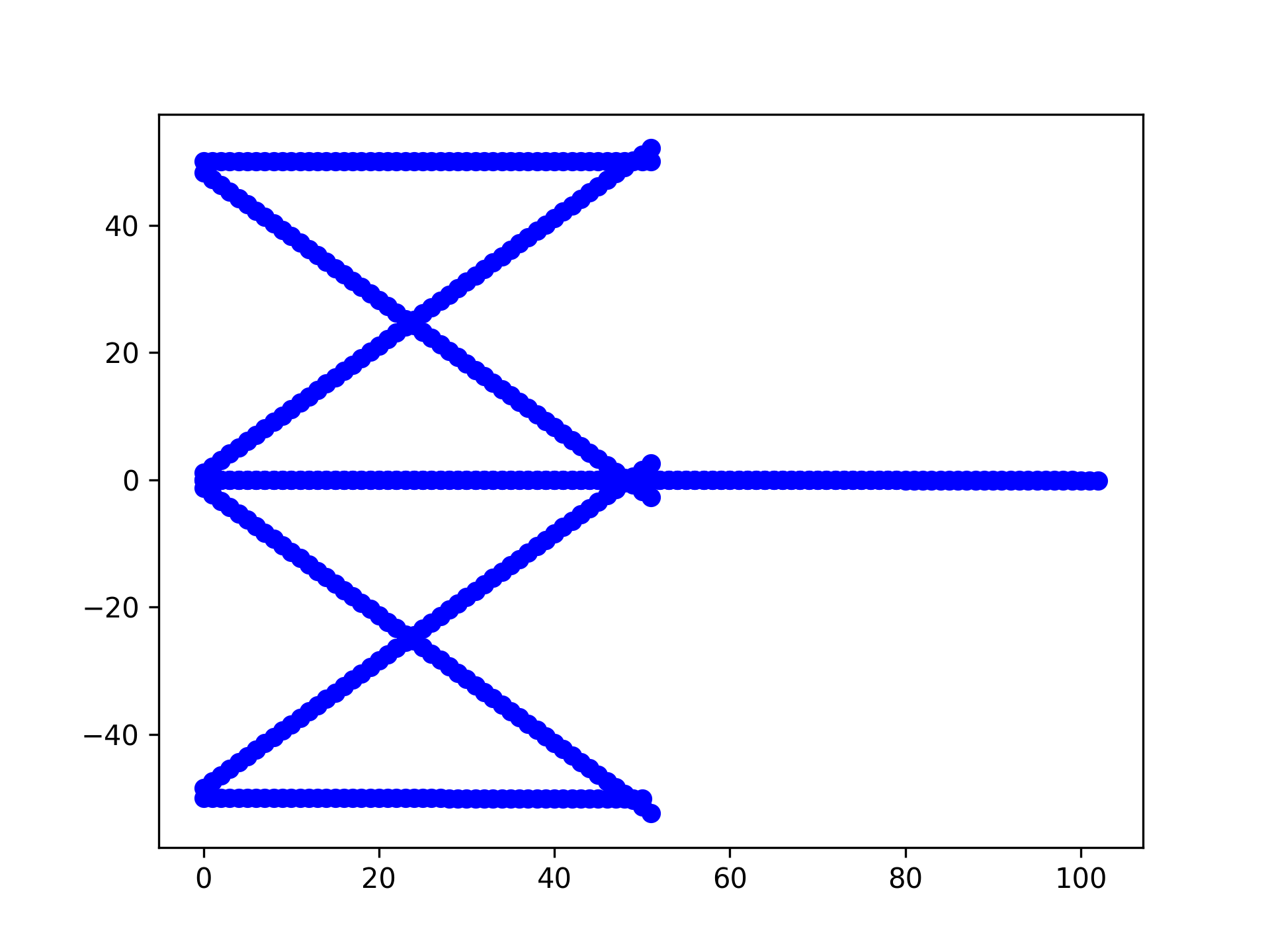}
    \caption{Learned atoms for Example 3 from the trajectories in Figure~\ref{fig:learning_example_traj1_appendix} after atom merging.}
    \label{fig:learning_example_atoms1_appendix}
\end{figure}

\paragraph{Example 4}
This is a real life scenario, where an agent has to move to a coffee machine avoiding an obstacle in between and then move to a goal location as shown in Figure~\ref{fig:coffee_trajectories}. 
Given the complexity of the trajectories, intially learnt DAG is complex as can be seen in Figure~\ref{fig:coffee_dag} and corresponding atoms as shown in Figure~\ref{fig:coffee_atoms}. After the merging of similar atom the DAG shown in Figure~\ref{fig:coffee_dag_simplified} and atoms shown in Figure~\ref{fig:coffee_atoms_simplified} are obtained.

\begin{figure}[ht]
    \centering
    \includegraphics[width=0.8\columnwidth]{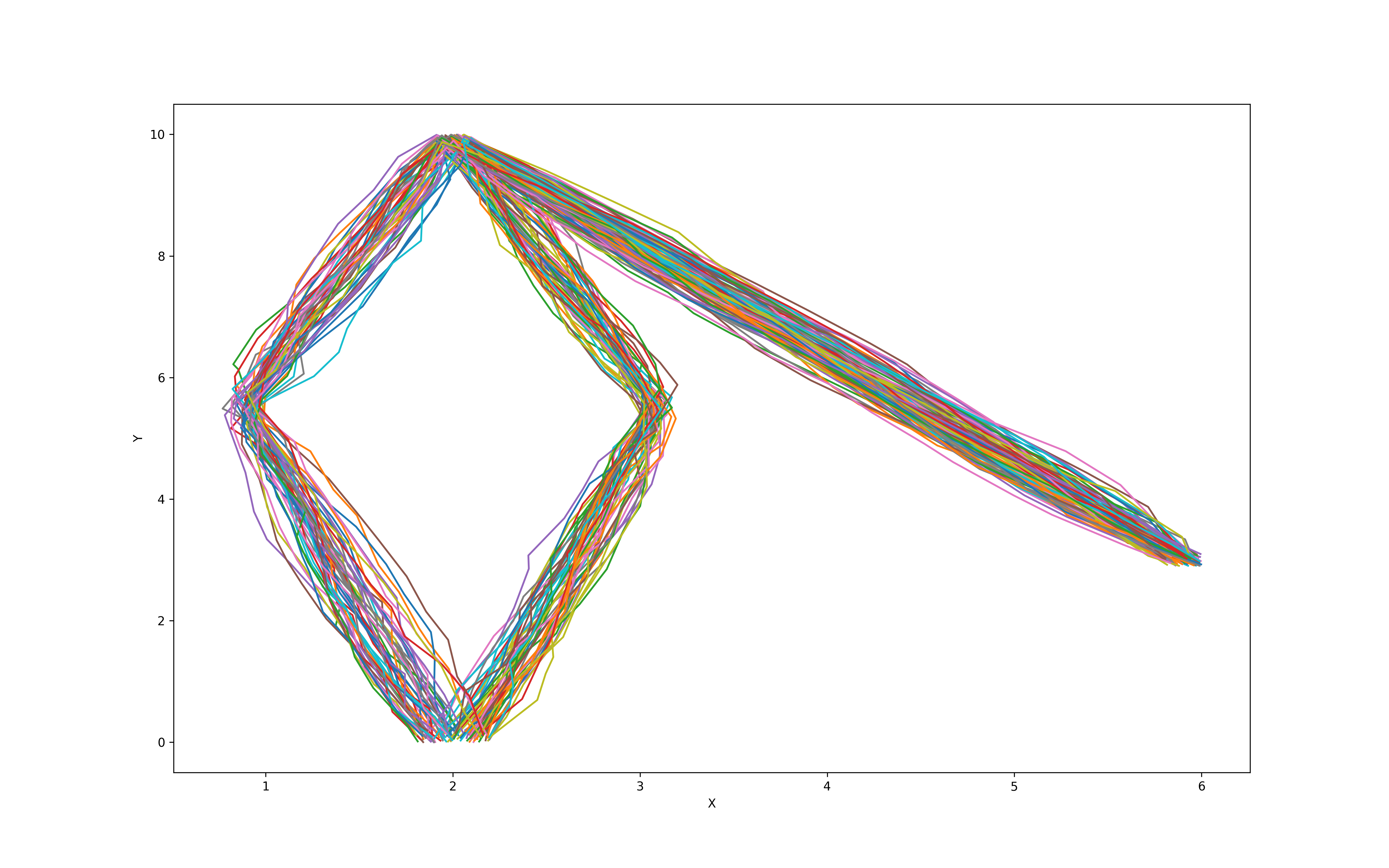}
    \caption{Example 4 trajectories where an agent moves to a coffee machine avoiding an obstacle and then moves to a goal location.}
    \label{fig:coffee_trajectories}
\end{figure}

\begin{figure}[ht]
    \centering
    \includegraphics[width=0.8\columnwidth]{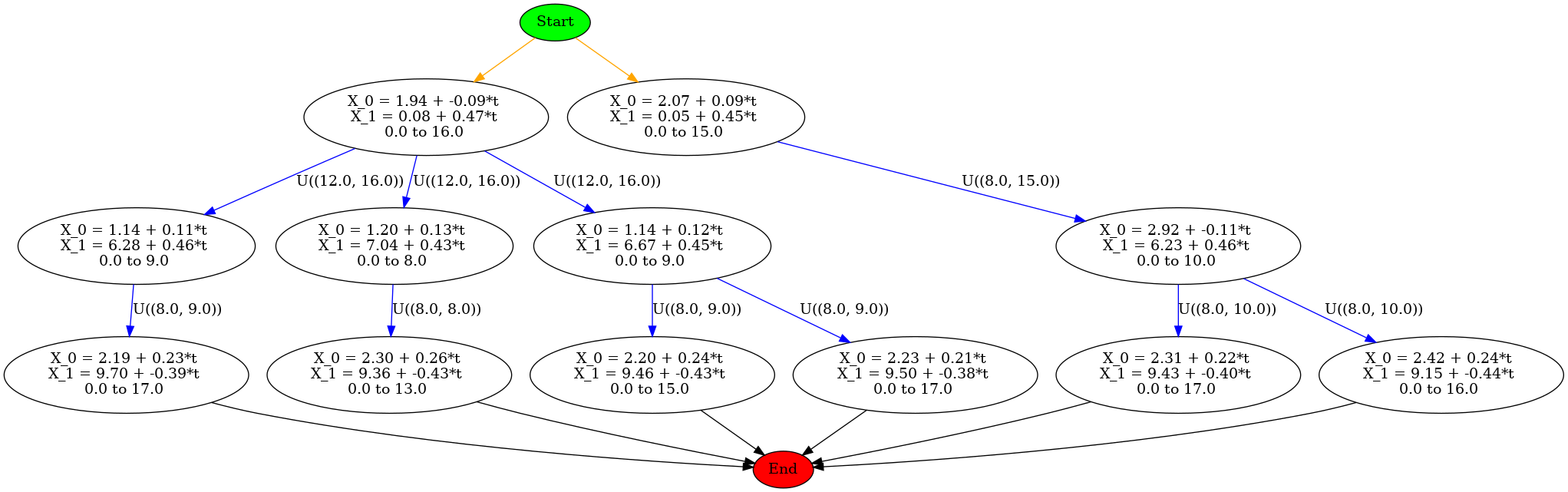}
    \caption{Learned DAG for Example 4}
    \label{fig:coffee_dag}
\end{figure}

\begin{figure}[ht]
    \centering
    \includegraphics[width=0.8\columnwidth]{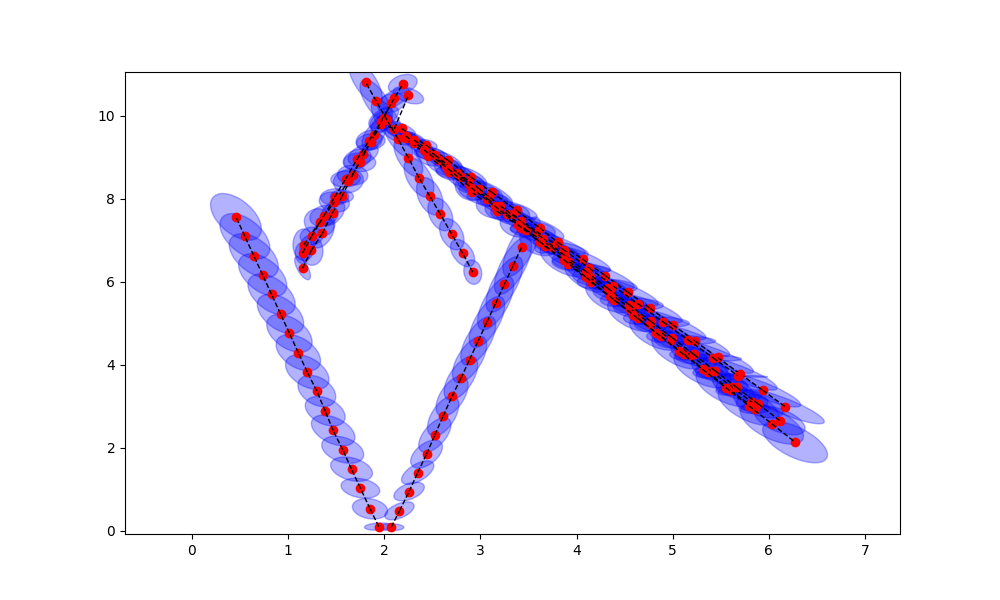}
    \caption{Learned atoms for Example 4 before atom merging. Note that there are several similar atoms corresponding to the final phase of the trajectory where the agent moves to the goal location.}
    \label{fig:coffee_atoms}
\end{figure}

\begin{figure}[H]
    \centering
    \includegraphics[width=0.5\columnwidth]{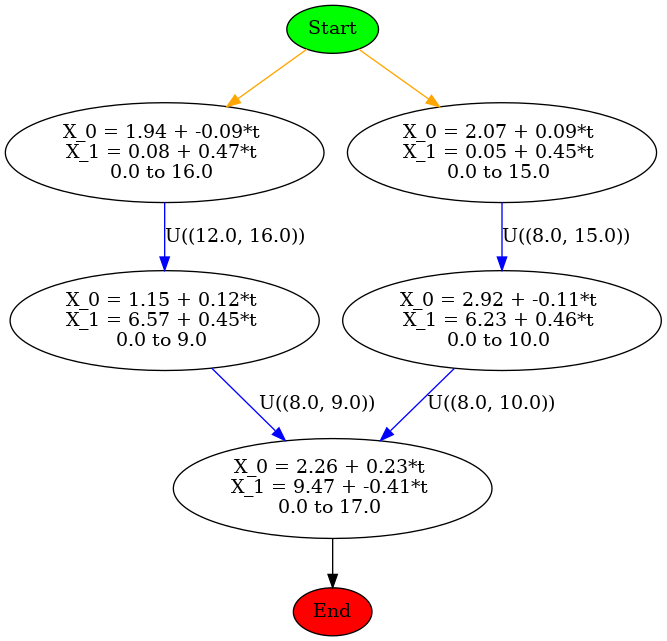}
    \caption{Learned DAG for Example 4 after simplification from the DAG in Figure~\ref{fig:coffee_dag}.}
    \label{fig:coffee_dag_simplified}
\end{figure}

\begin{figure}[ht]
    \centering
    \includegraphics[width=0.8\columnwidth]{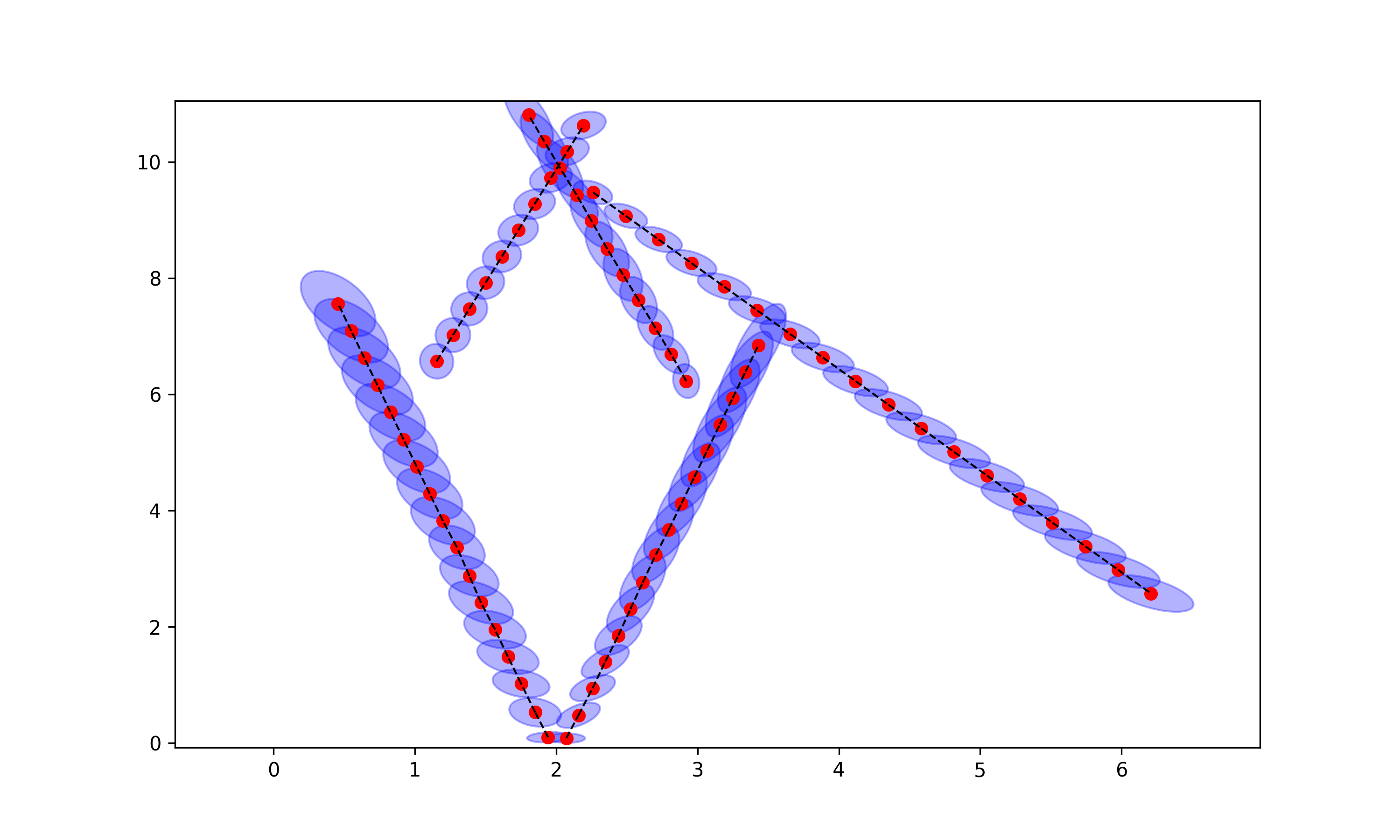}
    \caption{Learned atoms for Example 4 after simplification from the atoms in Figure~\ref{fig:coffee_atoms}.}
    \label{fig:coffee_atoms_simplified}
\end{figure}

\subsection{Parameters of the learning algorithm}

The algorithm has several parameters that should be chosen based on the characteristics of the data and the desired properties of the learned FPL. These hyperparameters include:
\begin{itemize}
    \item \textbf{Initial Subtrajectory Length}: The length of the initial sub-trajectories used for learning. A longer length may capture more complex behaviors but this value must be less than the shortest atom duration.
    \item \textbf{Minimum Trajectories for Processing}: The minimum number of trajectories required to process and learn a new atom. This helps ensure that the learned representations are robust and not overly sensitive to noise. If this value is too high then we may not have enough data towards the tail end of time horizons and hence FPL may not capture the full behavior.
    \item \textbf{Deviation Threshold}: The threshold for deviation when comparing trajectories to current atom. If the deviation exceeds this threshold, a new atom is learned. If the threshold is too low then noise may be read as deviation but if the value is to high then algorithm will be fuzzy around atom boundaries and suffer worse performance.
\end{itemize}

We are also using a hyperparameter corresponding to clustering process where if the standard deviation of trajectories is below a certain threshold then we do not perform clustering. This process can be avoided by using alternate clustering algorithms.

These hyperparameters can be adjusted based on the specific characteristics of the data and the desired properties of the learned FPL. 

We choose the following hyperparameters:
\begin{itemize}
    \item \textbf{For 1D Datasets}
    \begin{itemize}
        \item \textbf{Initial Subtrajectory Length}: Depends on the dataset
        \item \textbf{Minimum Trajectories for Processing}: 5
        \item \textbf{Deviation Threshold}: 4 standard deviations
    \end{itemize}
    \item \textbf{For 2D Datasets}
    \begin{itemize}
        \item \textbf{Initial Subtrajectory Length}: Depends on the dataset
        \item \textbf{Minimum Trajectories for Processing}: 5
        \item \textbf{Deviation Threshold}: 3 standard deviations
    \end{itemize}
\end{itemize}

\end{document}